\documentclass[11pt, dvipsnames]{ociamthesis}

\usepackage[a4paper,verbose]{geometry}
\usepackage{amsfonts}
\usepackage{amssymb}
\usepackage{amsthm}
\usepackage{amsmath}
\usepackage[margin=1cm]{caption}
\allowdisplaybreaks
\usepackage{tipa}
\usepackage{caption}
\usepackage[dvipsnames]{xcolor}
\usepackage[inline]{enumitem}
\setlist{itemsep=0em, topsep=0em, parsep=0em}
\setlist[enumerate]{label=(\alph*)}
\usepackage{etoolbox}
\usepackage{stmaryrd}
\usepackage{hyperref}
\hypersetup{
colorlinks=true,
linkcolor=[rgb]{0.0,0.27,0.13},
citecolor=[rgb]{0.0,0.27,0.13},
urlcolor=[rgb]{0.0,0.27,0.13}}
\usepackage{graphicx}
\graphicspath{{assets/}}
\usepackage{mathtools}
\usepackage{bussproofs}
\EnableBpAbbreviations{}
\usepackage[frozencache]{minted}
\usepackage{tikz}
\usepackage{float}
\usetikzlibrary{
matrix,
arrows,
shapes
}
\usepackage{tikz-cd}
\usepackage{CJKutf8}\usepackage[utf8]{inputenc}
\DeclareFontFamily{U}{min}{} \DeclareFontShape{U}{min}{m}{n}{<-> udmj30}{}
\newcommand{\hirayo}{\text{\usefont{U}{min}{m}{n}\symbol{'210}}}
\newcommand{\hirata}{\text{\usefont{U}{min}{m}{n}\symbol{'137}}}
\usepackage{macros}
\usepackage{diag}
\usepackage{eqproof}

\newcommand{\repl}{\mathrm{repl}}
\newcommand{\img}{\mathrm{img}}

\newcommand{\optic}[2]{\langle #1 \mid #2 \rangle}
\newcommand{\trv}{\operatorname{trv}}
\newcommand{\Optic}{\mathbf{Optic}}
\newcommand{\Sets}{\mathbf{Sets}}

\newcommand{\Ran}{\mathsf{Ran}}
\newcommand{\Prof}{\mathbf{Prof}}
\newcommand{\C}{\mathbf{C}}
\newcommand{\D}{\mathbf{D}}
\newcommand{\M}{\mathbf{M}}
\newcommand{\N}{\mathbf{N}}
\newcommand{\mact}{\underline{m}}
\newcommand{\nact}{\underline{n}}
\newcommand{\iact}{\underline{i}}

\newcommand\id{\mathrm{id}}

\newcommand\Nat{\mathrm{Nat}}

\newcommand\tonat{\Rightarrow}

\newcommand\V{{\cal{V}}}

\usepackage{drawings}
\degree{MSc in Mathematics and Foundations of Computer Science}
\degreedate{Trinity 2019}
\college{Hertford College}
\theoremstyle{plain}
\newtheorem{theorem}{Theorem}
\newtheorem{proposition}[theorem]{Proposition}

\newtheorem{lemma}[theorem]{Lemma}

\newtheorem{conjecture}[theorem]{Conjecture}
\newtheorem{corollary}[theorem]{Corollary}
\theoremstyle{definition}
\newtheorem{definition}[theorem]{Definition}

\theoremstyle{remark}
\newtheorem{remark}[theorem]{Remark}
\newtheorem{exampleth}[theorem]{Example}
\begingroup\makeatletter\@for\theoremstyle:=definition,remark,plain\do{\expandafter\g@addto@macro\csname th@\theoremstyle\endcsname{\addtolength\thm@preskip\parskip}}\endgroup
\numberwithin{theorem}{section}
\author{Mario Román}
\date{September 2, 2019}
\title{Profunctor optics and traversals}
\hypersetup{
 pdfauthor={Mario Román},
 pdftitle={Profunctor optics and traversals},
 pdfkeywords={},
 pdfsubject={},
 pdfcreator={Emacs 27.0.50 (Org mode 9.1.14)}, 
 pdflang={English}}
\begin{document}

\maketitle
\tableofcontents

\begin{abstract}
\emph{Optics} are bidirectional accessors of data structures.
They provide a powerful abstraction of many common data transformations. 
This abstraction is compositional thanks to a representation
in terms of profunctors endowed with an algebraic structure called
\emph{Tambara module} \cite{milewski17}.

There exists a general definition of optic \cite{boisseau17,riley18} in
terms of coends that, after some elementary application of the Yoneda
lemma, particularizes in each one of the basic optics.  \emph{Traversals}
used to be the exception; we show an elementary derivation of
traversals and discuss some other new derivations for optics.  We
relate our characterization of traversals to the previous ones showing
that the coalgebras of a comonad that represents and split into shape
and contents are traversable functors.

The representation of optics in terms of profunctors has many
different proofs in the literature; we discuss two ways of proving it,
generalizing both to the case of mixed optics for an arbitrary action.
Categories of optics can be seen as Eilenberg-Moore categories for a
monad described by Pastro and Street \cite{pastro08}. This gives us two different
approaches to composition between profunctor optics of different \emph{families}:
using distributive laws between the monads defining them, and using coproducts of
monads. The second one is the one implicitly used in Haskell
programming; but we show that a refinement of the notion of optic is
required in order to model it faithfully.

We provide experimental implementations of a library of optics in Haskell and
partial Agda formalizations of the profunctor representation theorem.
\end{abstract}

\begin{acknowledgements}
I am first grateful to my dissertation supervisor 
\textsf{Jeremy Gibbons} for generously taking time to make this dissertation
possible, for the advice, patience, ideas, pointers to the
literature, and all the suggestions that made it
readable.  His work both on optics \cite{boisseau18,pickering17} and traversable functors 
\cite{gibbons09} was many times, and sometimes in unexpected ways, an inspiration for this text.

My study of optics has been part of a joint project with a wonderful
group of people. I want to thank them all for their ideas,
encouragement and kindness.  Coordinating conversations on optics
over four different timezones was not easy, and they put a lot of
effort into making that happen. \textsf{Bartosz Milewski} crafted a
theory \cite{milewski17}, posed a nice problem, and then wholeheartedly
shared with us all his intuitions and insights.  \textsf{Derek Elkins}
put time, wisdom and patience into guiding me on how to transform
ideas into actual mathematics.
\textsf{Bryce Clarke} and \textsf{Emily Pillmore} were fantastic
colleagues during those days, and discussing both category theory and
the intricacies of Haskell with them was particularly stimulating.  I
learned the art of \emph{coend-fu} through the lucid teachings of
\textsf{Fosco Loregian} \cite{loregian15}, who also took the time to
clarify many of my doubts. 
\textsf{Daniel Cicala}, 
\textsf{Jules Hedges} and 
\textsf{Destiny Chen} organized the ACT School and made
this collaboration possible in the first place.  Thanks also go to 
\textsf{Guillaume Boisseau}, 
\textsf{Fatimah Ahmadi}, \textsf{Giovanni De Felice}, 
\textsf{Dan Marsden}, \textsf{Bob Coecke}, \textsf{David Reutter},
\textsf{David J. Myers}, \textsf{Brendan Fong}, \textsf{Jamie Vicary}, and 
\textsf{Carmen Constantin} for their engaging teachings on category theory and/or
conversations during the writing of this dissertation.  The text relies on
\textsf{Dan Marsden}'s macros for string diagrams.  I thank \textsf{Sam Staton} and
\textsf{Paul Blain Levy} for corrections and comments on the earlier versions of
this dissertation.

On a more personal note, friends and family in Granada have been
supporting me for years.
Besides them, I would like to thank
\textsf{Andrés},
\textsf{Daphne},
\textsf{Joel},
\textsf{Jean} and
\textsf{Lewis}, who
made enjoyable my day-to-day life during the MFoCS. 
Special thanks go to
\textsf{Elena} for all her affection and for making me recover the
passion for mathematics and category theory.
\textsf{A.} 
\textsf{P.} 
\textsf{A.} 
\textsf{E.} and her made the final 
days of writing this dissertation cozy. Grazie mille.
\end{acknowledgements}

\chapter{Introduction}
\label{sec:org099b9b5}
\section{Background and scope}
\label{sec:orge445c41}
In programming, \textbf{optics} are a compositional representation of data
accessors provided by libraries such as Kmett's \cite{kmett15}.  
Code needs to
deal with complex and nested data structures: \emph{records} with multiple
fields, \emph{union types} with multiple alternative contents, \emph{containers}
with lists of elements inside, and many other similar examples. In all
of these cases, we want to be able to focus on some internal part and
access it \emph{bidirectionally}, that is, we want to be able to read it, but
also to update it with new contents, propagating the changes to the
bigger data structure.  The most obvious pattern of this kind is a
\emph{lens}: a data accessor for the particular subfield of some data
structure.

\subsection{Lenses}
\label{sec:orgecfd106}

\begin{definition}
In a cartesian category \(\C\), a \textbf{lens} from a pair of objects \((s,t)\) with \emph{focus}
in a pair of objects \((a,b)\) is given by two morphisms: a \emph{view} \(\C(s , a)\), representing
the reading operation, and an \emph{update} \(\C(s \times b , t)\), representing the writing operation.
\[\mathbf{Lens}
\left( \begin{pmatrix} a \\ b \end{pmatrix}, \begin{pmatrix} s \\ t \end{pmatrix} \right) 
= \C(s , a) \times \C(s \times b , t).\]
\end{definition}

\begin{figure}[htbp]
\centering
\includegraphics[width=5cm]{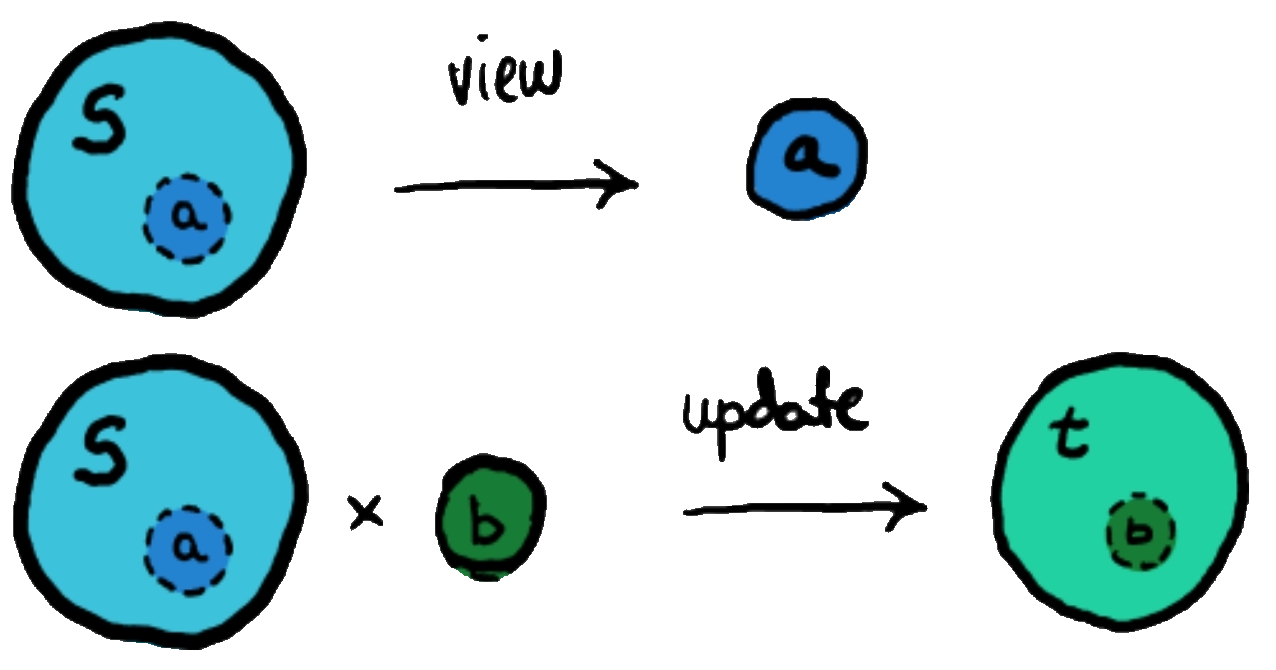}
\caption{\label{fig:org34ce678}
In \(\Sets\), a big data structure \(s\) contains a subfield \(a\) and the first function \((s \to a)\) provides a way of accessing it. The second function \((s \times b \to t)\) plugs something of type \(b\) in the place of \(a\), updating \(s\) and getting a new complex structure of type \(t\).}
\end{figure}

\begin{exampleth}
\label{org729dff4}
Suppose we have a data structure representing a \emph{postal address}, with the
\emph{ZIP code} being a subfield. We want to be able to view the ZIP code inside
an address and to modify it inside the bigger address.  This can be encoded
as an element of \(\mathbf{Lens}(( \mathbf{Postal} , \mathbf{Postal} ) , ( \mathbf{Zip} , \mathbf{Zip}))\), which is in turn
given by two functions with the following signature.
\[\begin{aligned}
\mathrm{viewZipCode} \colon & \mathbf{Postal} \to \mathbf{Zip} \\
\mathrm{updateZipCode} \colon & \mathbf{Postal} \times \mathbf{Zip} \to \mathbf{Postal}
\end{aligned}\]
\end{exampleth}

\begin{remark}
This definition leaves open the possibility of updating a subfield of type \(a\) in a 
bigger data structure \(s\) with an element of a different type \(b\), yielding
a structure of a new type \(t\).  The reader not interested in changing
types can just consider optics from \((s,s)\) to \((a,a)\); these are sometimes called
\emph{monomorphic lenses} (for instance, in \cite{hedges18l}) or \emph{simple lenses} in \cite{kmett15}.
\end{remark}

Lenses were first described in categorical terms in Oles' thesis
\cite{oles82}. They are considered as data accessors in
\cite{pickering17}. They have been also used in \emph{compositional game
theory}, as described in \cite{ghani18}, and in \emph{supervised learning}, as
described in \cite{fong19}.

\subsection{Prisms}
\label{sec:org6c6d74c}
We can consider, however, other patterns for data accessing:
\emph{lenses are optics, but not all} \emph{optics are lenses!} \emph{Prisms} give
a second data accessing pattern whose focusing deals with alternatives.

\begin{definition}
In a cocartesian category \(\C\), a \textbf{prism} from a pair of sets \((s,t)\) with
focus in \((a,b)\) is given by two functions: a \emph{match} \(s \to t + a\) representing
pattern-matching on \(a\), and a \emph{build} \(b \to t\) representing the construction
of the data structure from one of the alternatives.
\[\mathbf{Prism}
\left( \begin{pmatrix} a \\ b \end{pmatrix}, \begin{pmatrix} s \\ t \end{pmatrix} \right) 
= \C(s , t + a) \times \C(b , t).\]
\end{definition}

\begin{figure}[htbp]
\centering
\includegraphics[width=4.8cm]{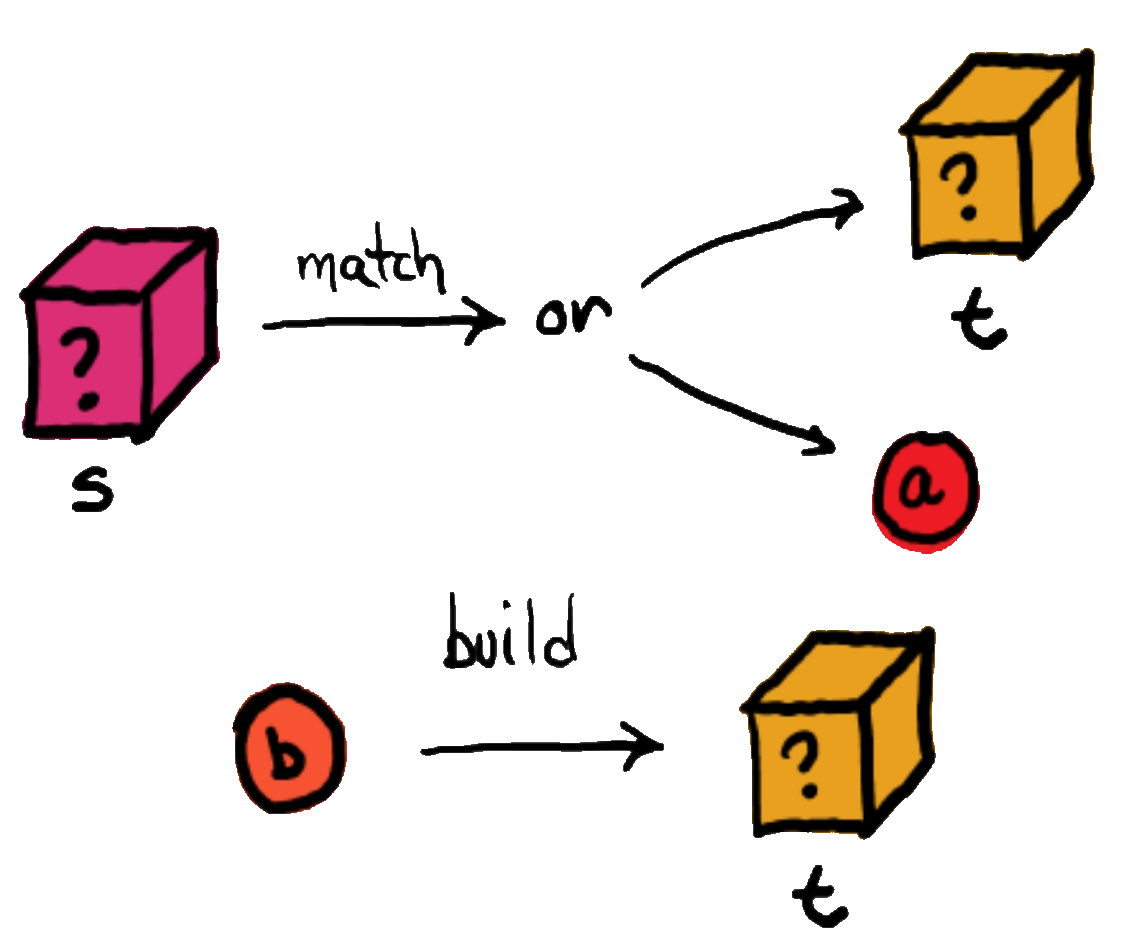}
\caption{\label{fig:org820f944}
An abstract data structure \(s\) could have different internal structures. In particular, we can try to see if it is of type \(a\) and, in case of failure, return something of type \(t\).  The build function takes some concrete data structure \(b\) to some more abstract structure \(t\).}
\end{figure}

\begin{exampleth}
\label{org47c3ee8}
Suppose we have a string we want to parse as an \emph{address}. That 
address could be in particular a postal address like in Example \ref{org729dff4}
and in that case we want to return a specific type capturing this.
That is, we can build an address from a postal address; and we can try to fit an address into
a postal address.  This can be encoded as an element of
\(\mathbf{Prism}((\mathbf{Postal} ,\mathbf{Postal} ) , (\mathbf{String} , \mathbf{Address}))\), which is, in turn,
given by two functions with the following signature.
\[\begin{aligned}
\mathrm{matchAddress} \colon & \mathbf{String} \to \mathbf{Address} + \mathbf{Postal} \\
\mathrm{buildAddress} \colon & \mathbf{Postal} \to \mathbf{Address}
\end{aligned}\]
\end{exampleth}

\begin{remark}
Prisms are dual to lenses.  A prism is precisely a lens in the
opposite category.  As we will see later, they are to the coproduct
what lenses are to the product.
\end{remark}

\subsection{Traversals}
\label{sec:orgc2fc4c8}
We can go further: optics do not necessarily need to deal with a single
focus. Traversals are optics that access a list of foci at the same time.

\begin{definition}
In a cartesian closed category with natural number-indexed limits and colimits \(\C\), a \textbf{traversal}
from a pair of sets \((s,t)\) with focus in \((a,b)\) is given by a single function: an \emph{extract}
\(\C(s , \sum\nolimits_{n \in \mathbb{N}} a^n \times (b^n \to t))\) that represents extracting a list of some length
and a function that takes a list of the same size to create a new structure.
\[\mathbf{Traversal}
\left( \begin{pmatrix} a \\ b \end{pmatrix}, \begin{pmatrix} s \\ t \end{pmatrix} \right) = 
\C\left(s , \sum\nolimits_{n \in \mathbb{N}} a^n \times (b^n \to t)\right).\]
\end{definition}

\begin{figure}[htbp]
\centering
\includegraphics[width=8cm]{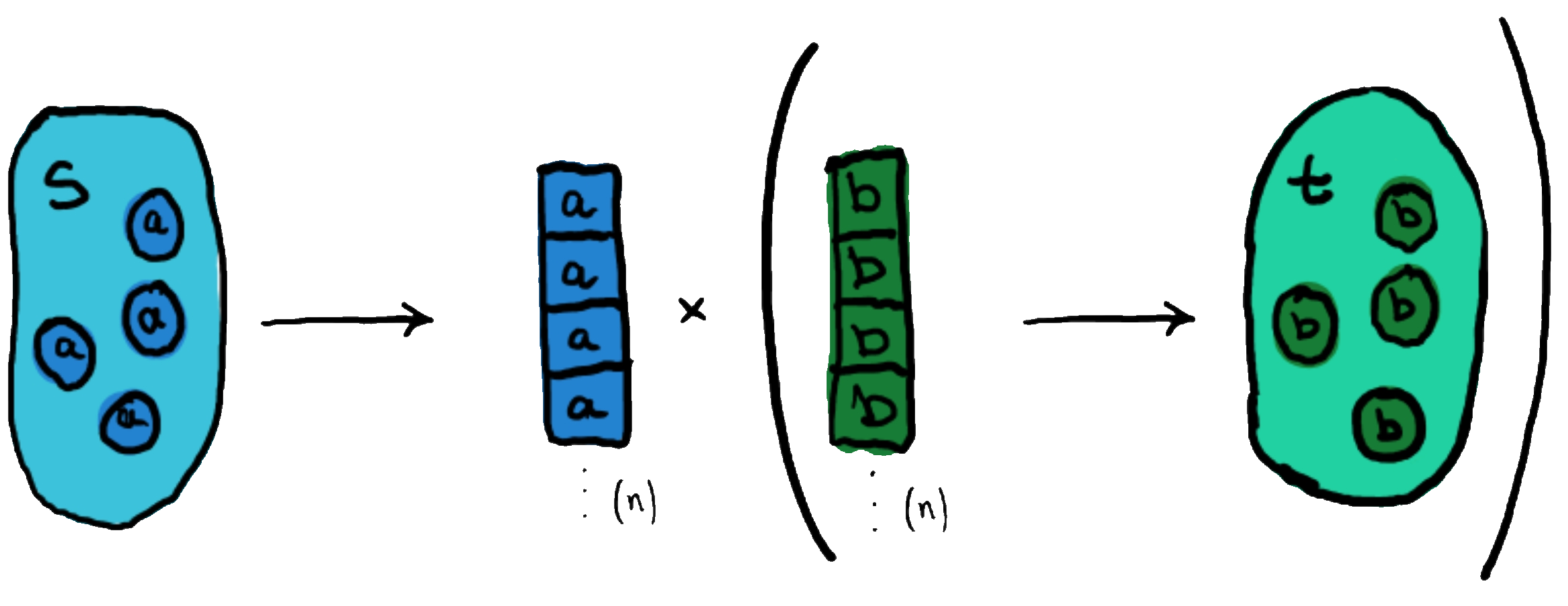}
\caption{\label{fig:orgd40965a}
From a complex data structure \(s\) we can extract (1) some list of values of type \(a\) of length \(n\) for some \(n \in \mathbb{N}\), and (2) a way of updating this structure with a new list of length \(n\).}
\end{figure}

\begin{exampleth}
\label{orgadaa07e} Suppose we have now a \emph{mailing list}
containing multiple email addresses associated to other data, such as
names or subscription options.  An accessor for the email addresses
is an element of
\(\mathbf{Traversal}(( \mathbf{MailList}, \mathbf{MailList} ) , ( \mathbf{Email} , \mathbf{Email}))\), which is given
by a single function with the following signature.
\[\begin{aligned}
\mathrm{extract} \colon & \mathbf{MailList} \to \sum_{n \in \mathbb{N}} \mathbf{Email}^n \times \left(\mathbf{Email}^n \to \mathbf{MailList}\right). \\
\end{aligned}\]
\end{exampleth}

\begin{remark}
One could think that the traversal can be rewritten in the same style as a lens:
that is, as two functions \(s \to a^n\) and \(s \times b^n \to t\). The fact that the two \(n \in \mathbb{N}\) need
to be the same prevents this split.
This is also what makes a traversal fundamentally
different from a lens focusing on a list or tuple type. For the traversal, the length of the
output can vary, but we need a list of the same length to be able to update.
\end{remark}

\subsection{The problem of modularity}
\label{sec:orge1a42d6}
A problem arises when accessing compound data structures.  The focus of
an optic can be itself a complex data structure in which we can make use
of an optic again.  We would expect optics to behave \emph{compositionally}, in
the sense that it should be straightforward to give a description of the
optic that, from the bigger data structure, applies two optics to focus
on the innermost subfield.

\begin{exampleth}
More concretely, consider the lens in Example \ref{org729dff4}
and the prism in Example \ref{org47c3ee8}. We would like to have a
composite optic that tries to access the Zip code from any
arbitrary string.  This problem can be formulated and solved for any \emph{lens} and
\emph{prism}, but there is no obvious unified way to solve it for any
pair of optics. Moreover, note that the resulting accessor
need be neither a lens nor a prism.

Explicitly, in a bicartesian closed category, assume we have a prism 
from \((s,t)\) to \((a,b)\) given by \(m \in \C(s, t + a)\)
and \(l \in \C(b,t)\); and a lens \((a,b)\) to \((x,y)\)
given by \(v \in \C(a,x)\) and \(u \in \C(a \times y , b)\).  We want to
compose them into a combined optic from \((s,t)\) to \((x,y)\).  A suitable composition
is the following, although it is tedious and not particularly illuminating to write it down.
Here, we write \(\Lambda\) for the product-exponential adjunction.
\[
\mathbf{m} \circ [ \mathrm{id}_t , \mathbf{v} \times \Lambda(\mathbf{b} \circ \mathbf{u})]
\in 
\C(s , t + x \times (y \to t)).
\]
Note that \(\Lambda(l \circ u) \in \mathbf{C}(a, y \to t)\), and the 
rest is just a combination of morphisms that uses the
product and the coproduct.
This naive approach to composition is not practical: we would need to write down a composition
for every pair of optics, and its output could be each time a completely new
kind of optic.
We need a general way of solving the problem of composition of optics.
\end{exampleth}

\subsection{Solving modularity}
\label{sec:orgf2413f6}
The problem has been solved in programming libraries such as
\cite{purescriptlens} using \textbf{profunctor optics}.
Perhaps surprisingly, some optics turn out to be equivalent to parametric
functions over certain families of profunctors.  For instance, we can prove
an isomorphism between lenses
\(\C(s , a) \times \C(s \times b , t)\) and families of morphisms
\(p(a,b) \to p(s,t)\) for profunctors \(p\) of a certain class that we will describe later.
This provides a solution for the composition of optics: composing an
optic from \((x,y)\) to \((a,b)\) with an optic from \((a,b)\) to \((s,t)\) becomes
function composition of the form \(p(x,y) \to p(a,b) \to p(s,t)\).

In summary, the profunctor representation of optics transforms all the
complicated cases of composition of optics into \emph{just} \emph{ordinary}
\emph{function} \emph{composition}.  This technique is documented originally in \cite{pickering17}.
We will justify in the following sections why something like this works
in terms of category theory and look into particular cases.

\section{Outline}
\label{sec:org13f314d}

\begin{itemize}
\item \textbf{Chapter 2} introduces main results that are not specific to the
theory of optics.  It starts with a brief summary of the (co)end calculus
as in \cite{loregian15}, underlining the importance of the Yoneda lemma and of Kan extensions.
It describes the formal theory of monoids capturing also the multiple uses of 
monads and monoidal categories that we will encounter in this text.  We take a moment to
describe replete subcategories and how they provide a well-behaved notion of 2-categorical
image of a functor.  Finally, we describe the bicategory of profunctors,
promonads and Kleisli categories.

\item \textbf{Chapter 3} provides the definition of optic and shows that it
actually captures the examples we were interested in.  Apart from
these main ones, it uses the definition to provide new examples of
optic.

\item \textbf{Chapter 4} is centered around traversals, a particular kind of optic
for which we give two different derivations. We relate both describing
traversals as coalgebras for a comonad that represents a split into
shape and contents.  We will show how they are related to flat combinatorial
species and how to use this to construct yet another new optic.

\item \textbf{Chapter 5} introduces and proves the main result of the theory of
optics: the \emph{profunctor representation theorem} (Theorem
\ref{org9a4afc7}), that links every optic to a
profunctor optic in an uniform way.  We also describe optics both
as Kleisli categories for a monad and Kleisli objects for a promonad.

\item \textbf{Chapter 6} discusses how function composition of profunctor optics
works.  The first part is motivated by how monads can be composed
when a suitable distributive law between them exists. The second
part is motivated by how Haskell joins constraints on
parametric functions to compose profunctor optics.  These two \emph{compositions}
do not necessarily give the same results, but we can reconcile
them considering some subclass of well-behaved optics that we call \emph{clear optics}.

\item \textbf{Chapter 7} concerns applications.  In particular, we have built a
minimalistic optics library in Haskell that directly implements the
profunctor representation theorem.  We also formalize the main optic
derivations in the Agda proof assistant, showing that the approach
based on Yoneda is particularly suited for formal verification.
\end{itemize}

\section{Contributions}
\label{sec:org4af46ab}
Parts of this dissertation started as a joint project; this section
lists the original contributions of the author of this dissertation
and gives acknowledgement for the rest.

\begin{itemize}
\item \textbf{Chapter 2} comes with no claim of originality apart from the
presentation of the ideas.  The propositions and proofs can be
found in the literature and we put special care into acknowledge
them.  \textsf{Fosco Loregian} helped me understand the coproduct
of monoidal actions.

\item In \textbf{Chapter 3}, I slightly generalize Mitchell Riley's construction
of the category of optics \cite{riley18} in a direction already
proposed there.  I derive traversals as the optic for power series
functors, solving a problem posed by Milewski \cite{milewski17}; \textsf{Bartosz Milewski}
suggested to me that the traversal was possibly an optic for an
action involving multiple products or exponentials.  In \S 3.4, I
define and derive concrete forms of \emph{Kaleidoscopes} and \emph{Algebraic}
\emph{lenses}, optics not present on the literature to the best of my knowledge; \textsf{Emily Pillmore}
helped me collect definitions of optics from Haskell libraries. \textsf{Guillaume Boisseau}
let me know about the \emph{achromatic lens}. I prove
that the \emph{generalized lens} is a mixed optic; this notion was introduced
to me by \textsf{David Jaz Myers}.  I present a mechanism for getting 
optics for (co)free, which is a specialization of the coalgebraic 
optics of \cite{riley18}, both requiring more hypothesis and getting 
a stronger result.

\item In \textbf{Chapter 4}, I prove in \S 4.1 that traversals are coalgebras for
a \emph{shape-contents} comonad that I define there; a related result using
parameterized comonads can be found in \cite{jaskelioff15}.  In Lemma
\ref{org8b21ea5}, I prove
that the natural family of transformations defining a traversal is
given by any coalgebra for the \emph{shape-contents} functor, trying to simplify
the approach in \cite{jaskelioff15}. Using there that
ends distribute over discrete colimits to complete the derivation was 
suggested to me by \textsf{Fosco Loregian} and \textsf{Derek Elkins}.
I connect in this way the derivation using traversables that could be
found in Riley's \cite{riley18} to the derivation using power series 
functors. I describe the \emph{unsorted traversal} and derive a concrete form
for it.  In the related Appendix \ref{appendixtraversables}, I construct cofree traversables.

\item \textbf{Chapter 5} is mostly an expository chapter.  I make explicit the
proof that can be found in \cite{pastro08} of the profunctor
representation theorem, extending it to mixed optics over arbitrary
actions; a proof of the same result with a different technique can
be found in \cite{riley18} and \cite{boisseau18}.

\item In \textbf{Chapter 6}, I study distributive laws for Pastro-Street comonads
and propose a composition of optics based on Kleisli categories. I
construct \emph{glasses} with this technique and I show that \emph{affine traversals} can be
also obtained with it.  In \S \ref{orga66504e}, I study
how composition of optics works in Haskell and provide a categorical
description in terms of coproduct actions.  The fact that a naive
composition of Tambara modules does not work as one would wish in
this case was pointed to me by \textsf{Bryce Clarke}.
I discuss the need for
a restricted definition of optic in order to recover the lattice of
optics and I propose a definition of \emph{clear optics}, addressing this
problem.  I prove that some
of the common optics are clear and that the composition of profunctor \emph{lenses}
and \emph{prisms} is precisely an \emph{affine traversal} in this setting.

\item In \textbf{Chapter 7}, I present a library of optics in Haskell, a
formalization of some (co)end derivations in Agda, and an
example usage of the \emph{kaleidoscope}.
\end{itemize}

\chapter{Preliminaries}
\label{sec:org001709f}
\section{(co)End calculus}
\label{sec:orga0bcef0}
Our basic tool will be (co)end calculus as described
by Loregian \cite{loregian15}.  For completeness, we replicate here the main
definitions and results we are going to use.

\subsection{(co)Ends}
\label{sec:orgb527d9c}
The intuition is that \emph{ends} should be a sort of universal quantifier
over the objects of a category plus some \emph{naturality} conditions; whereas
\emph{coends} can be thought of as their existential version. In fact, when
encoding optics in a programming language that provides parametricity and
existential types, these are used in place of coends, expecting that the syntax
of the language will enforce their naturality conditions.
(co)Ends provide a rich calculus based on the Yoneda lemma that we
will use throughout this text, being the basic building
block of many of our proofs.  A description on how to turn coends into
a calculus that can be used in theorem provers is described in \cite{caccamo01}.
\begin{definition}
The \textbf{end} of a functor \(p \colon \C^{op} \times \C \to \D\) is the universal
object \(\int\nolimits_{c \in \C} p(c,c)\) endowed with morphisms \(\pi_{a} \colon \left(  \int\nolimits_{c \in \C} p(c,c) \right) \to p(a,a)\) 
for every \(a \in \C\) such that, for any morphism \(f \colon a \to b\) in \(\C\), they
satisfy \(p(\id,f) \circ \pi_a = p(f,\id) \circ \pi_b\).
It is universal in the sense that any other object \(d\) endowed with
morphisms \(j_a \colon d \to p(a,a)\) satisfying the same condition factors uniquely through it.
\[\begin{tikzcd}
& d\ar[bend left]{dddr}{j_b} \ar[bend right,swap]{dddl}{j_{a}} \ar[dashed]{dd}{\exists!} & \\
& & \\
& \int_c p(c,c) \drar{\pi_b} \dlar[swap]{\pi_a} & \\
p(a,a) \drar[swap]{p(\id,f)} && p(b,b) \dlar{p(f,\id)} \\
& p(a,b)
\end{tikzcd}\]
In other words, the end is the equalizer of the action of morphisms
on both arguments of the functor.
\[
\int_{c \in \C} p(c,c) \cong \mathrm{eq} \left( \begin{tikzcd}
\prod_{c \in \C} p(c,c) \rar[yshift=0.5ex]\rar[yshift=-0.5ex] & \prod_{f \colon a \to b} p(a,b) 
\end{tikzcd}
\right)
\]
\end{definition}
\begin{definition}
Dually, the \textbf{coend} of a functor \(p \colon \C^{op} \times \C \to \D\) is the universal object
\(\int^{c \in \C} p(c,c)\) endowed with morphisms \(i_a \colon p(a,a) \to \left( \int\nolimits^{c \in \C} p(c,c) \right)\)
for every \(a \in \C\) such that, for any morphism \(f \colon b \to a\) in \(\C\),
they satisfy \(i_b \circ p(f, \id) = i_a \circ p(\id, f)\). It is universal in
the sense that any other object \(d\) endowed with morphisms \(j_a \colon p(a,a) \to d\)
satisfying the same condition factors uniquely through it.
\[\begin{tikzcd}
& p(a,b) \dlar[swap]{p(\id, f)}\drar{p(f,\id)} & \\
p(a,a) \ar[bend right, swap]{dddr}{j_a} \drar[swap]{i_a} && p(b,b) \dlar{i_b} \ar[bend left]{dddl}{j_b} \\
& \int^c p(c,c) \ar[dashed]{dd}{\exists!} & \\
&& \\
& d &
\end{tikzcd}\]
In other words, the coend is the coequalizer of the action on
morphisms on both arguments of the functor.
\[
\int^{c \in \C} p(c,c) \cong \mathrm{coeq}\left( \begin{tikzcd}
\bigsqcup_{f \colon b \to a} p(a,b) \rar[yshift=-0.5ex, swap] \rar[yshift=0.5ex] &
\bigsqcup_{x \in \C} p(x,x) 
\end{tikzcd}\right)
\]
\end{definition}
\begin{remark}
\label{orgd5ad7d1}
(Co)ends are particular cases of (co)limits. As such, they are unique
up to isomorphism when they exist. (Co)continuous functors preserve
them and, in particular, the co(ntra)variant hom-functor commutes with them.
For any \(p \colon \C^{op} \times \C \to \D\) and every \(d \in \D\), there
exist canonical isomorphisms with the following signatures.
\[
\D \left( \int^{c \in \C} p(c,c) , d \right) \cong
\int_{c \in \C} \D(p(c,c), d),\qquad
\D \left( d, \int_{c \in \C} p(c,c) \right) \cong
\int_{c \in \C} \D(d, p(c,c)).
\]
Note also that all (co)ends exist in (co)complete categories.
\end{remark}
\begin{proposition}[\cite{loregian15}, Remark 1.4]
(co)Ends are functorial. We can define a functor \([\C^{op} \times \C , \D] \to \D\)
that sends a profunctor \(p \colon \C^{op} \times \C \to \D\) to its end \(\int_c p(c,c) \in \D\).
\end{proposition}
\begin{proof}
Let \(\eta \colon p \tonat q\) be a natural transformation between two profunctors.
By the universal property of the ends, we have an induced
\(\eta_{\ast} \colon\int\nolimits_c p(c,c) \to \int\nolimits_c q(c,c)\) constructed as the unique morphism making the
following diagram commute for any \(x,x' \in \C\) and any \(f \colon x \to x'\).
\[\begin{tikzcd}
& \int\nolimits_c q(c,c) \ar{rr}{\pi_{x'}}\ar{dd}[near end]{\pi_x} && q(x',x') \ar{dd}{q(f,\id)} \\
\int\nolimits_c p(c,c) \ar{dd}[swap]{\pi_x} \ar[dashed]{ur}{\exists! \eta_{\ast}} \ar{rr}[near end]{\pi_{x'}} && p(x',x') \ar{dd}[near start]{p(f,\id)} \ar{ur}{\eta_{x',x'}} & \\
& q(x,x) \ar{rr}[near start]{q(\id,f)} && q(x,x') \\
p(x,x) \ar{ur}{\eta_{x,x}} \ar{rr}[swap]{p(\id,f)} && p(x,x') \ar{ur}[swap]{\eta_{x,x'}} &
\end{tikzcd}\]
We can check that the identity makes the diagram commute for the identity
natural transformation, getting \(\mathrm{Id}_{\ast} = \id\); and that composition is
preserved because it makes the composite diagram commute, getting \((\eta \circ \sigma)_{\ast} = \eta_{\ast} \circ \sigma_{\ast}\)
for any pair of natural transformations \(\eta \colon p \tonat q\) and \(\sigma \colon r \tonat p\).
\end{proof}

\subsection{Fubini rule}
\label{sec:orgc47fe17}
Writing ends and coends as integrals suggests the following sort
of Fubini rule. Note however that, 
contrary to what happens with the theorem of classical analysis, this rule can be \emph{always}
applied.
\begin{lemma}[Fubini rule for ends, {{\cite[\S1, Exercise 10]{loregian15}}}]
Let \(p \colon \C^{op} \times \C \times \D^{op} \times \D \to \mathbf{E}\) be a functor. We can consider
the following three ends, that are different in principle.

\begin{itemize}
\item Identify \(\C^{op} \times \C \times \D^{op} \times \D \to \mathbf{E}\) with \(\C^{op} \times \C \to (\D^{op} \times \D \to \mathbf{E})\).
Take the end over it to obtain \(\int\nolimits_c p(c,c,-,-) \colon \D^{op} \times \D \to \mathbf{E}\), and then
the end over the resulting functor to get \(\int\nolimits_d\int\nolimits_c p(c,c,d,d)\).

\item Identify \(\C^{op} \times \C \times \D^{op} \times \D \to \mathbf{E}\) with \(\D^{op} \times \D \to (\C^{op} \times \C \to \mathbf{E})\).
Take the end over it to obtain \(\int\nolimits_d p(-,-,d,d) \colon \C^{op} \times \C \to \mathbf{E}\), and then
the end over the resulting functor to get \(\int\nolimits_c\int\nolimits_d p(c,c,d,d)\).

\item Identify \(\C^{op} \times \C \times \D^{op} \times \D \to \mathbf{E}\) with \((\C \times \D)^{op} \times (\C \times \D) \to \mathbf{E}\).
Take the end over this product category to get \(\int\nolimits_{(c,d)} p(c,c,d,d)\).
\end{itemize}

The \textbf{Fubini rule} for coends states the following isomorphisms.
\[
\int_{(c,d) \in \C \times \D} p(c,c,d,d) \cong \int_{c \in \C} \int_{d \in \D} p(c,c,d,d) \cong \int_{d \in \D} \int_{c \in \C} p(c,c,d,d).
\]
\end{lemma}
\begin{proof}
We will just prove that \(\int\nolimits_{(c,d) \in \C \times \D} p(c,c,d,d) \cong \int\nolimits_{c \in \C} \int\nolimits_{d \in \D} p(c,c,d,d)\), the other
isomorphism is similar.  First, we construct the morphism from the left hand side
to the right hand side.  We start by constructing a family of morphisms
\(\sigma_x \colon \int\nolimits_{(c,d) \in \C \times \D} p(c,c,d,d) \to \int\nolimits_{d \in \D} p(x,x,d,d)\) as the unique ones making the
following diagram commute for any \(g \colon y \to y'\). Note that the external square commutes
as a particular case of dinaturality of the coend over the product.
\[\begin{tikzcd}
& \int_{c,d} p(c,c,d,d) \dlar[swap]{\pi_{x,y}} \drar{\pi_{x,y'}} \dar[dashed]{\exists! \sigma_x} & \\
p(x,x,y,y) \drar[swap]{p(-,-,-,g)} & \lar{\pi_{y}} \int_d p(x,x,d,d) \rar[swap]{\pi_{y'}} & p(x,x,y',y') \dlar{p(-,-,g,-)} \\
& p(x,x,y,y') &
\end{tikzcd}\]
We need to prove that the family \(\sigma_x\) is actually dinatural in \(x\). This comes also
as a particular case of the dinaturality of the end over the product, as the following
diagram shows.
\[\begin{tikzcd}
& \int_{c,d} p(c,c,d,d) \drar{\sigma_{x'}} \dlar[swap]{\sigma_x} \ar[bend left]{ddl}{\pi_{x,y}} \ar[bend right]{ddr}[swap]{\pi_{x',y'}} & \\
\int_d p(x,x,d,d) \drar[near start]{f_{\ast}} \dar[swap]{\pi_x} && \int_d p(x',x',d,d) \dlar[swap]{f_{\ast}} \dar{\pi_{x'}}\\
p(x,x,y,y) \drar[swap]{p(f,-,-,-)} & \int_d p(x',x,d,d) \dar{\pi_y} & p(x',x,y,y) \dlar{p(-,f,-,-)} \\
& p(x',x,y,y) &
\end{tikzcd}\]
Secondly, we construct a morphism from the right hand side to the left hand
side using the projections. The following diagram shows dinaturality for any \(f \colon x \to x'\)
and any \(g \colon y \to y'\), which amounts to saying that it is dinatural for any \((f,g) \colon (x,y) \to (x',y')\).
Commutativity for this diagram uses that
\(\pi_y \colon \int\nolimits_d p(-,-,d,d) \tonat p(-,-,y,y)\) must be given by a natural transformation for
any \(y \in \D\) because we are considering an end over a category of functors.
\[\begin{tikzcd}[column sep=-5ex]
                                                 &                                                             &                                 & {\int_c\int_dp(c,c,d,d)} \arrow[lld, "\pi_x"'] \arrow[rrd, "\pi_{x'}"] &                                  &                                                               &                                                    \\
                                                 & {\int_d p(x,x,d,d)} \arrow[rrd, "f"'] \arrow[ldd, "\pi_y"'] &                                 &                                                                        &                                  & {\int_dp(x',x',d,d)} \arrow[lld, "f"] \arrow[rdd, "\pi_{y'}"] &                                                    \\
                                                 &                                                             &                                 & {\int_dp(x,x',d,d)} \arrow[ldd, "\pi_y"'] \arrow[rdd, "\pi_{y'}"]      &                                  &                                                               &                                                    \\
{p(x,x,y,y)} \arrow[rdd, "g"'] \arrow[rrd, "f"'] &                                                             &                                 &                                                                        &                                  &                                                               & {p(x',x',y',y')} \arrow[ldd, "g"] \arrow[lld, "f"] \\
                                                 &                                                             & {p(x,x',y,y)} \arrow[rdd, "g"'] &                                                                        & {p(x,x',y',y')} \arrow[ldd, "g"] &                                                               &                                                    \\
                                                 & {p(x,x,y,y')} \arrow[rrd, "f"']                             &                                 &                                                                        &                                  & {p(x',x',y,y)} \arrow[lld, "f"]                               &                                                    \\
                                                 &                                                             &                                 & {p(x,x',y,y')}                                                         &                                  &                                                               &                                                   
\end{tikzcd}\]
\end{proof}
\begin{remark}
\label{remark-limits-commute}
It can be shown that (co)ends over profunctors that are mute in
the contravariant variable are canonically isomorphic to (co)limits.
As a consquence, the Fubini rule is also valid for limits. Given
any \(F \colon I \times J \to \C\), we have the following isomorphisms.
\[
\mathsf{lim}_{I} \mathsf{lim}_J F \cong
\mathsf{lim}_{I \times J} F \cong
\mathsf{lim}_J \mathsf{lim}_I F 
\]
In particular, ends distribute over products.
\end{remark}

\subsection{Yoneda lemma}
\label{sec:orgbf5145d}
We will now partially justify the intuition that ends are universal
quantifiers plus some naturality conditions showing that natural
transformations are particular cases of ends.  This motivates
a rewrite of the Yoneda lemma in terms of (co)ends.
\begin{proposition}
\label{org2c434cd}
The set of natural transformations between two functors \(F,G \colon \C \to \D\) is given by
the following end.
\[
\mathrm{Nat}(F,G) = \int_{c \in \C} \D(Fc,Gc).
\]
\end{proposition}
\begin{proof}
A natural transformation \(\alpha \colon F \tonat G\) is equivalently a family of morphisms
\(\alpha_c \colon 1 \to \D(Fc, Gc)\) indexed by \(c \in \C\) and such that, for any \(f \colon c \to d\),
it holds that \(\alpha_c \circ \D(\id,f) = \D(f,\id) \circ \alpha_d\).  This means that the elements
of the set that the end defines, \(1 \to \int\nolimits_{c \in \C} \D(Fc,Gc)\), are precisely the
natural transformations \(F \tonat G\).
\end{proof}

\begin{remark}
In the case of \(\Sets\), we will allow ourselves to sometimes write the internal hom
as \((a \to b) \cong \Sets(a,b)\).  This makes many (co)end derivations less noisy
and closely follows the arrow notation of some functional programming
languages.  Natural transformations between copresheaves \(F,G \colon \C \to \Sets\)
can be written as
\[
\Nat(F,G) = \int_{c \in \C} Fc \to Gc.
\]
In the enriched case, the same could be said for any cartesian Benabou
cosmos \(\V\), taking it to be enriched over itself.
\end{remark}
In an arbitrary category \(\C\), consider the representable functors \(\C(-,a)\)
and \(\C(a,-)\) for some \(a \in \C\). In the cases where we want to keep the category \(\C\)
implicit, we will use the hiragana ``yo'' to write them as
\(\hirayo^a \colon \C^{op} \to \Sets\) and \(\hirayo_a \colon \C \to \Sets\), respectively.  This notation is
inspired both from the one used in \cite{loregian15} and from the main role
the \emph{Yoneda lemma} will play in this text. The Yoneda lemma says that,
for any copresheaf \(F \colon \C \to \Sets\) and any \(c \in \C\), the set of
natural transformations \(\Nat(\hirayo_a, F)\) is in bijection with the set \(Fa\).
The bijection is natural in both \(F\) and \(a\), and it is constructed from
the fact that any natural transformation is uniquely determined by the
image of \(\id \in \hirayo_a(a)\).

Using (co)ends, we can rephrase this result as the fact that
every (co)presheaf can be written as a (co)end.  This is called
the \emph{ninja Yoneda lemma} after a comment by T. Leinster \cite{ninjayoneda}.

\begin{lemma}[Ninja Yoneda lemma]
For any functor \(F \colon \C \to \Sets\), we have canonical isomorphisms
\[
Fa \cong \int_{c \in \C} \hirayo_ac \to Fc,
\qquad\quad
Fa \cong \int^{c \in \C} Fc \times \hirayo^ac,
\]
which we call the \emph{Yoneda reduction} and \emph{Coyoneda reduction}.
\end{lemma}
\begin{proof}
The Yoneda reduction follows from the usual statement of the 
Yoneda lemma, knowing that natural transformations can be written 
as ends as in Proposition \ref{org2c434cd}.
We will prove the Coyoneda reduction from the Yoneda reduction on the
opposite category.  Take an arbitrary set \(s \in \Sets\), and consider the
following chain of natural isomorphisms.
\[\begin{aligned}
& \Sets \left(\int^{c \in \C} Fc \times \hirayo^ac \ ,\ s \right) \\
\cong & \qquad\mbox{(Continuity)} \\
& \int_{c \in \C}  \Sets (Fc \times \hirayo^ac , s) \\
\cong & \qquad\mbox{(Exponential)} \\
& \int_{c \in \C} \Sets (\hirayo^ac , \Sets (Fc , s)) \\
\cong & \qquad\mbox{(Yoneda reduction in $\C^{op}$)} \\
& \Sets(Fa, s).
\end{aligned}\]
Because of (the usual) Yoneda lemma, this means that
\(Fa\) is isomorphic to \(\int\nolimits^{c \in \C} Fc \times \hirayo^ac\).
\end{proof}

\begin{remark}
An intuition on the Yoneda lemma (see \cite{loregian15} Remark 2.6) is
that it allows us to integrate interpreting the \(\hirayo\) as a Dirac's delta
for ends.
\end{remark}

\subsection{Kan extensions}
\label{sec:org8723af3}
In this text we will be working with \emph{global} Kan extensions that arise
as adjoints to functor precomposition.  However, \emph{local} Kan extensions
with their usual definition can exist in particular cases even if the
adjunction that we are using to define them does not.  We refer, for
instance, to Chapter 6 of \cite{riehl17} for the more usual definition
and study of \emph{local} Kan extensions.

\begin{definition}
\label{org711ebd7}
The \emph{left} and \emph{right} \textbf{Kan extensions} along a functor \(F \colon \C \to \D\) are the
left and right adjoints, respectively, to the functor given by precomposition
\((- \circ F) \colon [\D, \Sets] \to [\C, \Sets]\).  In other words, we have the following
natural isomorphisms.
\begin{align*}
[\D, \Sets]( \mathsf{Lan}_FG , H ) \cong [\C, \Sets](G, H \circ F). \\
[\C, \Sets]( H \circ F , G ) \cong [\D, \Sets](H, \mathsf{Ran}_F G).
\end{align*}
We write \(\mathsf{Lan}_{F}\) for the left Kan extension along \(F\) and \(\mathsf{Ran}_{F}\) for the right
Kan extension along \(F\).  We say that \(\mathsf{Lan}_FG\) and \(\mathsf{Ran}_FG\) are the \emph{left/right} Kan
extensions of \(G\) \emph{along} \(F\).
\end{definition}

In some specially well-behaved categories, and particularly in \(\Sets\),
the left and right Kan extensions exist and have a formula in terms of
(co)ends \cite[\S 2.1]{loregian15}.

\begin{proposition}
\label{orge75964b}
For any \(F \colon \C \to \mathbf{D}\) and any \(G \colon \C \to \Sets\), the left and right Kan
extensions of \(G\) along \(F\) exist and are canonically isomorphic to the
following (co)ends.
\[
\mathsf{Ran}_F G \cong \int_{c \in \C} \Sets(\D(-, Fc) ,Gc),
\qquad
\mathsf{Lan}_F G \cong \int^{c \in \C} \D(Fc, -) \times Gc.
\]
\end{proposition}
\begin{proof}
\cite[\S 2.1]{loregian15}
The following chain of natural isomorphisms proves the bijection. It
relies on the previous results in (co)end calculus.  We will first
prove the adjunction for the right Kan extension.
\begin{align*}
& \Nat \left( H , \int_{c \in \C} \Sets (\D(-,Fc) , Gc) \right) \\
\cong & \qquad \mbox{(Natural transformation as an end)} \\
& \int_{d \in \D}  \Sets \left(   Hd , \int_{c \in \C} \Sets (\D(d,Fc) , Gc) \right) \\
\cong & \qquad \mbox{(Continuity)} \\
& \int_{d \in \D} \int_{c \in \C}  \Sets \left(   Hd ,  \Sets (\D(d,Fc) , Gc) \right) \\
\cong & \qquad \mbox{(Exponential)} \\
& \int_{d \in \D} \int_{c \in \C}  \Sets \left(  \D(d,Fc) \times Hd , Gc \right) \\
\cong & \qquad \mbox{(Fubini)} \\
&  \int_{c \in \C} \int_{d \in \D}  \Sets \left(  \D(d,Fc) \times Hd , Gc \right) \\
\cong & \qquad \mbox{(Continuity)} \\
&  \int_{c \in \C}   \Sets \left( \int^{d \in \D} \D(d,Fc) \times Hd , Gc \right) \\
\cong & \qquad \mbox{(Yoneda lemma)} \\
&  \int_{c \in \C}   \Sets \left( HFc , Gc \right) \\
\cong & \qquad \mbox{(Natural transformation as an end)} \\
&  \Nat \left( H \circ F , G \right). \\
\end{align*}
A similar reasoning can be used for the case of the left Kan extension.
\begin{align*}
& \Nat \left( \int^{c \in \C} \D(Fc, -) \times Gc , H\right) \\
\cong & \qquad \mbox{(Natural transformation as an end)} \\
& \int_{d \in \D} \Sets \left(\int^{c \in \C} \D(Fc, d) \times Gc , Hd\right) \\
\cong & \qquad \mbox{(Continuity)} \\
& \int_{d \in \D} \int_{c \in \C} \Sets \left( \D(Fc, d) \times Gc , Hd\right) \\
\cong & \qquad \mbox{(Exponential)} \\
& \int_{d \in \D} \int_{c \in \C} \Sets \left(  Gc , \Sets(\D(Fc, d), Hd)\right) \\
\cong & \qquad \mbox{(Fubini)} \\
& \int_{c \in \C}\int_{d \in \D}  \Sets \left(  Gc , \Sets(\D(Fc, d), Hd)\right) \\
\cong & \qquad \mbox{(Continuity)} \\
& \int_{c \in \C}  \Sets \left(  Gc , \int_{d \in \D} \Sets(\D(Fc, d), Hd)\right) \\
\cong & \qquad \mbox{(Yoneda lemma)} \\
& \int_{c \in \C}  \Sets \left(  Gc , HFc\right) \\
\cong & \qquad \mbox{(Natural transformation as an end)} \\
& \Nat\left(  G , H \circ F\right).
\end{align*}
As adjoints are unique up to isomorphism, these two derivations imply the result.
\end{proof}

\section{Monoids}
\label{sec:orgbff6d35}
Monoids will appear repeteadly in this text in various forms.  Instead
of describing each one of these separately, we will make use of the
formal theory of monoids and monads \cite{street72}, that defines them
as objects in an arbitrary 2-category.  We will even consider
\emph{pseudomonoids} (see \cite{nlab}), monoids up to isomorphism in a monoidal
2-category.

\subsection{The category of monoids}
\label{sec:orgf2864db}

\begin{definition}
In a monoidal category \((\M, \otimes, i, \lambda, \rho, \alpha)\), a \textbf{monoid} is an object \(m \in \M\) endowed
with morphisms \(e \colon i \to m\) and \(\mu \colon m \otimes m \to m\), called respectively \emph{unit}
and \emph{multiplication}; and such that the following equalities hold.  These are
called \emph{left/right unitality} and \emph{associativity}.
\[\begin{tikzcd}
m \rar{\lambda} & i\otimes m \dar{e \otimes \id} & m \otimes i \dar[swap]{\id \otimes e} & m \lar[swap]{\rho} \\
& m \otimes m \ular{\mu} & m \otimes m \urar[swap]{\mu} &
\end{tikzcd}\]
\[\begin{tikzcd}
\dar{\mu} m \otimes (m \otimes m) \ar{rr}{\alpha} && (m \otimes m) \otimes m \dar{\mu} \\
m \otimes m \rar{\mu} & m & m \otimes m \lar[swap]{\mu}
\end{tikzcd}\]
Diagramatically, following \cite{marsden14}, these equations are the following.

\[\diagramone = \diagramtwo = \diagramthree\]
\[\diagramfour = \diagramfive\]

The same definition can be done in a bicategory after fixing a 0-cell.
Recall that a monoidal category can be seen as a bicategory with a
single object. If we relax the equalities to be isomorphisms, we get the notion of
\textbf{pseudomonoid} in a monoidal 2-category. Finally, \textbf{comonoids} are monoids
in the opposite category.
\end{definition}
\begin{definition}
A \textbf{morphism of monoids} between \((m, \mu, e)\) and \((n, \mu', e')\) in the monoidal
category \(\M\) is given by a morphism \(f \colon m \to n\) such that the following
diagrams commute.
\[\begin{tikzcd}[column sep=tiny]
m \ar{rr}{f} && n & m \otimes m \ar{rrr}{f \otimes f} \dar[swap]{\mu} &&\phantom{a}& n \otimes n \dar{\mu'} \\
& i \ular{e}\urar[swap]{e'} &  & m\ar{rrr}{f} &&& n
\end{tikzcd}\]
In other words, it preserves the unit and the multiplication.
\end{definition}
We can construct a category of monoids, \(\mathbf{Mon}(\M)\), over the
original monoidal category. This makes it possible to talk about
freely generated monoids.  In some cases, we have a formula for constructing
them.

\begin{proposition}[{{\cite[\S VII.3, Theorem 2]{maclane71}}}]
\label{orgd94a312}
Let \(\C\) be a category with coproducts indexed by the natural numbers and where
the tensor product distributes over the coproducts. The forgetful functor \(U \colon \mathbf{Mon}(\M) \to \M\)
has a left adjoint \((-)^{\ast} \colon \mathbf{Mon}(\M) \to \M\) which is given by the following
\emph{geometric series}.
\[
x^{\ast} = i + x + x \otimes x + x \otimes x \otimes x + \dots
\]
\end{proposition}
Finally, we can consider monoids acting on another object of the
bicategory, called a \emph{monoid module}. This can also be extended to 
the case of a 2-category. In the more general case of a monoidal
2-category, the equalities are relaxed again to isomorphisms and
we can define actions of a pseudomonoid, which are called \emph{actegories}
\cite{nlab}.

\begin{definition}
\label{def:module}
Let \((m, e, \mu)\) be a monoid in a 2-category \(\M\) with \(m \colon A \to A\).  A \textbf{module}
for this monoid is a 1-cell \(n \colon B \to A\) endowed with a 2-cell \(h \colon m \otimes n \to n\).
This requires \(n\) to be composable with \(m\) on one side.  The module
must satisfy some axioms saying that it interplays nicely with 
unitality and multiplication, as described in the following diagrams
and in Figure \ref{figaxiomsmodule}.
\[\begin{tikzcd}
(m \otimes m) \otimes n \dar[swap]{\mu}\ar{rr}{\alpha} && m \otimes (m \otimes n) \dar{h} & n \rar{\lambda} & i \otimes n \dar{e \otimes \id} \\
m \otimes n \rar{h} & n & m \otimes n \lar[swap]{h} && m \otimes n \ular{h}
\end{tikzcd}\]
\end{definition}

The graphical representation of a monoid \(m \in \M\) acting on some
object \(a \in \M\) where \(\M\) is a 2-category can be seen in Figure \ref{figaxiomsmodule}.
\begin{figure}[h]
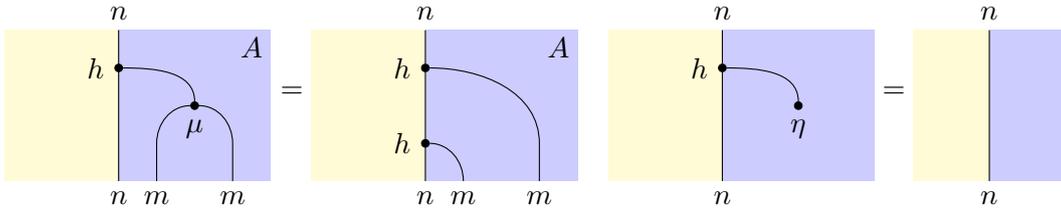

\begin{center}
\[ \monoidone = \monoidtwo \quad \monoidthree = \monoidfour\]
\end{center}
\caption{Axioms for a module in a 2-category, following Definition \ref{def:module}}
\label{figaxiomsmodule}
\end{figure}
\begin{definition}[\cite{street72}]
\label{def:kleisli}
Let \((m,e,\mu)\) be a monoid in a 2-category with \(m \colon A \to A\).
The \textbf{Kleisli object} for this monoid is given by a module \((A_m, f_m, \lambda)\) where 
\(f_m \colon A \to A_m\) is a 1-cell and \(\lambda \colon f_m \circ m \tonat f_m\) is a 2-cell.  The Kleisli object
is the universal module, in the sense that any other module \((B,g,\kappa)\) can be
written uniquely as this module in parallel with some \(u \colon A_m \to B\).
\end{definition}
\begin{figure}[h]
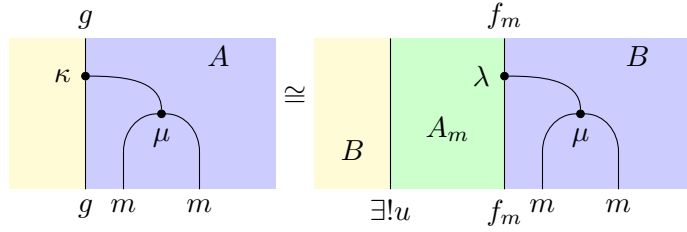

\begin{center}
\[\kleislione \cong \kleislitwo\]
\end{center}
\caption{A diagrammatic description of the universal property of the Klesli object, following Definition \ref{def:kleisli}.}
\label{diagramkleisli}
\end{figure}
In the particular case of monads in \(\mathbf{Cat}\), one recovers the usual notion
of Kleisli category.
\subsection{Applicative functors}
\label{sec:org0a9bd56}
\begin{definition}
Let \((\M,\otimes, i)\) be a monoidal category. The category of copresheaves over
it, \([\M , \Sets]\) can be endowed with the structure of a monoidal category
using \textbf{Day convolution}.  We write Day convolution of two copresheaves
\(F, G \colon \M \to \Sets\) as \(F \ast G\); it is given by the following coend.
\[
(F \ast G)(m) = \int^{x,y \in \M} \M(x \otimes y , m) \times F(x) \times G(y)
\]
\end{definition}

Using coend calculus, we can check both associativity and the fact
that the unit for this monoidal category is \(\hirayo_i\), the
representable copresheaf on the unit of the monoidal category.
See \cite[\S6]{loregian15} for details.
It makes sense now to ask what are the monoids of this new monoidal
category. In the case where we also take our base monoidal category
to be \(\Sets\), these are the \emph{applicative functors} widely used in
programming contexts and defined in \cite{mcbride08}. The description
of applicatives as monoids for Day convolution can be found in 
\cite{rivas17}.

\begin{definition}
\label{org478a9a6}
An \textbf{applicative functor} is a monoid on the category of endofunctors
on \(\Sets\) endowed with Day convolution as a tensor product.
\end{definition}

This means that it is endowed with natural transformations \(e \colon \hirayo_i \tonat F\)
and \(\mathbf{\mu} \colon F \ast F \tonat F\). Considering we are working in a cartesian closed
category, we can reduce the first to a family \(u_a \colon a \to Fa\) natural in \(a\); and the
second can be reduced via (co)end calculus to a natural family of morphisms
\(w_{a,b} \colon Fa \times Fb \to F(a \times b)\) natural on \(a, b \in \Sets\). 
\begin{align*}
& F \ast F \tonat F \\
\cong & \qquad\mbox{(Natural transformation as an end)} \\
& \int_c \Sets((F \ast F)c ,  Fc) \\
\cong & \qquad\mbox{(Definition of Day convolution)} \\
& \int_c \Sets\left(\left(  \int^{a,b} Fa \times Fb \times \Sets(a \times b , c) \right) , Fc \right) \\
\cong & \qquad\mbox{(Continuity)} \\
& \int_c\int_{a,b} \Sets(Fa \times Fb \times \Sets(a \times b , c) , Fc) \\
\cong & \qquad\mbox{(Fubini)} \\
& \int_{a,b}\int_c \Sets(Fa \times Fb \times \Sets(a \times b , c) , Fc) \\
\cong & \qquad\mbox{(Yoneda lemma)} \\
& \int_{a,b}\Sets(Fa \times Fb , F(a \times b)).
\end{align*}
We will write now the monoid axioms in terms of this representation.
Knowing how the Yoneda lemma is constructed, one can see how to write
\(\mu\) in terms of \(u\).  The laws then become the following diagrams; we left
implicit the unitors and associators of the cartesian product.
\[\begin{tikzcd}
Fa \drar[swap]{\id}\rar{u_1} & Fa \times F1 \dar{\mu_{a,1}} & F1 \times Fa \dar[swap]{\mu_{1,a}} & Fa\lar[swap]{u_1} \dlar{\id} \\
& Fa & Fa &
\end{tikzcd}\]
\[\begin{tikzcd}
Fa \times Fb \times Fc \rar{w_{a,b}} \dar[swap]{w_{b, c}} & F(a \times b) \times Fc \dar{w_{a \times b, c}} \\
Fa \times F(b \times c) \rar{w_{a,b \times c}} & F(a \times b \times c)
\end{tikzcd}\]
Applicative functors are also defined sometimes with a family of morphisms of the
form \(\mu'_{b,c} \colon F(b \to c) \times Fb \to Fc\). This is again equivalent via
the Yoneda lemma.
\begin{align*}
& \int_{a,b} \Sets(Fa \times Fb , F(a \times b)) \\
\cong & \qquad\mbox{(Yoneda lemma)} \\
& \int_{a,b} \Sets \left(Fa \times Fb ,\int_c \Sets \left(  \Sets(a \times b , c) , F(c) \right)\right) \\
\cong & \qquad\mbox{(Exponential)} \\
& \int_{a,b} \Sets \left(Fa \times Fb ,\int_c \Sets \left(  \Sets(a, b \to c) , F(c) \right)\right) \\
\cong & \qquad\mbox{(Continuity)} \\
& \int_{a,b} \int_c \Sets \left(Fa \times Fb , \Sets \left(  \Sets(a, b \to c) , F(c) \right)\right) \\
\cong & \qquad\mbox{(Fubini)} \\
& \int_{b,c} \int_{a} \Sets \left(Fa \times Fb , \Sets \left(  \Sets(a, b \to c) , F(c) \right)\right) \\
\cong & \qquad\mbox{(Exponential)} \\
& \int_{b,c} \int_{a} \Sets \left(Fa \times Fb \times \Sets(a, b \to c), F(c) \right) \\
\cong & \qquad\mbox{(Exponential)} \\
& \int_{b,c} \int_{a} \Sets \left(\Sets(a, b \to c), \Sets(Fa \times Fb , F(c)) \right) \\
\cong & \qquad\mbox{(Yoneda)} \\
& \int_{b,c} \Sets(F(b \to c) \times Fb , F(c)). \\
\end{align*}
In any of these cases, we can define a category \(\mathbf{App}\) of applicative functors
with monoid morphisms between them.
Being monoids, the next question is if we can generate \emph{free} applicative functors.
We will apply Proposition \ref{orgd94a312}, but we first need a lemma showing that
that our case actually satisfies the conditions of the theorem.

\begin{lemma}
\label{orgf98c8e5}
Day convolution distributes over coproducts.  Recall that colimits in a
category of presheaves are computed pointwise. We are saying then that
\(F \ast (G + H) \cong (F \ast G) + (F \ast H)\) for any functors \(F,G,H \in [\C , \Sets]\).
\end{lemma}
\begin{proof}
The proof is straightforward using (co)end calculus and making the coend
distribute over the coproduct.  We compute the following chain of isomorphisms
natural in \(c \in \Sets\).
\begin{align*}
& F \ast (G + H)  \\
\cong & \qquad\mbox{(Definition of Day convolution)} \\
& \int^{x,y} \Sets(a \times b , -) \times (G + H)a \times Fb \\
\cong & \qquad\mbox{(The coproduct of functors is computed pointwise)} \\
& \int^{x,y} \Sets(a \times b , -) \times (Ga + Ha) \times Fb \\
\cong & \qquad\mbox{(Product distributes over sum)} \\
& \int^{x,y} \Sets(a \times b , -) \times Ga \times Fb + \Sets(a \times b , -) \times Ha \times Fb\\
\cong & \qquad\mbox{(End distributes over sum)} \\
& \left(  \int^{x,y} \Sets(a \times b , -) \times Ga \times Fb \right) +  \left(  \int^{x,y} \Sets(a \times b , -) \times Ha \times Fb \right)\\
\cong & \qquad\mbox{(Definition of Day convolution)} \\
& (F \ast G) + (F \ast H). & \qedhere
\end{align*}
\end{proof}

\begin{theorem}
\label{org4c8a4e9}
Let \(X \colon \Sets \to \Sets\) be an endofunctor. The free applicative functor
over it, \(X^{\ast}\), can be computed as the following colimit.
\[
X^{\ast} = \mathrm{Id} + X + X \ast X + X \ast X \ast X + \dots
\]
In particular, we have the following adjunction \(\mathbf{App}(X^{\ast}, F) \cong [\C, \Sets](X,F)\)
natural in \(F\), an applicative functor; and \(X\), an arbitrary functor.
\end{theorem}
\begin{proof}
We apply Proposition \ref{orgd94a312}.  The category
of copresheaves has all small coproducts and Day convolution
distributes over coproducts because of Lemma \ref{orgf98c8e5}. Note
also that the unit of Day convolution in this case is given by \(\Sets(1,-)\),
which is naturally isomorphic to the identity functor.
\end{proof}

We will be interested in the following particular case both during our
study of traversable functors and while describing an optic for applicative
functors.

\begin{corollary}
\label{org79dc269}
Let \(S(c) = a \times (b \to c)\) be a functor for some fixed \(a,b \in \Sets\). Its free applicative
functor is \(S^{\ast}(c) = \sum\nolimits_n a^n \times (b^n \to c)\).
\end{corollary}
\begin{proof}
We will show that the Day convolution of \(a^n \times (b^n \to -)\) with \(a \times (b \to -)\) is
precisely \(a^{n+1} \times (b^{n+1} \to -)\).
\begin{align*}
& \int^{x, y} \Sets(x \times y , c) \times (a \times (b \to x)) \times (a^n \times (b^n \to y)) \\
\cong & \qquad\mbox{(Fubini rule)} \\
& a^{n + 1} \times \int^{x, y} \Sets(x \times y , c) \times (b \to x) \times (b^n \to y) \\
\cong & \qquad\mbox{(Exponential)} \\
& a^{n + 1} \times \int^{x, y} \Sets(x , y \to c) \times (b \to x) \times (b^n \to y) \\
\cong & \qquad\mbox{(Yoneda lemma)} \\
& a^{n + 1} \times \int^{y} (b \to (y \to c)) \times (b^n \to y) \\
\cong & \qquad\mbox{(Exponential)} \\
& a^{n + 1} \times \int^{y} (y \to (b \to c)) \times (b^n \to y) \\
\cong & \qquad\mbox{(Yoneda lemma)} \\
& a^{n + 1} \times (b^n \to (b \to c)) \\
\cong & \qquad\mbox{(Exponential)} \\
& a^{n + 1} \times (b^{n+1} \to c)
\end{align*}
We note also that the identity functor can be written as \(a^0 \times (b^0 \to -)\). Thus,
we can apply Lemma \ref{org4c8a4e9} and complete the proof by induction.
\end{proof}

\subsection{Morphisms of monads}
\label{sec:org0c11977}
\begin{definition}[\cite{reddy95}]
\label{org4babfba}
A \textbf{monad morphism} between two monads, \((S, \mu, \eta)\) and \((T, \mu', \eta')\)
is a natural transformation \(\alpha \colon S \to T\) preserving units and multiplications.
Diagrammatically, it must make the following diagram commute.
\[\begin{tikzcd}
S^2 \rar{\mu} \dar[swap]{\alpha^2} & S \dar{\alpha} & \lar[swap]{\eta} \mathrm{Id} \dar{\id} \\
T^2 \rar[swap]{\mu'} & T & \lar{\eta'} \mathrm{Id} \\
\end{tikzcd}\]
In other words, it is a morphism of monoids in the monoidal category of endofunctors.
Dually, a morphism of comonads is a morphism of comonoids in the category of
endofunctors.
\end{definition}

Every monad morphism \(\alpha \colon S \to T\) induces a functor \(T\mbox{-Alg} \to S\mbox{-Alg}\) 
between the Eilenberg-Moore categories of the monads. A morphism of \(T\mbox{-algebras}\) \(f \colon (a, \rho) \to (b, \rho')\),
can be reinterpreted as a morphism of \(S\mbox{-algebras}\) as follows.
\[\begin{tikzcd}
Sa \dar[swap]{\alpha} \rar{Sf} & Sb \dar{\alpha} \\
Ta \dar[swap]{\rho} \rar{Tf} & Tb \dar{\rho'} \\
a  \rar{f} & b
\end{tikzcd}\]
In this diagram, the upper square commutes because of the definition of
monad morphism. The lower square commutes because \(f\) is a morphism of
\(T\mbox{-algebras}\).

\begin{proposition}
\label{org88623d6}
Let \((S, \mu, \eta)\) and \((T, \mu', \eta')\) be two monads on the same category \(\C\).
Let \(A \colon T\mbox{-Alg} \to S\mbox{-Alg}\) be a functor preserving the forgetful
functor.
\[\begin{tikzcd}[column sep=0ex]
T\mbox{-Alg} \drar[swap]{U}\ar{rr}{A} && S\mbox{-Alg} \dlar{U}\\
& \C &
\end{tikzcd}\]
Then \(A\) is induced by some monad morphism in the way described earlier.
\end{proposition}
\begin{proof}
We know that, for any \(x \in \C\), the functor \(A\) takes the free \(T\mbox{-algebra}\)
\((\mu_x' \colon T^2x \to Tx)\) into some \(S\mbox{-algebra}\) that we call \((\overline{\mu}_x \colon STx \to Tx)\).
We will show that this defines a natural transformation \(\overline{\mu} \colon ST \tonat T\).
In fact, given any \(f \colon x \to y\), the morphism \(Tf \colon Tx \to Ty\) is a morphism
between the free algebras because of naturality of \(\mu'\). Given that \(A\) 
preserves the forgetful functor, \(Tf\) must also be a morphism between the
algebras given by \(\overline{\mu_x}\) and \(\overline{\mu_y}\).  This is precisely naturality.
\[\begin{tikzcd}
STx \dar[swap]{\overline{\mu}_x} \rar{STf} & STy \dar{\overline{\mu}_y}\\
Tx \rar{Tf} & Ty \\
T^2x \uar{\mu'} \rar{T^2f} & T^2y \uar[swap]{\mu'}
\end{tikzcd}\]
Now we define the natural transformation \(\alpha = \overline{\mu} \circ S\eta \colon S \tonat T\). We will
show it is a monad morphism. It preserves the unit because of commutativity
of the following diagram, that uses that \(\eta\) is a natural transformation and
\(\overline{\mu}\) is a monad algebra.
\[\begin{tikzcd}
x \dar{\eta_{x}}\rar{\eta'_{x}} & Tx \dar{\eta_{Tx}} \drar{\id} & \\
Sx \rar[swap]{S\eta'_{x}} & STx \rar[swap]{\overline{\mu}_{x}}& Tx
\end{tikzcd}\]
It also preserves the multiplication because of commutativity of the following
diagram. Here we use that \(\mu\) is a natural transformation, that \(\overline{\mu}\) is a monad
algebra, the unitality of the monad \(T\), and the fact that \(\mu'\) must be a morphism
between the corresponding \(S\mbox{-algebras}\) thanks to the action of the functor.
\[\begin{tikzcd}
S^2x \rar{S^2\eta_x} \dar{\mu_x} & S^2Tx \rar{S\overline{\mu}_x} \dar{\mu_{Tx}}& STx \rar{S\eta'_{Tx}}\dar{\overline{\mu}_x} & ST^2x \dar{\overline{\mu}_{Tx}} \lar[bend left]{S\mu'_x} \\
Sx \rar{S\eta'_x} & STx \rar{\overline{\mu}_x} & Tx & T^2x \lar[swap]{\mu'_x} \\
&& T^2x \uar{\mu_{x}} & T^3x \lar{T\mu'_x} \uar[swap]{\mu_{Tx}}
\end{tikzcd}\]
We can finally check that this \(\alpha\) induces the original functor \(A\). In fact,
for any \(T\mbox{-algebra}\) \((l \colon Tx \to x)\), we have an algebra morphism from the free
algebra \((\mu' \colon T^2x \to Tx)\) given precisely by \(l \colon Tx \to x\).  This must be sent by
the functor to the following \(S\mbox{-algebra}\) morphism, whose commutativity
determines that \(A\) is induced by \(\alpha\). \qedhere
\[\begin{tikzcd}
STx\dar{\overline{\mu}_{x}} \rar{Sl}& Sx \lar[bend left]{S\eta} \dar \\
Tx\rar{l} & x
\end{tikzcd}\]
\end{proof}

\subsection{Monoidal actions}
\label{sec:org807ac02}
A monoidal category can be seen as a pseudomonoid in \(\mathbf{Cat}\), and
we can consider strong monoidal functors between them that preserve
the action up to isomorphism.

\begin{definition}
\label{orga75d84e}
A \textbf{strong monoidal functor} between two monoidal categories \((\M,\otimes,i)\) and \((\N,\boxtimes,j)\)
is given by a functor \(F \colon \M \to \N\) together with \emph{structure isomorphisms}
\(\phi_i \colon j \cong F(i)\) and \(\phi_{m,n} \colon F(m) \boxtimes F(n) \cong F(m \otimes n)\) that interplay nicely
with the associators and unitors of the monoidal category in the sense that 
they make the following diagrams commute.
\[\begin{tikzcd}
(Fx \boxtimes Fy) \boxtimes Fz \rar{\alpha_{\N}} \dar[swap]{\phi_{x,y}} & Fx \boxtimes (Fy \boxtimes Fz) \dar{\phi_{y,z}}\\
F(x \otimes y) \boxtimes Fz \dar[swap]{\phi_{(x \otimes y),z}} & Fx \boxtimes F(y \otimes z) \dar{\phi_{x,(y \otimes z)}}\\
F((x \otimes y) \otimes z) \rar{\alpha_M} & F(x \otimes (y \otimes z)) \\
\end{tikzcd}\]
\[\begin{tikzcd}
Fi \boxtimes Fx \dar[swap]{\phi_{i,x}}  & \lar[swap]{\phi_i} j \boxtimes Fx \dar{\lambda_\N} & Fx \boxtimes j \rar{\phi_i}\dar[swap]{\rho_{\N}} & Fx \boxtimes Fi \dar{\phi_{x,i}}\\
F(i \otimes x) \rar[swap]{F\lambda_{\M}} & Fx              & Fx             & F(x \otimes i) \lar{F\rho_{\M}}
\end{tikzcd}\]
Let us denote by \(\mathbf{MonCat}\) the category of monoidal categories with
strong monoidal functors between them. Note that the identity
and the composition of two strong monoidal functors are again
strong monoidal functors.
\end{definition}

\begin{definition}
\label{org7b5a63d}
A \textbf{monoidal action} from a monoidal category \(\M\) into a category \(\C\) is a strong monoidal
functor \(F \colon \M \to [\C,\C]\) that has as target the monoidal category of endofunctors of \(\C\)
endowed with composition as the monoidal product. We can consider the slice category
\(\mathbf{MonCat}/[\C,\C]\) as the category of monoidal actions.
\end{definition}

It can be shown that these are precisely pseudomonid modules \(\M \times \C \to \C\),
whose laws are true up to some isomorphism.  In the
case of strong monoidal actions, because of the strictness of the
monoidal category of endofunctors, the coherence diagrams become
simplified.

\section{The bicategory of profunctors}
\label{sec:orgc46321e}
\begin{definition}
The bicategory \(\mathbf{Prof}\) has categories as 0-cells. 1-cells
between two categories \(\mathbf{C}\) and \(\mathbf{D}\) are profunctors
\(\mathbf{C} \nrightarrow \mathbf{D}\), and 2-cells between two profunctors are natural transformations.
The composition of two profunctors \(p \colon \mathbf{C} \nrightarrow \mathbf{D}\) and \(q \colon \mathbf{D} \nrightarrow \mathbf{E}\) is
written as \((q \diamond p) \colon \mathbf{C} \nrightarrow \mathbf{E}\), and is given by the following coend.
\[
(q \diamond p)(c,e) = \int^{d \in \mathbf{D}} p(c,d) \times q(d,e).
\]
The identity for this composition in a category \(\mathbf{C}\) is the hom-profunctor 
\(\mathbf{C}(-,-) \colon \mathbf{C}^{op} \times \mathbf{C} \to \mathbf{Sets}\). A detailed description of this bicategory,
together with proofs for unitality and associativity can be found for
instance in \cite[\S5]{loregian15}, where profunctors are called \emph{relators}.
\end{definition}

\subsection{Promonads}
\label{sec:orgaa7d52c}
\begin{definition}
We give the name \textbf{promonads} to the monoids in the 2-category \(\mathbf{Prof}\).
\end{definition}

Let \(\C\) be a category and \(p \colon \C \nrightarrow \C\) an endoprofunctor.  A promonad
structure on this endoprofunctor is given by some \emph{unit}, a family
of functions \(\eta_{a,b} \colon \C(a,b) \to p(a,b)\) natural in \(a,b \in \C\); and some
\emph{multiplication}, a family of functions 
\[\mu_{a,c} \colon \left(   \int^b p(a,b) \otimes p(b,c) \right) \to p(a,c)\]
natural in \(a,b \in \C\). Using continuity, we can rewrite the multiplication
as an element of the following end, that resembles function composition.
\[
\mu_{a,b} \colon \int_b p(a,b) \times p(b,c) \to p(a,c)
\]
It can be shown that unitality for the promonad is the fact that
the identity \(\id \in \C(a,a)\) is sent to some element in \(p(a,a)\) that acts as
the identity for this composition. Associativity for the promonad is the
fact that this composition is associative.  In this sense, a promonad could
be seen as embedding the category \(\C\) into a category with the same objects
but new morphisms.  This intuition can be made precise by the following
proposition.

\begin{proposition}[\cite{street72}]
\label{orgb64b29f}
A promonad induces an identity-on-objects functor to the Kleisli
category of the promonad.  Every identity-on-objects functor arises in
this way for some promonad.
\end{proposition}

\subsection{Promonad modules}
\label{sec:org5e9c872}
\begin{lemma}
\label{org8d8bb19}
If \(\psi \colon \C^{op} \times \C \to \Sets\) is a promonad, then \((\psi \diamond -) \colon [\C^{op} \times \C, \Sets] \to [\C^{op} \times \C, \Sets]\) is a monad.
Modules over the promonad \(\psi\) are precisely the algebras over the monad.
\end{lemma}
\begin{proof}
The unit of the monad is given by the unit of the promonad \(\eta_p \colon p \cong \C(-,-)\diamond p \tonat \psi \diamond p\).
The multiplication is in turn given by the multiplication of the promonad as
\(\mu_p \colon \psi \diamond (\psi \diamond p) \cong (\psi \diamond \psi) \diamond p \tonat \psi \diamond p\).  It can be seen that the axioms for the
monad are precisely the axioms for the promonad.

Now a module over the promonad is given by \(\psi \diamond p \tonat p\), which is precisely 
the data for an algebra over the monad \((\psi \diamond -)\).  Because of the definition
of the unit and multiplication of the promonad, the axioms of the monad
algebra are precisely the axioms for the module.
\end{proof}

\section{Pseudomonic functors and replete subcategories}
\label{sec:org2dc131b}
\label{orga37ce19}
We will follow \cite{nlab} into getting the necessary definitions to construct
a well-behaved notion of 2-categorical image.

\begin{definition}
\label{orged771c2}
A functor \(F \colon \C \to \D\) is \textbf{pseudomonic} if it is faithful and full on
isomorphisms.  That is, given two \(c, c' \in \C\), every isomorphism
\(Fc \cong Fc'\) is the image under the functor of an isomorphism \(c \cong c'\).
\end{definition}

An alternative definition of pseudomonic functor is that they are
precisely the functors \(F \colon \C \to \D\) such that the following square
is a pullback.  It can be seen that pseudomonic functors are stable under
pullback.
\[\begin{tikzcd}
\C \rar{\id}\dar[swap]{\id} & \C\dar{F} \\
\C \rar{F} & \D
\end{tikzcd}\]

\begin{definition}
\label{org5ebfdb2}
A subcategory \(\C \subseteq \D\) is \textbf{replete} if, for all \(c \in \C\), the existence
of an isomorphism \(f \colon c \cong d\) in \(\D\) implies that \(d \in \C\) and \(f \in \C(c,d)\).
That is, the subcategory respects the isomorphisms of the original
category.
\end{definition}

The smallest replete subcategory containing a subcategory \(\C \subseteq \D\) is called
its \textbf{repletion}, \(\repl(\C)\). It can be explicitly constructed taking all objects 
of \(\D\) that admit an isomorphism to an object in \(\C\) and taking all morphisms
that can be written as composites of morphisms in \(\C\) and isomorphisms in \(\D\).

\begin{proposition}
\label{orge16820a}
Let \(\C \subseteq \D\) be a subcategory.
The inclusion \(\C \to \repl(\C)\) is an equivalence if and only if the inclusion
\(\C \to \D\) is pseudomonic.
\end{proposition}
\begin{proof}
Note that the inclusion of a replete subcategory, \(\repl(\C) \to \D\), is pseudomonic by definition.
If \(\C \to \repl(\C)\) is an equivalence, then it is fully faithful and pseudomonic in particular.
These two facts make the composite \(\C \to \mathrm{repl}(\C) \to \D\) pseudomonic.

Let \(\C \to \D\) be pseudomonic. Morphisms in \(\repl(\C)\) can be written
as the composition of morphisms of the form \(g\circ f \circ h\) where \(f\) is
a morphism in \(\C\) and \(g,h\) are isomorphisms in \(\D\).  We will show
that for any \(c,c' \in\C\), \(\repl(\C)(c,c') \subseteq \C(c,c')\) by induction on
the minimum number of pieces \((g\circ f \circ h)\) required to form the morphism.
In the case of a morphism \((g\circ f \circ h) \in \repl(\C)(c,c')\), the
isomorphisms \(g\) and \(h\) have their source and target in \(\C\), and they must
be morphisms of \(\C\) because of the pseudomonic condition.  In
the case of a morphism \((g_1 \circ f_1 \circ h_1) \circ \dots \circ (g_k \circ f_k \circ h_k)\), we can
apply the inductive hypothesis to show that \(f_1 \circ h_1 \circ \dots \circ g_k \circ f_k\) is 
a morphism in \(\C\) and then apply the same reasoning as in the base
case to conclude that \(g_1\) and \(h_k\) are morphisms in \(\C\).

This shows that the inclusion \(\C \to \repl(\C)\) is full. It is faithful
and essentially surjective by definition, which shows it is an equivalence.
\end{proof}

Given an arbitrary functor \(F \colon \C \to \D\), that does not need to be pseudomonic,
we can define its \textbf{image}, \(\mathrm{img}(F)\), as the least subcategory containing all objects
and morphisms that are images of the functor.  Note that this does not mean
that every morphism in \(\mathrm{img}(F)\) will be the image of some morphism in \(\C\) under
the functor \(F\); there could exist morphisms that become composable only after
applying the functor.  The image of any functor is a subcategory, and we
can define the \textbf{replete image} of a functor as the repletion of its image, \(\repl( \mathrm{img}(F) )\).

\chapter{Existential optics}
\label{sec:org06c3c04}
\section{Existential optics}
\label{sec:org7d4abd3}
\label{org2dbf2a9}
The structure that is common to all optics is that they split a bigger
structure \(s\) into the focus \(a\) and some context \(m\) acting on it. In
some sense, we cannot access or act on that context, only on its shape.
The definition will capture this fact using the dinaturality condition
of a coend. However, we can still use this context to update the
original data structure, replacing the current focus by a new element.

\begin{definition}
\label{org8ab3812}
Consider a monoidal category \(\M\) and two arbitrary categories \(\C\) and \(\D\).
Let \((\underline{\phantom{a}}) \colon \M \to [ \C , \C ]\) and \((\underline{\phantom{a}}) \colon \M \to [ \D , \D ]\) be two strong monoidal
functors and let \(s,a \in \C\) and \(t,b \in \D\). An \textbf{optic} from \((s,t)\) with focus on \((a,b)\) is an element
of the following set described as a coend.
\[\Optic
\left(\begin{pmatrix} a \\ b \end{pmatrix} , \begin{pmatrix} s \\ t \end{pmatrix}\right) =
\int^{m \in \M} \C(s, \underline{m} a) \times \D(\underline{m} b, t)
\]
\end{definition}
\begin{remark}
The definition of \emph{optic} given by \cite{riley18}, or the one considered in
\cite{boisseau18}, deal only with the particular case in which \(\D\) and
\(\C\) are the same category and both actions are identified.  We are
actually defining what \cite{riley18} calls \emph{mixed optics}, that do not
have this limitation.  The following construction of the category
of optics is similar to that of
Proposition 2.0.3 in \cite{riley18} but it provides a more general result,
as it is considering arbitrary monoidal actions instead of monoidal
products and \emph{mixed} optics instead of assuming \(\C = \D\).
\end{remark}

It can be shown directly that \(\Optic\) can be given the
structure of a category; but note that we could also wait for the
profunctor representation theorem (Theorem \ref{org9a4afc7})
to describe \(\Optic\) as a Kleisli category. The main idea is that 
we can compose \(s \to \mact a\) and \(a \to \nact x\) into \(s \to \underline{(m \otimes n)}x\); and in the
same way, we can compose \(\mact b \to t\) and \(\nact y \to b\) into \(\underline{(m \otimes n)}y \to t\). The
following definition and proofs are just a formality around the structure
morphisms of the monoidal category and the monoidal functor that ensures that
everything works nicely.

In order to show that
\(\Optic\) is a category we use a notation directly taken from \cite{riley18}.
We write elements of the end \(\int\nolimits^{m\in\M} \C(s,\mact a) \times \D(\mact b, t)\) as pairs of
functions \(\optic{l}{r}\) for \(l \in \C(s,\mact a)\) and \(r \in \D(\mact b, t)\).  These pairs are
quotiented by a relation that equates \(\optic{\alpha \circ l}{r} \sim \optic{l}{r \circ \alpha}\) for any
\(\alpha \in \M(m,n)\), any \(l \in \C(s,\mact a)\) and any \(r \in \D(\nact b, t)\).

The \emph{identity} of \(\Optic\) is defined as \(\id = \optic{\phi^{-1}_i}{\phi_i}\), where 
\(\phi_i \colon \iact a \to a\) is the structure map of the monoidal action. 
Consider two optics \(\optic{l}{r} \in \Optic((a,b),(s,t))\)
and \(\optic{l'}{r'} \in \Optic((a,b),(x,y))\) given by \(l \in \C(s,\mact a)\), \(l' \in \C(a, \nact x)\),
\(r' \in \D(\nact y , b)\) and \(r \in \D(\mact b, t)\); their \emph{composition} is defined as \(\optic{\phi^{-1}_{m,n} \circ \mact l' \circ l}{r \circ \mact r' \circ \phi_{m,n}}\),
where \(\phi_{m,n} \colon \underline{(m \otimes n)} a \to \mact\nact a\) is the structure map of the monoidal action.
We show now that composition is well-defined with respect
to the equivalence relation. In fact, for any \(\alpha \colon \M(m, m')\) and
\(\beta \colon \M(n, n')\), the relations \(\optic{\alpha \circ l}{r} \sim \optic{l}{r \circ \alpha}\) and
\(\optic{\beta \circ l'}{r'} \sim \optic{l'}{r'\circ \beta}\) translate into the following relation.
\begin{align*}
& \optic{\phi^{-1}_{m,n} \circ \mact'(\beta \circ l') \circ (\alpha \circ l)}{r \circ \mact r' \circ \phi_{m,n}} \\
= & \qquad\mbox{(Functoriality of $\mact'$)} \\
& \optic{\phi^{-1}_{m,n} \circ \mact'\beta \circ \mact'l' \circ \alpha \circ l}{r \circ \mact r' \circ \phi_{m,n}} \\
= & \qquad\mbox{(Functoriality of the action makes $\alpha$ natural)} \\
& \optic{\phi^{-1}_{m,n} \circ \mact'\beta \circ  \alpha \circ \mact l' \circ l}{r \circ \mact r' \circ \phi_{m,n}} \\
= & \qquad\mbox{(Naturality of $\phi$)} \\
& \optic{(\beta \otimes \alpha) \circ \phi^{-1}_{m,n} \circ \mact l' \circ l}{r \circ \mact r' \circ \phi_{m,n}} \\
\sim & \qquad\mbox{(Equivalence relation)} \\
& \optic{\phi^{-1}_{m,n} \circ \mact l' \circ l}{r \circ \mact r' \circ \phi \circ (\beta \otimes \alpha)} \\
= & \qquad\mbox{(Naturality of $\phi$)} \\
& \optic{\phi^{-1}_{m,n} \circ \mact l' \circ l}{r \circ \mact r'  \circ \mact\beta \circ \alpha \circ \phi_{m,n}} \\
= & \qquad\mbox{(Functoriality of $\mact$)} \\
& \optic{\phi^{-1}_{m,n} \circ \mact l' \circ l}{r \circ \mact (r'  \circ \beta) \circ \alpha \circ \phi_{m,n}} \\
= & \qquad\mbox{(Functoriality of the action makes $\alpha$ natural)} \\
& \optic{\phi^{-1}_{m,n} \circ \mact l' \circ l}{(r \circ \alpha) \circ \mact' (r'  \circ \beta) \circ \phi_{m,n}} \\
\end{align*}
We proceed now to check that the identity is neutral with respect to
composition and that composition is associative.  For the first, we
are going to use that the definition of strong monoidal functor
imposes \(\phi^{-1}_{m,i} \circ \mact \phi_i = \underline{\lambda}\).  Composition with the identity on the left
goes as follows.
\begin{align*}
  & \optic{l}{r} \circ \optic{\phi_i^{-1}}{\phi_i} \\
= & \qquad\mbox{(Definition of composition)} \\
  & \optic{\phi_{m,i}^{-1} \circ m\phi_i^{-1} \circ l}{r \circ m\phi_i\circ \phi_{m,i}} \\
= & \qquad\mbox{(Conditions on a strong monoidal functor)} \\
  & \optic{\lambda^{-1} \circ l}{r \circ \lambda} \\
\sim & \qquad\mbox{(Equivalence relation)} \\
  & \optic{l}{r} \\
\end{align*}
And composition on the right follows a similar reasoning.
\begin{align*}
  & \optic{\phi_i^{-1}}{\phi_i} \circ \optic{l}{r} \\
= & \qquad\mbox{(Definition of composition)} \\
  & \optic{\phi_{m,i}^{-1} \circ \iact l \circ \phi_i^{-1}}{\phi_i \circ \iact r\circ \phi_{m,i}} \\
= & \qquad\mbox{(Naturality of $\phi$)} \\
  & \optic{\phi_{m,i}^{-1} \circ \phi_i^{-1} \circ l}{r \circ\phi_i \circ \phi_{m,i}} \\
= & \qquad\mbox{(Conditions on a strong monoidal functor)} \\
  & \optic{\lambda^{-1} \circ l}{r \circ \lambda} \\
\sim & \qquad\mbox{(Equivalence relation)} \\
  & \optic{l}{r} \\
\end{align*}
Finally, associativity holds because of the following chain of
equations. Let \(s,a,x,u \in \C\) and \(t,b,y,v \in \D\). Let, on one
side, \(l \colon s \to \mact a\), \(l' \colon a \to \nact x\) and \(l'' \colon x \to \underline{k}u\) be morphisms in \(\C\);
let, on the other side, \(r \colon \mact b \to t\), \(r' \colon \nact y \to b\) and \(r'' \colon \underline{k}v \to y\) be
morphisms in \(\D\).
\begin{align*}
  & \optic{l}{r} \circ (\optic{l'}{r'} \circ \optic{l''}{r''}) \\
= & \qquad\mbox{(Definition of composition)} \\
  & \optic{l}{r} \circ \optic{\phi^{-1}_{n,k} \circ \nact l''\circ l'}{r' \circ \nact r'' \circ \phi_{n,k}} \\
= & \qquad\mbox{(Definition of composition)} \\
  & \optic{\phi^{-1}_{m,n\otimes k}\circ \mact (\phi^{-1}_{n,k} \circ \nact l''\circ l') \circ l}{r \circ \mact(r' \circ \nact r'' \circ \phi_{n,k}) \circ \phi_{m,n \otimes k}} \\
= & \qquad\mbox{(Functoriality of $\mact$)} \\
  & \optic{\phi^{-1}_{m,n\otimes k}\circ \mact \phi^{-1}_{n,k} \circ \mact\nact l'' \circ \mact l' \circ l}{r \circ \mact r' \circ \mact\nact r'' \circ \mact\phi_{n,k} \circ \phi_{m,n \otimes k}} \\
\sim & \qquad\mbox{(Equivalence relation)} \\
  & \optic{\alpha_{m,n,k} \circ \phi^{-1}_{m,n\otimes k}\circ \mact \phi^{-1}_{n,k} \circ \mact\nact l'' \circ \mact l' \circ l}{r \circ \mact r' \circ \mact\nact r'' \circ \mact\phi_{n,k} \circ \phi_{m,n \otimes k} \circ \alpha_{m,n,k}^{-1}} \\
= & \qquad\mbox{(Axioms of a strong monoidal functor)} \\
  & \optic{\phi^{-1}_{m\otimes n, k}\circ \phi^{-1}_{m,n} \circ \mact\nact l'' \circ \mact l' \circ l}{r \circ \mact r' \circ \mact\nact r'' \circ \phi_{m,n} \circ \phi_{m \otimes n, k}} \\
= & \qquad\mbox{(Naturality of $\phi$)} \\
  & \optic{\phi^{-1}_{m\otimes n, k}\circ  \underline{(m \otimes n)} l'' \circ \phi^{-1}_{m,n} \circ\mact l' \circ l}{r \circ \mact r' \circ \phi_{m,n} \circ \underline{(m \otimes n)} r'' \circ \phi_{m \otimes n, k}} \\
= & \qquad\mbox{(Definition of composition)} \\
  & \optic{\phi^{-1}_{m,n} \circ\mact l' \circ l}{r \circ \mact r' \circ \phi_{m,n}} \circ \optic{l''}{r''} \\
= & \qquad\mbox{(Definition of composition)} \\
  & \optic{l}{r} \circ \optic{l'}{r'} \circ \optic{l''}{r''}
\end{align*}

\section{Lenses and prisms}
\label{sec:orgbd71d99}
The next step is to show that Definition \ref{org8ab3812} actually captures
our motivating examples.  We will need to make use of the Yoneda lemma
to translate between the two forms of the optic.  We go from the form of the optic 
that the definition prescribes (the \emph{existential optic}) to their description 
after eliminating the coend  (the \emph{concrete optic}).

\begin{proposition}[\cite{milewski17}]
\label{prop:lenses}
Lenses are optics for the cartesian product.
\[\begin{pmatrix}\includegraphics[width=0.3\linewidth]{./images/lens.png}\end{pmatrix}
\cong
\begin{pmatrix}\includegraphics[width=0.3\linewidth]{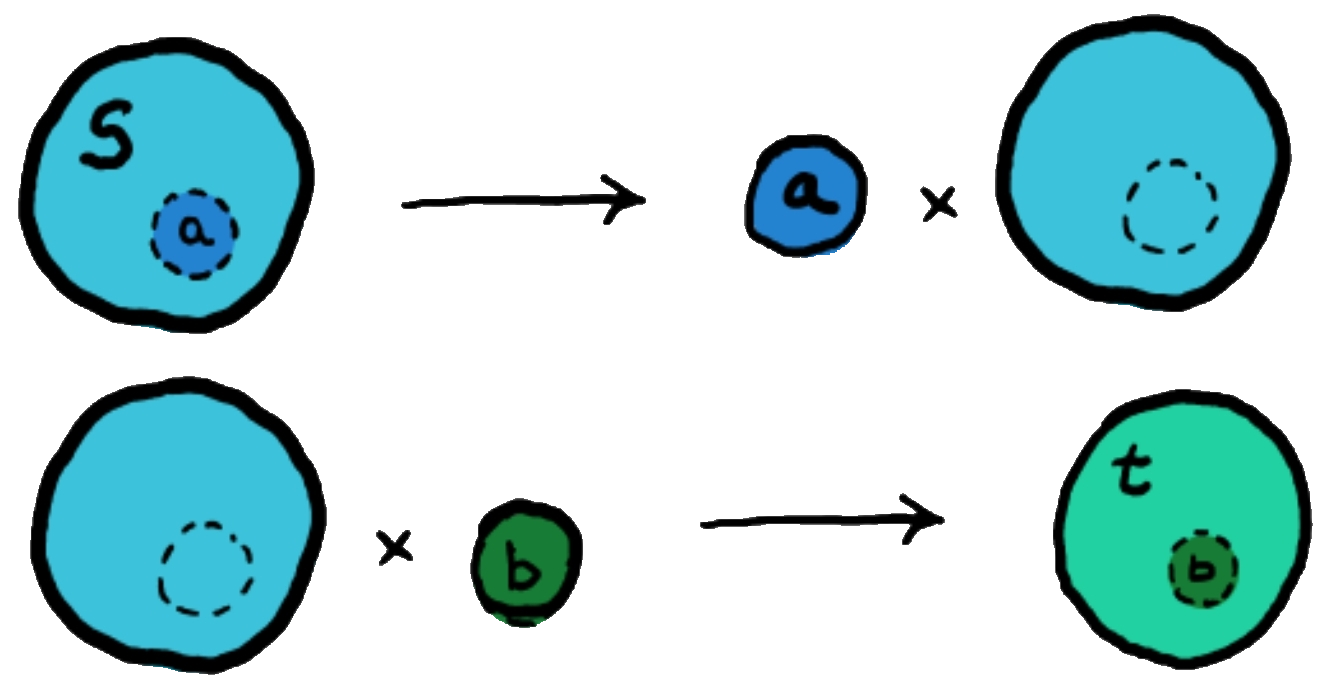}\end{pmatrix}\]
\end{proposition}
\begin{proof}
\[\begin{aligned}
& \int^{c \in \C} \C(s , c \times a) \times \C(c \times b , t) \\
\cong & \qquad\mbox{(Product)} \\
& \int^{c \in \C} \C(s , c) \times \C(s , a) \times \C(c \times b , t) \\
\cong & \qquad\mbox{(Yoneda lemma)} \\
& \C(s , a) \times \C(s \times b , t). \qedhere
\end{aligned}\]
\end{proof}

\begin{proposition}[\cite{milewski17}]
\label{prop:prisms}
Dually, prisms are optics for the coproduct.
\[\begin{pmatrix}\includegraphics[width=0.3\linewidth]{./images/prism.png}\end{pmatrix}
\cong
\begin{pmatrix}\includegraphics[width=0.3\linewidth]{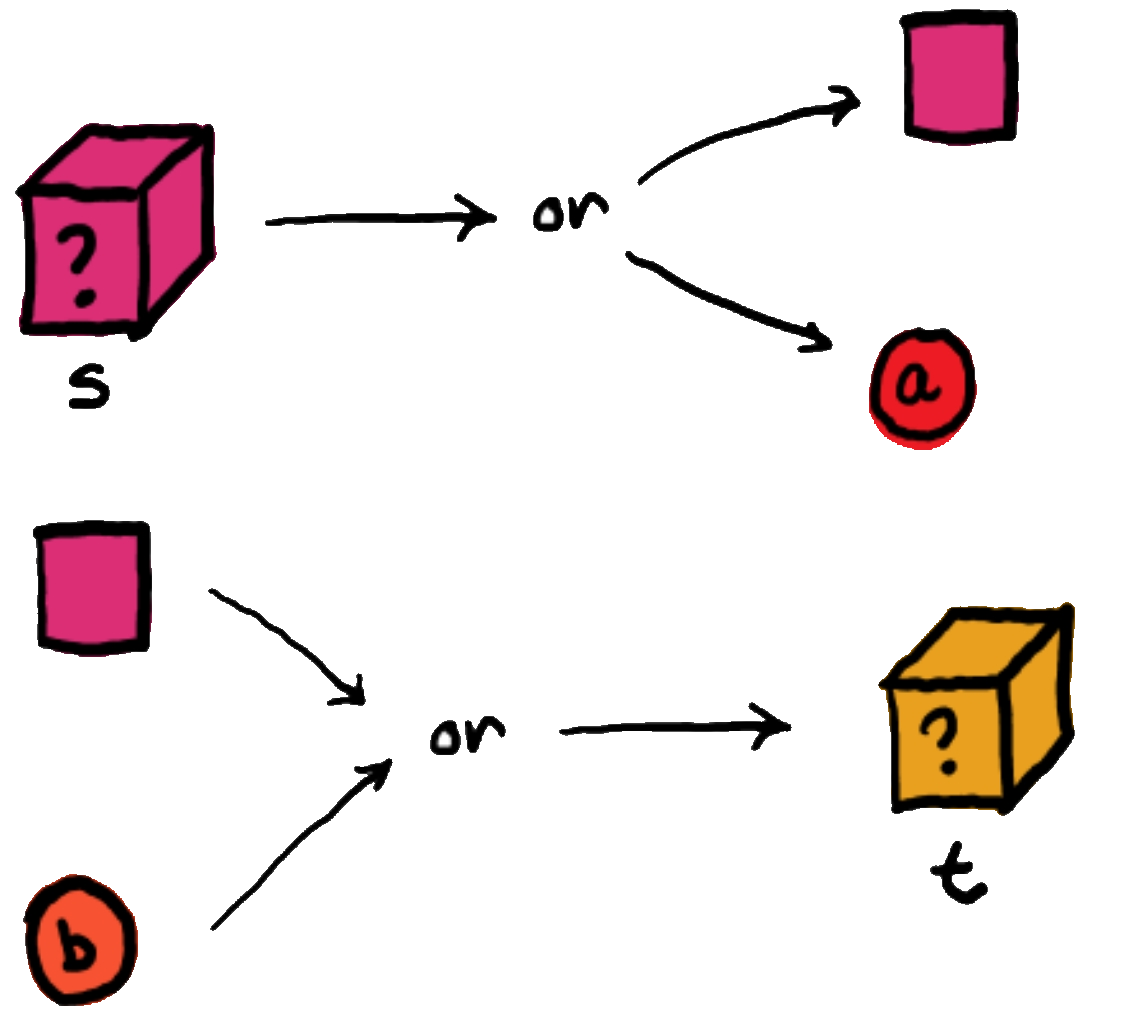}\end{pmatrix}\]
\end{proposition}
\begin{proof}
\[\begin{aligned}
& \int^{m \in \C} \C(s , m + a) \times \C(m + b , t) \\
\cong &\qquad\mbox{(Coproduct)} \\
& \int^{m \in \C} \C(s , m + a) \times \C(m , t) \times \C(m + b , t) \\
\cong &\qquad\mbox{(Yoneda lemma)} \\
& \C(s , t + a) \times \C(b , t). \qedhere
\end{aligned}\]
\end{proof}

\section{Traversals}
\label{sec:org03e5f83}
Given some functor \(c \in [\mathbb{N}, \C]\) from the discrete category of the natural numbers,
we can define a \emph{power series} functor \(F \colon \C \to \C\)
given by \(F(a) = \sum\nolimits_{n \in \mathbb{N}} c_n \times a^n\).  This induces a monoidal action that we call
\(\mathrm{Series} \colon [ \mathbb{N} , \C ] \to [ \C , \C ]\).  The monoidal product for this action
is given by substitution of morphisms, and it corresponds to the fact that two power series
\(F(a) = \sum\nolimits_{n \in \mathbb{N}} c_n \times a^n\) and \(G(a) = \sum\nolimits_{n \in \mathbb{N}} d_n \times a^n\) can be composed
into \(GF(a) = \sum\nolimits_{n \in \mathbb{N}} d_n  \times \left(  \sum\nolimits_{m \in \mathbb{N}} c_m \times a^m \right)^n\), which can be seen again as
a power series functor.  A detailed description of this substitution can be
found in \cite{kock09} or \cite{yorgey14} for the case of general combinatorial
species.

\begin{proposition}
\label{prop:traversal}
Traversals are optics for the action \(\mathrm{Series} \colon [ \mathbb{N} , \C ] \to [ \C , \C ]\) given by
evaluation of the power series functor.
\[\begin{pmatrix}\includegraphics[width=0.4\linewidth]{./images/traversal.png}\end{pmatrix}
\cong
\begin{pmatrix}\includegraphics[width=0.4\linewidth]{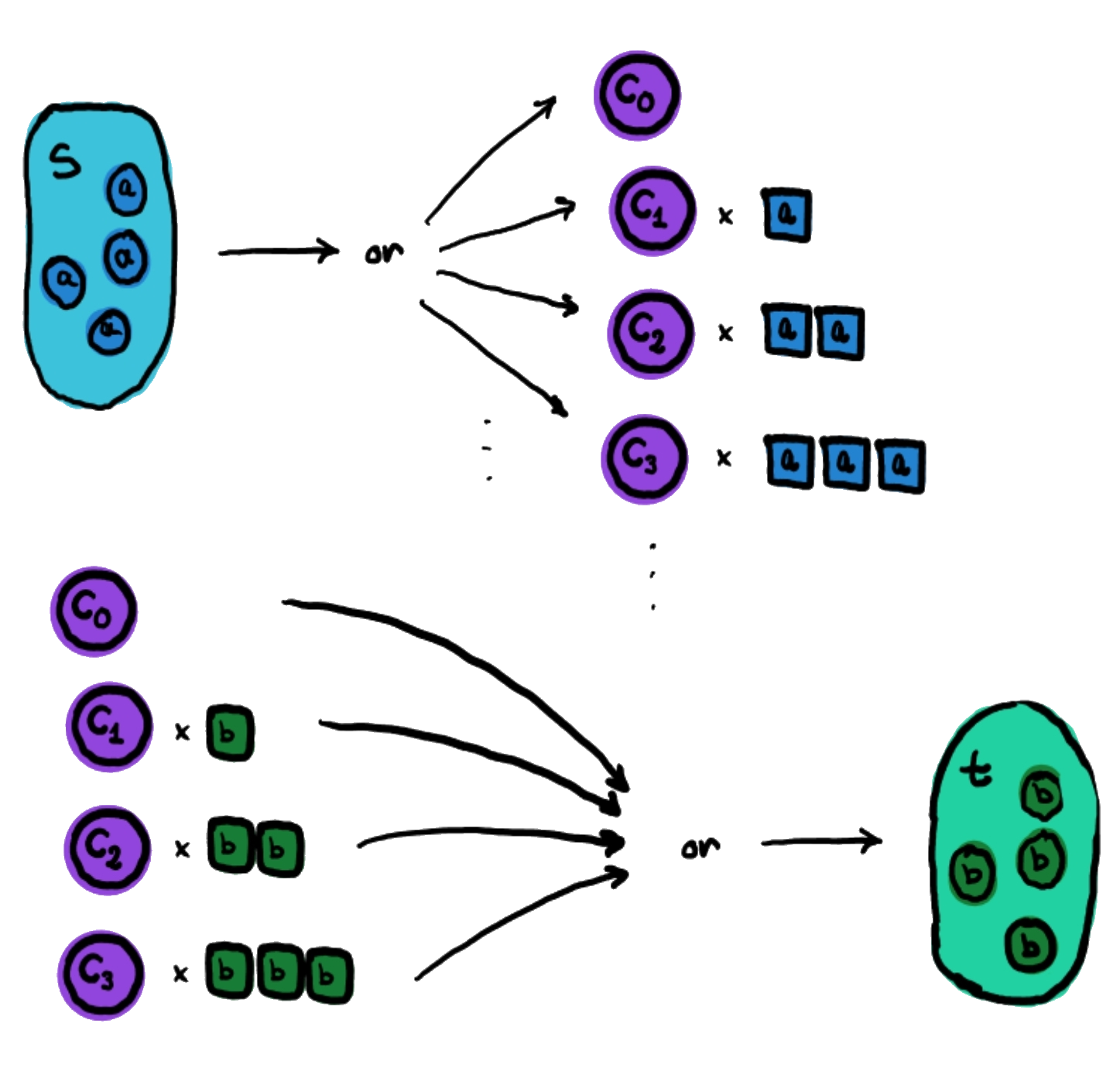}\end{pmatrix}\]
\end{proposition}
\begin{proof}
Unfolding the definitions, the formula we want to prove is the
following one.
\[\int^{c \in [ \mathbb{N} , \C]} 
\C \left(  s , \sum\nolimits_n c_n \times a^n \right) \times
 \C\left(\sum\nolimits_n c_n \times b^n , t \right)
\cong
\C \left( s , \sum\nolimits_n a^n \times (b^n \to t) \right).\]
This is Yoneda, this time for functors \(c \colon \mathbb{N} \to \C\). Note that,
because we are taking the discrete category of the natural numbers, ends over this category
are products, and a natural transformation \([ \mathbb{N} , \C ](F,G)\) can be written as \(\prod\nolimits_{n \in \mathbb{N}} \C(F(n) , G(n))\).
\begin{align*}
& \int^{c} \C \left(  s , \sum_{n \in \mathbb{N}} c_n \times a^n \right) \times \C\left( \sum_{n \in \mathbb{N}} c_n \times b^n , t\right) \\
\cong &\qquad\mbox{{(Cocontinuity)}}\\
& \int^{c} \C \left(  s , \sum_{n \in \mathbb{N}} c_n \times a^n \right) \times \prod_{n \in \mathbb{N}} \C\left( c_n \times b^n , t\right) \\
\cong & \qquad\mbox{{(Exponential)}}\\
& \int^{c} \C \left(  s , \sum_{n \in \mathbb{N}} c_n \times a^n \right) \times \prod_{n \in \mathbb{N}} \C\left( c_n , b^n \to t\right) \\
\cong & \qquad\mbox{{(Natural transformation as an end)}}\\
& \int^{c} \C \left(  s , \sum_{n \in \mathbb{N}} c_n \times a^n \right) \times [ \mathbb{N} , \C ] \left( c_{(-)} , b^{(-)} \to t\right) \\
\cong & \qquad\mbox{{(Yoneda lemma)}}\\
& \C \left(  s , \sum_{n \in \mathbb{N}} a^n \times (b^n \to t) \right). & \qedhere
\end{align*}
\end{proof}

This derivation solves the problem posed in \cite{milewski17} of finding
a derivation of the Traversal fitting the same elementary pattern as
the other optics described there.  It should be noted, however, that derivations of the traversal
as the optic for a certain kind of functors called \textbf{Traversables} (which should
not be confused with traversals themselves) have been previously described
by \cite{boisseau18} and \cite{riley18}.  For a derivation using Yoneda, \cite{riley18}
recalls a parameterised adjunction that has an equational proof in \cite{jaskelioff15}.
The \emph{Traversable} characterization is the one that was known and commonly used in programming libraries.
This characterization in terms of \emph{power series polynomials} could be considered
a more elementary description, although the profunctor description can give problems when implementing
it in languages such as Haskell with only partial support for families of types
indexed by natural numbers.  The concrete description can be implemented as
a \emph{nested datatype} \cite{bird98}.

\section{More examples of optics}
\label{sec:orgf3e3cb1}
\label{org087c8c4}
\subsection{Grates}
\label{sec:org974f2f8}
In this section, we assume we are working with a bicartesian closed category \(\C\)
instead of detailing the precise requisites that we would need to make each one
of these optics definable in the category.

\begin{proposition}[\cite{milewski17}]
\label{prop:grates}
Grates are optics for the action of the exponential \(( \to ) \colon \C^{op} \to [ \C , \C ]\).
\end{proposition}
\begin{proof}
\begin{align*}
& \int^{c \in \C^{op}} \C(s , c \to a) \times \C(c \to b , t) \\
\cong & \qquad\mbox{(Exponential)} \\
& \int^{c \in \C^{op}} \C(c , s \to a) \times \C(c \to b , t) \\
\cong & \qquad\mbox{(Yoneda lemma)} \\
& \C((s \to a) \to b, t).& \qedhere
\end{align*}
\end{proof}

\subsection{Achromatic lenses}
\label{sec:org4654c07}
\begin{proposition}
\label{prop:achromatic}
\textbf{Achromatic lenses} (described by \cite{boisseau17}) are optics for the
action \((1 + (-)) \times (-) \colon \C \times \C \to \C\). They have a concrete description
\[ \mathbf{AchrLens}
\left( \begin{pmatrix} a \\ b \end{pmatrix}, \begin{pmatrix} s \\ t \end{pmatrix} \right) 
= \C(s , (b \to t) + 1) \times \C(s,a) \times \C(b,t).
\]
\end{proposition}
\begin{proof}
Again applying the Yoneda lemma.
\begin{align*}
& \int^{c \in \C} \C(s , (c + 1) \times a) \times \C((c + 1) \times b , t) \\
\cong & \qquad\mbox{(Product)} \\
& \int^{c \in \C} \C(s ,c + 1) \times \C(s , a) \times \C((c + 1) \times b , t) \\
\cong & \qquad\mbox{(Distributivity)}\\
& \int^{c \in \C} \C(s ,c + 1) \times \C(s , a) \times \C(c \times b + b , t) \\
\cong & \qquad\mbox{(Coproduct)} \\
& \int^{c \in \C} \C(s ,c + 1) \times \C(s , a) \times \C(b , t) \times \C(c\times b , t) \\
\cong & \qquad\mbox{(Exponential)} \\
& \int^{c \in \C} \C(s ,c + 1) \times \C(s , a) \times \C(b , t) \times \C(c , b \to t) \\
\cong & \qquad\mbox{(Yoneda lemma)} \\
& \C(s , (b \to t) + 1) \times (s \to a) \times (b \to t). & \qedhere
\end{align*}
\end{proof}

\subsection{Kaleidoscopes}
\label{sec:org86f8912}
\begin{proposition}
\label{prop:kaleidoscope}
\textbf{Kaleidoscopes} are optics for the evaluation of applicative
functors (Definition \ref{org478a9a6}), \(\mathbf{App} \to [ \Sets , \Sets ]\). They have a concrete description
\[\mathbf{Kaleidoscope}
\left( \begin{pmatrix} a \\ b \end{pmatrix}, \begin{pmatrix} s \\ t \end{pmatrix} \right) =
\prod_n \Sets\left(s^n \times (a^n \to b),t \right).
\]

\begin{figure}[htbp]
\centering
\includegraphics[width=8cm]{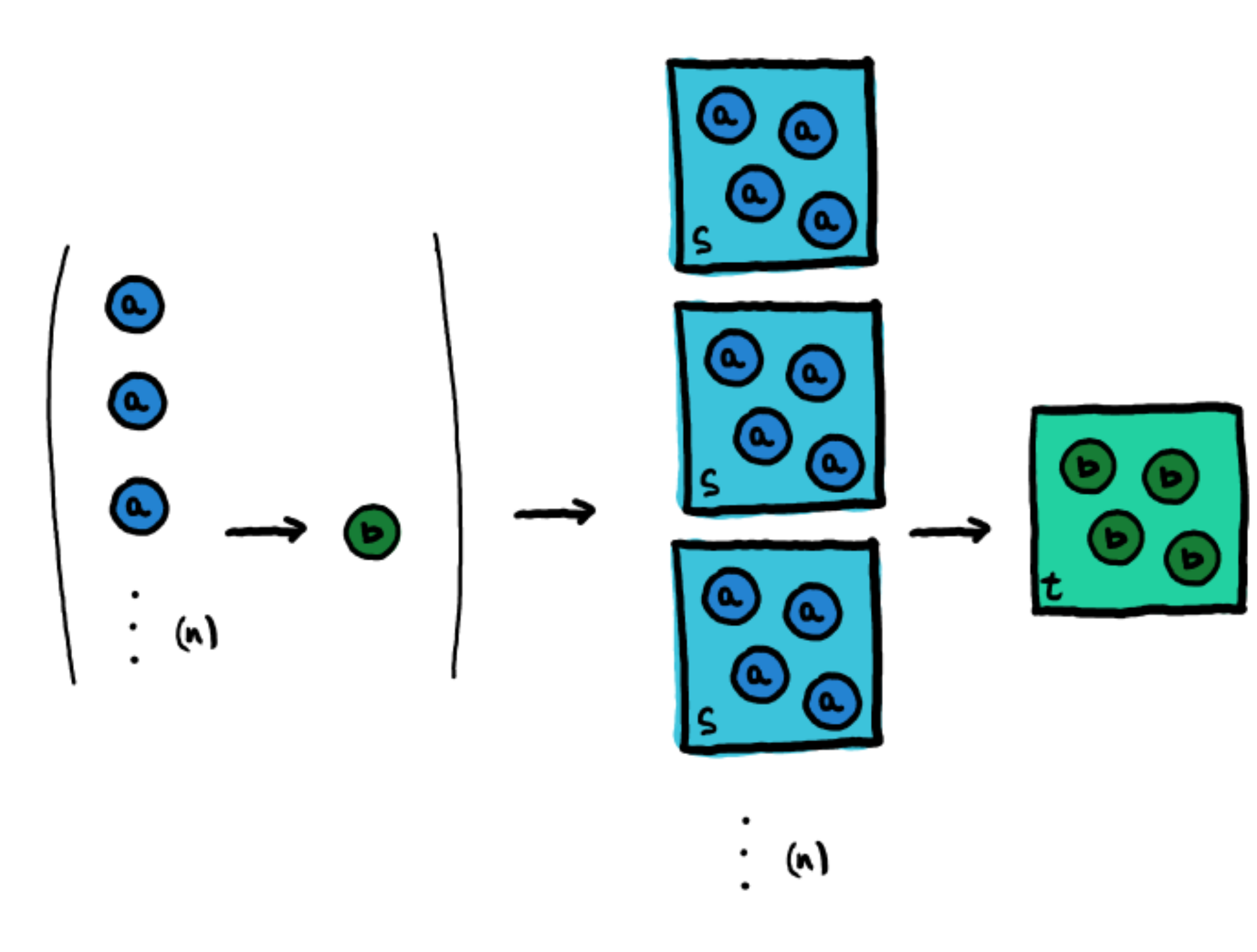}
\caption{\label{fig:org4553fac}
A big data structure \(s\) contains many substructures of type \(a\); a way of folding them \((a^n \to b)\) gives a way of folding the big data structure \((s^n \to t)\).}
\end{figure}
\end{proposition}
\begin{proof}
We will make use of the construction of free applicatives given in
Corollary \ref{org79dc269}. Note that we are implicitly applying a forgetful
functor over the applicative \(F\).
\begin{align*}
& \int^{F \in \mathbf{App}} \Sets(s, Fa) \times \Sets(Fb,t)
\\ \cong & \qquad\mbox{{(Yoneda lemma)}} \\
& \int^{F \in \mathbf{App}}  \Sets\left(s, \int_{c} (a \to c) \to Fc \right) \times \Sets(Fb,t) 
\\ \cong & \qquad\mbox{(Continuity)}\\
& \int^{F \in \mathbf{App}} \left( \int_{c}  \Sets\left(s, (a \to c) \to Fc \right) \right) \times \Sets(Fb,t) 
\\ \cong & \qquad\mbox{(Exponential)}\\
& \int^{F \in \mathbf{App}} \left( \int_{c} \Sets\left(s \times (a \to c),  Fc \right) \right) \times \Sets(Fb,t) 
\\ \cong & \qquad\mbox{(Natural transformations as ends)}\\
& \int^{F \in \mathbf{App}} \mathrm{Nat}\left(  s \times (a \to (-)) , F \right) \times \Sets(Fb,t) 
\\ \cong & \qquad\mbox{(Free-forgetful adjunction for applicative functors)} \\
& \int^{F \in \mathbf{App}} \mathbf{App}\left(  \sum_n s^n \times \left(  a^n \to (-) \right) , F \right) \times \Sets(Fb,t) 
\\ \cong & \qquad\mbox{(Yoneda lemma)} \\
& \Sets \left(  \sum_n s^n \times (a^n \to b),t \right).
\end{align*}
In \(\Sets\), we write this as \(\prod\nolimits_n (a^n \to b) \to (s^n \to t)\).
\end{proof}

\subsection{Algebraic lenses}
\label{sec:orgcb14fa0}
\begin{remark}
\label{org2930c80}
The action that defines lenses is not included in the action that
gives kaleidoscopes because not every product by 
a set gives a lax monoidal functor.  As we will see in
\S \ref{org7bbc905}, this implies that not every lens induces a kaleidoscope. 
However, when the object \(c\) is a monoid, the unit and
multiplication induce functions \(1 \to c \times 1\) and \((c \times a)\times (c \times b) \to c \times (a \times b)\),
making products a particular example of applicative functor.
This observation inspires the following optic.
\end{remark}

\begin{proposition}
\label{prop:algebraiclens}
Let \(\psi\) be a monad in a category \(\C\).  We consider the
action of its algebras \(\psi\mbox{-Alg}\) given by
forgetting about the algebra structure and taking the cartesian
product.  We know that the product of two algebras has again algebra
structure and that the terminal object has an algebra structure: the
forgetful functor from the Eilenberg-Moore category 
\(U \colon \psi\mbox{-Alg} \to \C\) creates all limits that exist in \(\C\).
\end{proposition}
\begin{proof}
This gives a concrete optic we can call \textbf{algebraic lens}.  Note that these
are different from the \emph{monadic lenses} studied in \cite{abou16}.
\begin{align*}
& \int^{c \in \psi\mbox{-Alg}} \C(s, c \times a) \times \C(c \times b,t) 
\\\cong & \qquad\mbox{{(Product)}} \\
& \int^{c \in \psi\mbox{-Alg}} \C(s, c) \times \C(s, a) \times \C(c \times b,t) 
\\\cong & \qquad\mbox{{(Free-forgetful adjunction for the algebras)}} \\
& \int^{c \in \psi\mbox{-Alg}} \mathbf{\psi\mbox{-Alg}}(\psi s, c) \times \C(s, a) \times \C(c \times b,t) 
\\ \cong & \qquad\mbox{{(Yoneda lemma)}} \\
& \C(s, a) \times \C(\psi s \times b,t). & \qedhere \\
\end{align*}
\end{proof}

In particular, taking \(\psi\) to be the \emph{list} monad makes Remark \ref{org2930c80}
work appropiately.  We call \(\mathbf{ListLens}\) to this particular case of algebraic
lens.  \emph{Coalgebraic prisms} work in exactly the same way.

\subsection{Setters and adapters}
\label{sec:org1af143e}
\begin{proposition}
\label{prop:setters}
For the case of setters, we will work in the category of \(\Sets\).
\textbf{Setters} are optics for the action given by evaluation of any endofunctor
\(\mathrm{ev} \colon [\mathbf{Sets} , \mathbf{Sets}] \to [\mathbf{Sets} , \mathbf{Sets}]\).  They have a concrete description
\[
\mathbf{Setter}
\left( \begin{pmatrix} a \\ b \end{pmatrix}, \begin{pmatrix} s \\ t \end{pmatrix} \right) = 
\mathbf{Sets}(a \to b , s \to t).
\]
\end{proposition}
\begin{proof}
\begin{align*}
& \int^{F \in [ \mathbf{Sets} , \mathbf{Sets} ]} \mathbf{Sets}(s, F(a)) \times \mathbf{Sets}(F(b) , t) \\
\cong & \qquad\mbox{{(Yoneda lemma)}} \\
& \int^{F \in [ \mathbf{Sets} , \mathbf{Sets} ]} \mathbf{Sets}\left(  s,  \int_{c} ((a \to c) \to Fc\right) \times \mathbf{Sets}(F(b) , t) \\
\cong & \qquad\mbox{{(Continuity)}} \\
& \int^{F \in [ \mathbf{Sets} , \mathbf{Sets} ]} \int_{c} \mathbf{Sets}\left(  s \times (a \to c), Fc \right) \times \mathbf{Sets}(F(b) , t) \\
\cong & \qquad\mbox{{(Yoneda lemma)}} \\
& \int^{F \in [ \mathbf{Sets} , \mathbf{Sets} ]} \mathrm{Nat}\left(  s \times (a \to \square), F \right) \times \mathbf{Sets}(F(b) , t) \\
\cong & \qquad\mbox{{(Natural transformation as an end)}} \\
& \mathbf{Sets}(s \times (a \to b) , t). && \qedhere \\
\end{align*}
\end{proof}

In \cite{riley18}, a similar derivation is given but in the more general
case where we only ask our category to be powered and copowered over \(\mathbf{Sets}\).
\begin{proposition}
\label{prop:adapters}
\textbf{Adapters} are optics for the single action of the identity functor.
By definition, they have a concrete description
\[
\mathbf{Adapters}
\left( \begin{pmatrix} a \\ b \end{pmatrix}, \begin{pmatrix} s \\ t \end{pmatrix} \right) = 
\mathbf{C}(s , a) \times \mathbf{C}(b,t).
\]
\end{proposition}

\subsection{Generalized lens}
\label{sec:orgd39ada6}
The reader will notice that even when we were allowing the two parts
of the optic to live in different categories, we are not using the construction
in that generality.  There are not many examples where \emph{mixed} optics have an
application, but we will show one example.  The following generalization
of the concept of lens was described to the author by David J. Myers \cite{david}, and
similar generalizations have been described in \cite[\S2.2]{spivak19}. They
were described for the study of autopoietic systems and more generally,
as a definition of lens that can work in a huge variety of different categories.

\begin{definition}
Let \(\C\) be a symmetric monoidal category. A \textbf{generalized lens} in this monoidal
category is a mixed optic for the action of the tensor product of cocommutative comonoids both on
the (symmetric monoidal) category of cocommutative comonoids in \(\C\), that
we call \(\mathbf{Comon}\); and in \(\C\) itself.  They have a concrete description
given by the following formula, where we call \({\cal U} \colon \mathbf{Comon} \to \mathbf{C}\) to the forgetful functor.
\[ \mathbf{GeneralizedLens}
\left( \begin{pmatrix} a \\ b \end{pmatrix}, \begin{pmatrix} s \\ t \end{pmatrix} \right) =
\mathbf{Comon}(s, a) \times \C({\cal{U}}s \otimes b, t)
\]
\end{definition}
\begin{proof}
The main result we need to use is that a comonoid homomorphism between cocommutative comonoids \(s \to c \otimes a\) can
be split uniquely as the monoidal product of two comonoid homomorphisms
\(s \to c\) and \(s \to a\).
\begin{align*}
& \int_{c \in \mathbf{Comon}} \mathbf{Comon}(s, c \otimes a) \times \mathbf{C}({\cal U} c \otimes b , t) \\
\cong & \qquad\mbox{(Split of the comonoid morphism)} \\
& \int_{c \in \mathbf{Comon}} \mathbf{Comon}(s, c) \times \mathbf{Comon}(s , a) \times \mathbf{C}({\cal U} c \otimes b , t) \\
\cong & \qquad\mbox{(Yoneda lemma)} \\
& \mathbf{Comon}(s , a) \times \mathbf{C}({\cal U} s \otimes b , t). & \qedhere \\
\end{align*}
\end{proof}

\subsection{Optics for (co)free}
\label{sec:org64be483}
\label{opticscofree}

\begin{remark}
In the derivation of the concrete Kaleidoscope (Proposition \ref{prop:kaleidoscope})
we have only used the fact that we can generate free applicative functors.
On the other hand, in the derivation of the Traversal on the next chapter
(Proposition \ref{prop:traversal2}) we only use the fact that we can
generate cofree traversable functors.  These two observations can be generalized
into a class of concrete optics.  For a different but similar class of optics and their
laws, see \cite[\S4.4]{riley18}.
\end{remark}

We start by considering the following two functors for some fixed
\(s,a \in \Sets\) and some fixed \(b,t \in \Sets\).
\[
L_{s,a} = s \times (a \to (-)), \qquad
R_{t,b} = ((-) \to b) \to t.
\]
The names come from their similarity to left and right Kan extensions,
which will be justified from the fact that they arise from (co)Yoneda
reductions.  In fact, the property that is interesting to us is that
for any given functor \(H \colon \Sets \to \Sets\), the following isomorphisms hold.
\[
\Sets(s,Ha) \cong [\Sets,\Sets](L_{s,a},H),\qquad
\Sets(Hb,t) \cong [\Sets,\Sets](H,R_{t,b}).
\]
We can prove this again using the Yoneda lemma.
\begin{align*}
& \mathbf{Sets}(s,Ha) && \mathbf{Sets}(Hb,t) \\
\cong & \qquad\mbox{(Yoneda lemma)} &\cong&\qquad\mbox{(Coyoneda lemma)} \\
& \Sets\left(  s, \int_{c \in \Sets} (a \to c) \to Hc \right) && \Sets \left(  \int^{c \in \Sets} Hc \times (c \to b) , t \right)\\
\cong & \qquad\mbox{(Continuity)} &\cong&\qquad\mbox{(Continuity)} \\
& \int_{c \in \Sets} \Sets\left(  s,  (a \to c) \to Hc \right) && \int_{c \in \Sets}\Sets \left( Hc \times (c \to b) , t \right)\\
\cong & \qquad\mbox{(Exponential)} & \cong & \qquad\mbox{(Exponential)} \\
& \int_{c \in \Sets} \Sets\left(  s \times (a \to c) , Hc \right) && \int_{c \in \Sets}\Sets \left( Hc  , (c \to b) \to t \right)\\
\cong & \qquad\mbox{(Natural transformations as ends)} &\cong & \qquad\mbox{(Natural transformations as ends)}  \\
& [\Sets, \Sets]\left( L_{s,a} , H \right). && [\Sets,\Sets]\left( H,R_{b,t} \right). \\
\end{align*}

\begin{proposition}
\label{orgb2fe82b}
Let a monoidal action \(U \colon \M \to [\Sets,\Sets]\) have a left adjoint given
by some \(F \colon [\Sets,\Sets] \to \M\), that is, \([\Sets,\Sets](f,Ug) \cong \M(Fg,f)\). 
The optic determined by that monoidal action has a concrete form 
given by \(\Sets\left( UFL_{s,a}(b) , t\right)\).  Dually, let it have a right adjoint
given by some \(G \colon [\Sets, \Sets] \to \M\), that is, \([\Sets,\Sets](Uf,g) \cong \M(g,Gf)\).
The optic determined by that monoidal action has then a concrete form
given by \(\Sets(s,UGR_{b,t}(a))\).
\end{proposition}
\begin{proof}
\begin{align*}
& \int^{f \in \M} \Sets(s, Uf(a)) \times \Sets(Uf(b),t) && \int^{g \in \M} \Sets(s, Ug(a)) \times \Sets(Ug(b),t)
\\ \cong & \qquad\mbox{(Definition of $L_{s,a}$)} & \cong & \qquad\mbox{(Definition of $R_{b,t}$)}\\
& \int^{f \in \M} \Nat\left(  L_{s,a} , Uf \right) \times \Sets(Uf(b),t) && \int^{f \in \M} \Sets(s,Ug(a))  \times \Nat \left( Ug, R_{b,t} \right)
\\ \cong & \qquad\mbox{(Adjunction)} & \cong & \qquad\mbox{(Adjunction)}\\
& \int^{f \in \M} \M\left(  FL_{s,a} , f \right) \times \Sets(Uf(b),t) && \int^{f \in \M} \Sets(s,Ug(a)) \times \M\left( g , GR_{b,t} \right)
\\ \cong & \qquad\mbox{(Yoneda lemma)} & \cong & \qquad\mbox{(Yoneda lemma)} \\
& \Sets \left( UFL_{s,a} (b) ,t \right). && \Sets \left( s , UGR_{b,t}a \right). & \qedhere
\end{align*}
\end{proof}

\chapter{Traversals}
\label{sec:org6b5a11d}
\label{org598c0f5}
We have characterized traversals as the optic for \emph{power series}
functors in Proposition \ref{prop:traversal}. However, the result that is usually presented and used in
programming libraries is that traversals are the optic for \emph{traversable functors}.
Recall that we are taking power series functors to be these that can be written
as \(T(a) = \sum\nolimits_n c_n \times a^n\) for some \(c \colon \mathbb{N} \to \Sets\); while a functor \(T\) will be
traversable if it has a distributive law \(TF \tonat FT\) for every applicative functor
\(F \in \mathbf{App}\) (Definition \ref{org478a9a6}).  We will show that both give rise to the optic called \emph{traversal}.

\section{Traversables as coalgebras}
\label{sec:org7bbbd12}
\subsection{Traversable functors}
\label{sec:orgb60444e}

\begin{definition}[\cite{rypacek12}]
A \textbf{traversable} structure on a functor \(T \colon \C \to \C\) is a family
of transformations \(\trv_F \colon TF \tonat FT\) satisfying three
additional rules called \emph{naturality}, \emph{unitarity} and \emph{linearity;} which
can be expressed respectively as the commutativity of the following diagrams,
where \(\alpha \colon F \tonat G\) is a morphism of applicative functors.
\[\begin{tikzcd}[row sep=tiny]
TF \rar{\trv_F} \ar{dd}[swap]{T \alpha} & FT\ar{dd}{\alpha}  & &  & TFG \ar{ddr}[swap]{\trv_{FG}} \rar{\trv_F} & FTG \ar{dd}{F\trv_G} \\
&& T \rar[bend left]{\trv_1} \rar[bend right,swap]{\id} & T &     &     \\
TG \rar{\trv_G} & GT  &                                     &     & & FGT
\end{tikzcd}\]
\end{definition}
\begin{remark}
The first of the three rules is equivalent to a dinaturality condition
over an end.  With this in mind, we can define the traversable structure
to be given by \(\int\nolimits_{F \in \mathbf{App}} (TF \tonat FT)\) instead.  Because of the definition
of right Kan extensions, we can rewrite this as \(T \tonat \int\nolimits_{F \in \mathbf{App}} \Ran_FFT\).
This motivates our study of this particular end.
\end{remark}
We will characterize traversable functors as coalgebras; the first step will
be to simplify the end \(\int\nolimits_{F \in \mathbf{App}} \Ran_FFT\). 
The following lemma relies on the construction of a particular
free applicative functor from Corollary \ref{org79dc269}.  Recall that the
free applicative functor over
\((a \times (b \to -))\) is precisely
\[
(a \times (b \to -))^{\ast} \cong \sum_{n \in \mathbb{N}} a^n \times (b^n \to -).
\]

\begin{lemma}
\label{org8b21ea5}
There exists an isomorphism with the following signature, natural
on \(T \colon \Sets \to \Sets\).
\[\int_{F \in \mathbf{App}} \mathsf{Ran}_FFT(a) \cong \sum_{n \in \mathbb{N}} a^n \times T(n).\]
\end{lemma}
\begin{proof}
Note for this derivation that coproducts commute with connected limits (see for instance \cite{nlab}).  
Intuitively, if we need to choose a number for
each set and it has to be preserved by morphisms, it needs to be constant.
\begin{align*}
&\int_{F \in \mathbf{App}} \mathsf{Ran}_FFT(a) \\
\cong & \qquad\mbox{(Formula for a right Kan extension)} \\
& \int_{F \in \mathbf{App}} \int_{b \in \Sets} (a \to Fb) \to FT(b) \\
\cong & \qquad\mbox{(Yoneda lemma)} \\
& \int_{F \in \mathbf{App}} \int_{b \in \Sets} \left(a \to \int_{c \in \Sets}(b \to c) \to  Fc\right) \to FT(b) \\
\cong & \qquad\mbox{(Continuity)} \\
& \int_{F \in \mathbf{App}} \int_{b \in \Sets} \left(\int_{c \in \Sets} a \to (b \to c) \to  Fc\right) \to FT(b) \\
\cong & \qquad\mbox{(Currying)} \\
& \int_{F \in \mathbf{App}} \int_{b \in \Sets} \left(\int_{c \in \Sets} a \times (b \to c) \to  Fc\right) \to FT(b) \\
\cong & \qquad\mbox{(Natural transformation as an end)} \\
& \int_{F \in \mathbf{App}} \int_{b \in \Sets} \mathrm{Nat}\left( a \times (b \to -) ,  F\right) \to FT(b) \\
\cong & \qquad\mbox{(Free-forgetful adjunction for applicative functors)} \\
& \int_{F \in \mathbf{App}} \int_{b \in \Sets} \mathbf{App}\left( (a \times (b \to -))^{\ast} ,  F\right) \to FT(b) \\
\cong & \qquad\mbox{(Free applicative functor)} \\
& \int_{F \in \mathbf{App}} \int_{b \in \Sets} \mathbf{App}\left( \sum_{n \in \mathbf{N}} a^n \times (b^n \to -) ,  F\right) \to FT(b) \\
\cong & \qquad\mbox{(Fubini)} \\
&  \int_{b \in \Sets} \int_{F \in \mathbf{App}} \mathbf{App}\left( \sum_{n \in \mathbf{N}} a^n \times (b^n \to -) ,  F\right) \to FT(b) \\
\cong & \qquad\mbox{(Yoneda lemma)} \\
&  \int_{b \in \Sets} \sum_{n \in \mathbf{N}} a^n \times (b^n \to T(b)) \\
\cong & \qquad\mbox{(Ends distribute over discrete colimits)} \\
&  \sum_{n \in \mathbf{N}} \int_{b \in \Sets} a^n \times (b^n \to T(b)) \\
\cong & \qquad\mbox{(Fubini, as in Remark \ref{remark-limits-commute})} \\
&  \sum_{n \in \mathbf{N}} \left(\int_{b \in \Sets} a\right)^n \times \left(  \int_{b \in \Sets} b^n \to T(b) \right)\\ 
\cong & \qquad\mbox{(Connected end over a constant functor)} \\
&  \sum_{n \in \mathbf{N}} a^n \times \left(  \int_{b \in \Sets} b^n \to T(b) \right)\\ 
\cong & \qquad\mbox{(Exponential as function from a finite set)} \\
& \sum_{n \in \mathbf{N}} a^n \times \left( \int_{b \in \Sets} (n \to b) \to T(b) \right) \\
\cong & \qquad\mbox{(Yoneda lemma)} \\
&  \sum_{n \in \mathbf{N}} a^n \times T(n). & \qedhere
\end{align*}
\end{proof}

\subsection{The shape-contents comonad}
\label{sec:org44febd2}
\label{orgab52623}
We will be studying the following higher-order functor \(K \colon [\Sets,\Sets] \to [\Sets,\Sets]\)
defined as \(KT(a) = \sum_{n \in \mathbb{N}} T(n) \times a^n\). It is meant to represent a split between the
shape and the contents of \(T\), regarded as a container.  
Because of this, we write
the elements of \(KT(a)\) as \((n ; s , c)\) with \(n \in \mathbb{N}\) the \emph{length}, \(s \in T(n)\) the \emph{shape},
and \(c \in a^n\) the \emph{contents}.  The inspiration comes from \cite{gibbons09}, which mentions
how traversables provide this kind of shape-contents split, studied in \cite{jay94};
our goal is to show that it is
precisely what characterizes them. One can see that, actually, the elements of \(T(n)\)
are more than the valid shapes; indexes could be repeated or not even present at all. We claim,
however, that the coalgebra axioms are enough to ensure a valid \emph{shape-contents} split.
\begin{proposition}
There exists a functor \(K \colon [\Sets,\Sets] \to [\Sets,\Sets]\) defined on objects
by \(KT = \sum_{n \in \mathbb{N}} T(n) \times a^n\) that can be given a comonad structure.
\end{proposition}
\begin{proof}
First, we check that it is indeed a well-defined functor. If we want it to
be well-defined on objects, we need to check that \(\sum\nolimits_{n \in \mathbb{N}} T(n) \times a^n\) is
a functor for any \(T \in [\Sets, \Sets]\). Given \(f \colon a \to b\), the
corresponding \(KT(f) \colon KT(a) \to KT(b)\) is defined by 
\(KT(f)(n; s , c) = (n; s, f \circ c)\). We can see that this is functorial.
Now we need to define its action on morphisms. Given any natural
transformation \(h \colon T \tonat R\), we can define \(Kh(n;s,c) = (n;h(s),c)\), which
is also functorial.

Now we define the \emph{counit} \(\varepsilon \colon KT \tonat T\) as the morphism given by evaluation
\(T(n) \times a^n \to T(a)\) in every possible \(n \in \mathbb{N}\). In other words, \(\varepsilon(n; s,c) = T(c)(s)\).
The \emph{comultiplication} \(\delta \colon KT \tonat K^2T\) is given by the universal property of
the coproduct on the following diagram.
\[\begin{tikzcd}
\sum_n T(n) \times a^n \ar[dashed]{rr}{\exists!} && \sum_m \left(   \sum_l T(l) \times m^l \right) \times a^m  \\
T(n) \times a^n \ar{rr}{(\id,(\id),\id)} \uar{i_n} && T(n) \times n^n \times a^n \uar{i_{n,n}}
\end{tikzcd}\]
That is, we choose both \(m\) and \(l\) to be \(n\) and then we use the identity. In
other words, \(\delta(n;s,c) = (n;(n; s, \id),c)\). 

We will check now counitality and coassociativity.
For \emph{counitality}, we
need the following two diagrams to commute.
\[\begin{tikzcd}
\sum_{n} T(n) \times a^n & \lar[swap]{K\varepsilon_T} \sum_m \left( \sum_l T(l) \times m^l \right) \times a^m 
 \rar{\varepsilon_{KT}} & \sum_n T(n) \times a^n \\
& \sum_n T(n) \times a^n \ular{\id} \uar[swap]{\delta_T} \urar[swap]{\id} &
\end{tikzcd}\]
They do commute because of the following chains of equations that arise from
unfolding the definitions.
\begin{align*}
& K\varepsilon_T(\delta(n;s,c))                           & & \varepsilon_{KT}(\delta(n;s,c)) \\
= & \qquad\mbox{(Definition of $\delta$)}             & = & \qquad\mbox{(Definition of $\delta$)} \\
& K\varepsilon_T(n;(n,s,\id),c)                           & & \varepsilon_{KT}(n;(n;s,\id),c) \\
= & \qquad\mbox{(Definition of $K$)}                  & = & \qquad\mbox{(Definition of $\varepsilon$)} \\
& (n;\varepsilon_T(n,s,\id),c)                            & & KTc(n;s,\id) \\
= & \qquad\mbox{(Definition of $\varepsilon$)}        & = & \qquad\mbox{(Definition of $KT$)} \\
& (n;T(\id)(s),c)                                         & & (n;s,c \circ \id) \\
= & \qquad\mbox{(Identity)}                           & = & \qquad\mbox{(Identity)} \\
& (n;s,c).                                                 & & (n;s,c).
\end{align*}
\emph{Coassociativity} is the fact that the following diagram commutes.
\[\begin{tikzcd}
\sum_n \left( \sum_m  \left( \sum_l T(l) \times m^l \right) \times n^m \right) \times a^n  &
\sum_n \left( \sum_l T(l) \times n^l \right) \times a^n \lar[swap]{K\delta_T}\\
\sum_n \left( \sum_m T(m) \times n^m \right) \times a^n \uar{\delta_{KT}}&
\sum_n T(n) \times a^n \uar{\delta_T} \lar{\delta_T}
\end{tikzcd}\] 
It does because of the following chain of equations, again following the
definitions.
\begin{align*}
& \delta_{KT}(\delta_T(n;s,c)) \\
= &\qquad\mbox{(Definition of $\delta$)} \\
& \delta_{KT}(n;(n;s,\id),c) \\
= &\qquad\mbox{(Definition of $\delta$)} \\
& (n;(n;(n;s,\id),\id),c) \\
= &\qquad\mbox{(Definition of $\delta$)} \\
& (n;\delta_T(n;s,\id),c) \\
= &\qquad\mbox{(Definition of $K$)} \\
& K\delta_T(n;(n;s,\id),c) \\
= &\qquad\mbox{(Definition of $\delta$)} \\
& K\delta_T(\delta_T(n;s,\id)).
\end{align*}
This finishes the construction of a comonad over \(K\).
\end{proof}
The coalgebra axioms follow from the structure, but we will write them explicitly
and comment on them.  Let \(\sigma \colon Ta \to \sum_n Tn \times a^n\) be a coalgebra.
The first axiom is \emph{counitality}, and in our case, it says that the
following diagram commutes.
\[\begin{tikzcd}
Ta \drar[swap]{\id} \rar{\sigma} & \sum_n Tn \times a^n \dar{\varepsilon} \\
   &  Ta 
\end{tikzcd}\] 
When \(\sigma(t) = (n;s,c)\), we have that \(Tc(s) = t\). This is to say that, 
if we split into shape and contents and then we
put back the contents onto the shape, we should get back our original
structure.  The second axiom is \emph{coassociativity}, and in our case, it
says that the following diagram commutes.
\[\begin{tikzcd}
Ta\dar[swap]{\sigma} \rar{\sigma} & \sum_{n} Tn \times a^n  \dar{K\sigma} \\
\sum_{n} Tn \times a^n \rar{\delta} & \sum_{n} \left(  \sum_{m} Tm \times n^m \right) \times a^n
\end{tikzcd}\]  
When \(\sigma(t) = (n;c,s)\), we have that \(\sigma(s) = (n;\id,s)\). This is to say that
the shape of a shape \(s\) is again \(s\). In this sense, taking the shape is idempotent.

\subsection{Linearity and unitarity from coalgebra laws}
\label{sec:orga9670fb}
We will show that traversables can be defined equivalently as
coalgebras for the shape-contents comonad.  This definition feels
intuitive to us: traversables are precisely functors equipped with a split
into shape and contents.

\begin{theorem}
\label{orge517918}
A coalgebra for the shape-contents comonad is a traversal.
\end{theorem}
\begin{proof}
Because of Lemma \ref{org8b21ea5}, we already have a bijection between natural
transformations \(T \tonat \sum_nTn \times (-)^n\) and natural transformations \(T \tonat \int\nolimits_{F \in \mathbf{App}}\mathsf{Ran}_FFT\). 
We have already shown that \(KT = \sum\nolimits_n T(n) \times (-)^n\) acts as a comonad;
we will show that linearity and unity follow from the coalgebra laws.

We first show that there exists a family \(n_{T,F} \colon KT \circ F \tonat F \circ KT\) natural
in both \(T \in [\Sets,\Sets]\) and \(F \in \mathbf{App}\).  In fact, we can use the multiplication
of \(F\) and the fact that every functor is lax monoidal with respect to the coproduct
to construct the following map and check that it is natural.  Let \(\alpha \colon F \tonat G\)
be a morphism of applicatives, which must preserve the multiplication and unit, and make
the first and second squares commute.  The third square commutes because of naturality of \(\alpha\).
\[\begin{tikzcd}
\sum_n (Fa)^n \times Tn \rar{u_{F,Tn}} \ar[bend left=13]{rrr}{n_{T,F}} \dar{\alpha}& 
\sum_n F(a^n) \times FTn \rar{w_{F}} \dar{\alpha}& 
\sum_n F(a^n\times Tn) \rar \dar{\alpha}& 
F \left( \sum_n a^n \times Tn \right) \dar{\alpha}\\
\sum_n (Ga)^n \times Tn \rar{u_{G,Tn}} \ar[bend right=13]{rrr}{n_{T,G}} & 
\sum_n G(a^n) \times GTn \rar{w_{G}} & 
\sum_n G(a^n\times Tn) \rar & 
G \left( \sum_n a^n \times Tn \right) \\
\end{tikzcd}\]
We define now \(m_{T,F} = F\varepsilon_T \circ n_{T,F} \colon KT \circ F \tonat F \circ T\), and the traverse of our
functor \(T\) will be \(\trv_{T,F} = m_{T,F} \circ \sigma_T \colon T \circ F \tonat F \circ T\).

We now prove \emph{unitarity} from the counitality axiom \(\varepsilon \circ \sigma = \id\).  The following is the
relevant diagram, showing how they both imply each other. Note that \(n_{T,\id} = \id\).
\[\begin{tikzcd}
T \circ \id \ar{rr}{\sigma_T} && KT \circ \id \ar{rr}{m_{T,\id}} \ar{dr}[swap]{n_{T,\id} = \id} && \id \circ T \\
&&& \id \circ KT \urar[swap]{\varepsilon_T} &
\end{tikzcd}\]

We now prove \emph{linearity} from the coalgebra axiom \(K\sigma \circ \sigma = \delta \circ \sigma\). In order to
do that, we will simplify the two sides of the linearity equation to make them
match coassociativity. The first side of the equation can be simplified as follows.
\[\begin{tikzcd}
T \circ F \circ G \rar{\sigma} \ar[bend left=15]{rr}{\trv_{F \circ G}}& 
KT \circ F \circ G \rar{m_F} \dar{K\sigma}&
F \circ T \circ G \dar{F\sigma} \ar[bend left=70]{dd}{F\trv_F}\\
& K^2T \circ F \circ G\rar{m_F} & 
F \circ KT \circ G \dar{m_G}\\ &&
F \circ G \circ T &
\end{tikzcd}\]
The second side of the equation can be also simplified.  In order to do this, we first write
\(m_{F \circ G}\) in terms of \(m_F\) and \(m_G\).  Applicative functors are composed composing
their units and multiplications; this makes \(n_{FG} = Fn_G \circ n_F\).
\begin{align*}
& m_{F \circ G} \\
= & \qquad\mbox{(Definition of $m$)} \\
& FG\varepsilon_T \circ n_{F\circ G} \\
= & \qquad\mbox{(Using that $n_{FG} = Fn_G \circ n_F$)} \\
& FG\varepsilon_T \circ Fn_G \circ n_F \\
= & \qquad\mbox{(Functoriality)} \\
& F(G\varepsilon_T \circ n_G) \circ n_F \\
= & \qquad\mbox{(Definition of $m$)} \\
& Fm_G \circ n_F \\
= & \qquad\mbox{(Counitality for the algebra)} \\
& Fm_G \circ F\varepsilon \circ F\delta \circ n_F \\
= & \qquad\mbox{(Naturality of $n$)} \\
& Fm_G \circ F\varepsilon \circ n_F \circ \delta \\
= & \qquad\mbox{(Definition of $m$)} \\
& Fm_G \circ m_F \circ \delta. \\
\end{align*}
We now proceed to simplify the diagram with the other side of the linearity
equation as follows.
\[\begin{tikzcd}
T \circ F \circ G \ar[bend left=15]{rr}{\trv_{F \circ G}}\rar{\sigma} & KT \circ F \circ G\dar[swap]{\delta} \rar{m_{FG}} & F \circ G \circ T \\
& K^2T \circ F \circ G \rar{m_F} & F \circ KT \circ G \uar[swap]{m_G}
\end{tikzcd}\]
We can see now that coassociativity implies linearity. In fact, when
we use \(\delta\circ\sigma = K\sigma \circ \sigma\), we get the desired result. 

Let us show now that linearity implies coassociativity. Let \(a,b \in \C\) and consider the
functor \(L_{a,b} = a \times (b \to (-))\) we described in \S \ref{opticscofree}
and its free applicative \(A_{x,y} = \sum\nolimits_{n} x^n \times (y^n \to (-))\).  Note that we have
a trivial monomorphism \(i \colon a \to A_{a,b}(b)\).  It can be
checked that \(A_{a,b} T(b) \cong KT(a)\), but also that the following diagram commutes.
\[\begin{tikzcd}
KT(a) \rar{i} \drar[swap]{\cong} & KT (A_{a,b}(b)) \dar{n_{R_{a,b}}} \\
& R_{a,b}T(b) 
\end{tikzcd}\]
By linearity, the following diagram commutes. Note that the internal square does
not necessarily commute yet, but we will show that coassociativity follows from
this.
\[\begin{tikzcd}
TA_{a,b}(b) \rar{i} & TA_{a,b}A_{b,c}(c) \rar{\sigma}\dar{\sigma} & \dar{\delta} KT A_{a,b}A_{b,c}(c) & A_{a,b}A_{b,c}T(c)\\
T(a) \uar{i} & KTA_{a,b}A_{b,c}(c) \rar{\sigma}& K^2TA_{a,b}A_{b,c}(c)\rar{n_{A_{a,b}}} & A_{a,b}KTA_{b,c}(c) \uar{n_{A_{b,c}}}
\end{tikzcd}\]
Naturality allows us to rewrite this into the following diagram. Because the last
part of the diagram is an isomorphism, the internal square commutes.
\[\begin{tikzcd}
T(a) \dar{\sigma}\rar{\sigma} & KT(a)\dar{\delta} & \\
KT(a) \rar{\sigma}& K^2T(a) \dar{i} \\ & \int_b K^2TA_{a,b}(b) \rar{n}
& \int_b A_{a,b}KT(b) \dar{i} \\&& \int_{b,c} A_{a,b}KTA_{b,c}(c) \rar{n} & \int_{b,c} A_{a,b}A_{b,c}T(c)
\end{tikzcd}\]
\end{proof}

\begin{remark}
\label{remark:algebrafromlinearity}
At the moment it is not clear to us that the isomorphism constructed
in Lemma \ref{org8b21ea5} is precisely the one used to construct the
traversal from the coalgebra in Theorem \ref{orge517918}.  
The main idea here is that Lemma \ref{org8b21ea5} should be a morphism of comonads,
and linearity and unitarity should correspond (in general) to the comonad axioms
of the left hand side.  We assume this for \S\ref{sec:traversablesasoptic}; see
also the Appendix \ref{appendixtraversables}.
\end{remark}

Relating traversables and coalgebras makes us consider \emph{cofree traversables}
given by \(KH\) for an arbitrary functor \(H\).  If we consider coalgebra morphisms
between traversables to define a category \(\mathbf{Trv}\), we have the adjunction
\(\mathbf{Trv}(T,KH) \cong [\Sets,\Sets](T,H)\), see Appendix \ref{appendixtraversables}.

\section{Traversals as the optic for traversables}
\label{sec:org9e10ef2}
\label{sec:traversablesasoptic}
The definition of traversables as coalgebras and the construction of cofree
traversables also provides a new description of the traversals as optics for
the evaluation of traversable functors.  This is the result widely used by
optic libraries to provide a profunctor description of the traversal.

\begin{proposition}
\label{prop:traversal2}
In \(\Sets\), the traversal is the optic for traversable functors.
\end{proposition}
\begin{proof}
\begin{align*}
& \int^{T \in \mathbf{Trv}}\Sets(s, Ta) \times \Sets(Tb , t) \\
\cong & \qquad\mbox{(Yoneda lemma)} \\
& \int^{T \in \mathbf{Trv}}\Sets(s, Ta) \times \Sets\left( \int^c Tc \times (c \to b) , t \right) \\
\cong & \qquad\mbox{(Continuity)} \\
&\int^{T \in \mathbf{Trv}}\Sets(s, Ta) \times \int_c \Sets\left(  Tc \times (c \to b) , t \right) \\
\cong & \qquad\mbox{(Exponential)} \\
& \int^{T \in \mathbf{Trv}}\Sets(s, Ta) \times \int_c \Sets\left(  Tc , (c \to b) \to t \right) \\
\cong & \qquad\mbox{(Natural transformation as an end)} \\
&\int^{T \in \mathbf{Trv}}\Sets(s, Ta) \times \mathrm{Nat} \left(  T , ((-) \to b) \to t \right) \\
\cong & \qquad\mbox{(Cofree traversable)} \\
&\int^{T \in \mathbf{Trv}}\Sets(s, Ta) \times \mathbf{Trv} \left(  T , \sum_n (-)^n \times (b^n \to t) \right) \\
\cong & \qquad\mbox{(Yoneda lemma)} \\
&\Sets \left(  s , \sum\nolimits_n a^n \times (b^n \to t) \right) & \qedhere
\end{align*}
\end{proof}

\section{Species}
\label{sec:org9f485b8}
A final observation on power series functors \(P\) is that they are
precisely the left Kan extensions of some family of sets indexed
by the natural numbers \(S \colon \mathbb{N} \to \Sets\) over the inclusion \(i \colon \mathbb{N} \to \Sets\)
of every natural number as a set with that cardinality.
\[\begin{tikzcd}
\mathbb{N}  \dar[hook, swap]{i} \rar{S} & \Sets \\
\Sets \urar[swap]{P} & 
\end{tikzcd}\]
The formula for left Kan extensions (Proposition \ref{orge75964b}) gives us \(P(a) = \sum_{n \in \mathbb{N}}a^n \times S(n)\),
where the fact that the category is discrete turns the end into a
sum.  \emph{Linear species}, as in \cite{bergeron98}, are described on this way,
just substituting natural numbers for an equivalent category given by
linearly ordered sets with monotone bijections.  \emph{Combinatorial species}
(or just \emph{species}) generalize these. This motivates the idea of extending
our discussion of traversals to arbitrary species.

We follow \cite{yorgey14} in discussing combinatorial species. 
A \textbf{combinatorial species} is a copresheaf on the
groupoid of finite sets with bijections, \(F \colon \mathbf{B} \to \Sets\). 
In other words, we are assigning a set of shapes to
any finite set of labels, in such a way that for any bijection that we
apply to the labels we get back a bijection for the shapes.

\begin{definition}[Joyal, 1986]
\textbf{Analytic functors} are those that arise as the left Kan extension of a
combinatorial species \(F \colon \mathbf{B} \to \Sets\) along the inclusion \(\mathbf{B} \to \Sets\).
\[\begin{tikzcd}
\mathbf{B}  \dar[hook,swap]{i} \rar{F} & \Sets \\
\Sets \urar[swap]{\hat F} & 
\end{tikzcd}\]
The formula for left Kan extensions gives us \(\hat{F}(a) = \int\nolimits^{l \in \mathbf{B}} (il \to a) \times F(l)\).
Note that taking a coend over the category \(\mathbf{B}\) instead of over a discrete
category means we need to quotient by the group of permutations.
\[
\hat{F}(a) = \int^{n \in \mathbf{B}} F(n) \times a^n
\]
\end{definition}

Following the analogy with traversals, we can define the
\textbf{unsorted traversal} as the optic associated with the evaluation
of analytic functors.  Note that we have an operation that composes species
described in \cite{yorgey14} and that the inclusion is unital with respect
to this operation.

\begin{proposition}
\label{prop:unsortedtraversal}
The unsorted traversal has a concrete form given by
\[
\mathbf{UnsortedTraversal}
\left( \begin{pmatrix} a \\ b \end{pmatrix}, \begin{pmatrix} s \\ t \end{pmatrix} \right) 
=\Sets \left( s, \int^{n \in \mathbf{B}} a^{n} \times \left(b^n \to t\right) \right).
\]
\end{proposition}
\begin{proof}
We will make use of the adjunction \(\Sets(\hat{F}b , t) \cong \mathbf{Species}(F , b^{(-)} \to t)\)
described by \cite[\S6.4]{yorgey14}.  Note also that the analytic functor
for the species \(b^{(-)} \to t\) is given by \(\int\nolimits^n (-)^n \times (b^n \to t)\). 
\begin{align*}
& \int^{F \in \mathbf{Species}} \Sets(s , \hat{F}a) \times \Sets(\hat{F}b , t) \\
\cong & \qquad{\mbox{(Adjunction in \cite{yorgey14})}} \\
& \int^{F \in \mathbf{Species}} \Sets(s , \hat{F}a) \times \mathbf{Species}(F , b^{(-)} \to t) \\
\cong & \qquad{\mbox{(Yoneda, and analytic functor for the species)}} \\
& \Sets \left(   s , \int^{n \in \mathbf{B}} a^{n} \times (b^{n} \to t) \right). &\qedhere \\
\end{align*}
\end{proof}

\chapter{Profunctor optics}
\label{sec:org301ea79}
Tambara modules were first described in \cite{tambara06} for the case of monoidal
categories. Fixing a monoidal category \(\C\), Tambara modules are structures on top
of the endoprofunctors \(\Prof(\C)\). It was shown in \cite{pastro08} that a category
of this structures with morphisms preserving them in some suitable sense is
equivalent to the copresheaves of some category called there the \emph{double} of
a monoidal category.

Our interest in  Tambara modules and their characterization comes from
the fact that profunctor optics are functions parametric precisely
over Tambara modules. In our case, these \emph{doubles} are categories of optics. 
We first provide a proof that these structures
are the coalgebras for a comonad also described in \cite{pastro08}.
From this result, we directly get the profunctor representation
theorem, which relates optics in existential form
to their profunctor form.

\section{Tambara modules}
\label{sec:orgbed2b5a}
During this section we fix a monoidal category \(\M\) that acts both on
two arbitrary categories \(\C\) and \(\D\).  We write \(\mact\) for the image
of \(m \in \M\) both in \([\C,\C]\) and \([\D,\D]\). We write \(\phi\) for the
structure isomorphisms of these strong monoidal actions.

\begin{definition}
\label{org0779c9e}
A \textbf{Tambara module} consists of a profunctor \(p \colon \C^{op} \times \mathbf{D} \to \Sets\)
endowed with a family of morphisms \(\alpha_m \colon p(a,b) \to p(\underline{m} a , \underline{m} b)\)
natural in both \(a \in \C\) and \(b \in \D\), and dinatural  in \(m \in \M\); which additionally
satisfy the two equations \(\alpha_i = p(\phi_{i}^{-1},\phi_{i})\) and \(\alpha_{m \otimes n} = p(\phi_{m,n}^{-1},\phi_{m,n}) \circ \alpha_{m} \circ \alpha_{n}\).
\[\begin{tikzcd}
p(a,b) \drar[swap]{\mathrm{id}}\rar{\alpha_{i}} & p(\iact a, \iact b) \dar{p(\phi^{-1},\phi)} &&
p(a,b) \ar[swap]{drr}{\alpha_{n \otimes m}}\rar{\alpha_n} & p(\underline{n}a,\underline{n}b) \rar{\alpha_m} & p(\underline{mn}a,\underline{mn}b)\dar{p(\phi^{-1},\phi)} \\
& p(a,b) &&&& p(\underline{m \otimes n}a,\underline{m \otimes n}b)
\end{tikzcd}\]
\end{definition}

\begin{remark}
The definition of Tambara module in \cite{pastro08} deals only
with actions that arise from a monoidal product \(\otimes \colon \C \to [ \C , \C ]\). We
have decided to use the term \emph{Tambara module} also for the more general concept,
instead of introducing new nomenclature.  This generalization also includes
\emph{mixed optics}, that were proposed by \cite{riley18} as further work.
\end{remark}

\begin{definition}
In the same way that we introduced \(\hirayo\) to represent the Yoneda embedding, we
will write the hiragana ``ta'', \(\hirata\), to refer to Tambara modules.
Let \(\hirata\) be the category of Tambara modules with morphisms \((p,\alpha) \to (q,\alpha')\) given
natural transformations \(\eta \colon p \tonat q\) that satisfy \(\eta_{a,b} \circ \alpha_{m,a,b} = \alpha'_{m,a,b}\circ \eta_{\underline{m}a,\underline{m}b}\).
Diagrammatically, these are transformations such that the following diagram commutes.
\[\begin{tikzcd}
p(a,b) \dar[swap]{\eta_{a,b}} \rar{\alpha_m} & p(\underline{m}a,\underline{m}b) \dar{\eta_{\underline{m}a, \underline{m}b}}\\
q(a,b) \rar{\alpha'_{m}}& q(\underline{m}a,\underline{m}b) \\
\end{tikzcd}\]
\end{definition}

\section{The Pastro-Street comonad}
\label{sec:orgf55dd77}
Tambara modules are equivalently coalgebras for a comonad studied in
\cite{pastro08}.  That comonad has a left adjoint that must therefore
be a monad, and Tambara modules can be also characterized as algebras
for that monad.  We will get the category of Tambara modules \(\hirata\) as an
Eilenberg-Moore category, and this will be the main lemma towards the
profunctor representation theorem.
\begin{definition}
\label{orgdb10a66}
We define \(\Theta \colon \Prof(\C,\D) \to \Prof(\C,\D)\) as
\[
\Theta(p)(a,b) = \int_{m \in \M} p(\underline{m}a, \underline{m}b).
\]
This is a comonad.
\end{definition}
\begin{proof}
We start by showing that it is indeed a functor. We have already defined its
action on objects, so we proceed to define its action on morphisms. Given two
profunctors \(p , q \in \Prof(\C,\D)\), let \(\eta \colon p \tonat q\) be a natural transformation. The
following diagram constructs a wedge for some \(h \colon m \to n\) that in turns defines a unique map
\(\int_{m} p(ma,mb) \to \int_{m} q(ma,mb)\). The map is natural in \(a \in \C\) and \(b \in \D\), as it is
composed of natural maps. That gives a natural transformation \(\Theta p \tonat \Theta q\).
Squares commute because of naturality of \(\eta\) and dinaturality of the coend in \(p\).
\[\begin{tikzcd}
& & \int_{m} p(\underline{m}a, \underline{m}b) \ar{dr}{\pi_n}\ar{dll}[swap]{\pi_m} \ar[dashed, near end]{ddd}{\Theta\eta_{a,b}} & \\
p(\underline{m}a, \underline{m}b) \ar{ddd}{\eta_{ma,mb}}\ar{dr}{p(\mathrm{id},\underline{h})} & && p(\underline{n}a,\underline{n}b)\ar{dll}[swap]{p(h,\mathrm{id})}\ar{ddd}{\eta_{na,nb}} \\
& p(\underline{m}a, \underline{n}b)\ar[near end]{ddd}{\eta_{ma,nb}} &&& \\
& & \int_{m} q(\underline{m}a, \underline{m}b) \ar{dr}{\pi_n}\ar{dll}[swap]{\pi_m} & \\
q(\underline{m}a, \underline{m}b)  \ar{dr}[swap]{q(\mathrm{id},h)} & && q(\underline{n}a,\underline{n}b) \ar{dll}{q(h,\mathrm{id})} \\
& q(\underline{m}a, \underline{n}b) &&& \\
\end{tikzcd}\]
Functoriality follows from the fact that, when \(\eta = \id\), the identity map makes the diagram
commute; and from the fact that, when \(\eta = \eta_2\circ \eta_1\), the composite \(\Theta\eta_2 \circ \Theta\eta_1\)
makes the diagram commute. More abstractly, we are using the functoriality of the coends and the
naturality to lift the natural transformation.

We proceed to describe the components of the comonad.
The \emph{counit} is \(\varepsilon_{p} = p(\phi^{-1}, \phi) \circ \pi_{i}\), defined by projecting on the 
monoidal unit component.
\[\begin{tikzcd}
\int\nolimits_{m \in \M} p(\underline{m}a , \underline{m}b) \rar{\pi_{i}}&
p(\iact a , \iact b) \rar{p(\phi^{-1}, \phi)}&
p(a , b)
\end{tikzcd}\]
The \emph{comultiplication} \(\delta_{p}\) is obtained by the universal property of the
end as the unique morphism making the following diagram commute.
\[\begin{tikzcd}
p(\underline{m \otimes n}a , \underline{m \otimes n} b) \dar[swap]{p(\phi^{-1},\phi)}& 
\int_{m \in \M} p(\underline{m}a , \underline{m}b) \lar[swap]{\pi_{m \otimes n}}\dar[dashed]{\exists! \delta_{p}} \\
p(\underline{m}\underline{n}a , \underline{m} \underline{n} b) & 
\int_{n\in \M}\int_{m\in \M} p(\underline{m}\underline{n}a , \underline{m} \underline{n} b) \lar[swap]{\pi_{m}\circ\pi_{n}} \\
\end{tikzcd}\]
We will show now that it is indeed a comonad, proving counitality and
coassociativity. \emph{Counitality}, \(\Theta \varepsilon \circ \delta = \mathrm{id}\), 
follows from commutativity of the following diagram. We
use our definitions and the coherence of the end.  The other side of
counitality, \(\varepsilon \circ \delta = \mathrm{id}\), is similar.
\[\begin{tikzcd}
p(\underline{i \otimes u} a , \underline{i \otimes u} a) \dar && \int_m p(\underline{m}a, \underline{m}b) \dar{\delta}\ar{ll}[swap]{\pi_{i\otimes u}} \\
p(\iact \underline{u}a, \iact \underline{u}b) \dar & \int_m p(\underline{m}ua, \underline{m}ub) \lar{\pi_i} & \int_{n}\int_{m} p(\underline{m}\underline{n}a,\underline{m}\underline{n}b) \lar{\pi_u}\dar{\Theta\varepsilon} \\
p(ua,ub) && \int_n p(\underline{n}a , \underline{n}b) \ar{ll}{\pi_u}
\end{tikzcd}\]
\emph{Coassociativity}, \(\Theta\delta\circ\delta = \delta\circ\delta\), follows from commutativity of the following
diagram. The internal squares commute by definition and coherence of the action.
Finally, the two outermost morphisms are the same because of coherence of
the first coend.
\[\begin{tikzcd}[column sep=tiny]
& \int_m p(\underline{m}a,\underline{m}b)\drar[swap]{\delta}\dlar{\delta} & \\
\int_n\int_m p(\underline{m}\underline{n}a,\underline{m}\underline{n}b) \rar[swap]{\Theta\delta}\dar{\pi} & \int_o\int_n\int_m p(\underline{m}\underline{n}\underline{o}a,\underline{m}\underline{n}\underline{o}b) \dar{\pi} & \int_n\int_m p(\underline{m}\underline{n}a,\underline{m}\underline{n}b) \ar{dd}{\pi}\lar{\delta} \\
\int_mp(\underline{m}\underline{w}a,\underline{m}\underline{w}b) \rar\ar{dd}{\pi} & \int_n\int_mp(\underline{mnw}a,\underline{mnw}b) \dar{\pi}& \\
&\int_mp(\underline{mvw}a,\underline{mvw}b) \dar{\pi} & \int_mp(\underline{m(v\otimes w)}a, \underline{m(v\otimes w)}b)\dar{\pi}\lar \\
p(\underline{(u\otimes v)w}a,\underline{(u\otimes v)w}b) \rar\drar & p(\underline{uvw}a,\underline{uvw}b) \rar\dar{\pi} & p(\underline{u(v\otimes w)}a, \underline{u(v\otimes w)}b) \lar\dlar\\
& p(\underline{u\otimes v\otimes w}a, \underline{u\otimes v\otimes w}b) &
\end{tikzcd}\]
\end{proof}
\begin{proposition}
Tambara modules are equivalently coalgebras for this comonad. The category
\(\hirata\) is equivalent, with a bijective-on-objects functor, to the Eilenberg-Moore
category of \(\Theta\).
\end{proposition}
\begin{proof}
Note that the data for a coalgebra is a natural transformation \(\alpha \colon p \tonat \Theta p\)
whose projections are components of the Tambara module, \(\alpha_{m,a,b} = \pi_m \circ \alpha_{a,b}\).
The naturality of \(\alpha\) is exactly the naturality of the components of the
Tambara module; the coherence conditions of the end are precisely the
dinaturality of the components.  The only thing we need to show is that
the coalgebra axioms correspond with the Tambara axioms.

For the counit, we know that \(\mathrm{id} = \varepsilon \circ \alpha = p(\phi^{-1},\phi) \circ \pi_i \circ \alpha = p(\phi^{-1},\phi) \circ \alpha_{i}\),
giving the first axiom. Diagrammatically, we have the following.
\[\begin{tikzcd}
p(a,b) \drar[swap]{\mathrm{id}}\rar{\alpha} \ar[bend left]{rr}{\alpha_i} & 
\int\nolimits_m p(\underline{m}a, \underline{m} b) \dar{\varepsilon}\rar{\pi_i} & 
p(\iact a, \iact b) \dlar{p(\phi^{-1},\phi)}\\
& p(a,b) &
\end{tikzcd}\]
For the comultiplication, we note that the following two diagrams are giving
the same morphism if and only if the Tambara condition holds.  Because of the
universal property of the end, this is the same as to say that \(\delta \circ \alpha = \Theta\alpha \circ \alpha\).
\[\begin{tikzcd}
p(a,b) \rar{\alpha}\drar[swap]{\alpha_{u \otimes v}}& 
\int\nolimits_m p(\underline{m}a, \underline{m} b) \rar{\delta} \dar{\pi}&
\int\nolimits_m\int\nolimits_n p(\underline{mn}a, \underline{mn} b)\dar{\pi} \\
& p(\underline{u \otimes v} a, \underline{u \otimes v} b) \rar{p(\phi^{-1},\phi)} & 
p(\underline{u v} a, \underline{u v} b)
\end{tikzcd}\]

\[\begin{tikzcd}
p(a,b) \rar{\alpha} \drar[swap]{\alpha_v} &
\int_m p(\underline{m}a,\underline{m}b) \rar{\Theta\alpha} \dar{\pi_v} &
\int_n\int_m p(\underline{mn}a,\underline{mn}b) \dar{\pi_v} \\& 
p(\underline{v}a,\underline{v}b) \drar[swap]{\alpha_u} \rar{\alpha} &
\int_m p(\underline{mv}a,\underline{mv}b) \dar{\pi_u} \\ && 
p(\underline{uv}a,\underline{uv}b)
\end{tikzcd}\]
Finally, we will show that morphisms of Tambara modules and coalgebra morphisms
are the same thing. This gives a bijective-on-objects and fully faithful functor
between \(\hirata\) and the Eilenberg-Moore category of \(\Theta\).  Given a natural
transformation \(\eta \colon p \tonat q\) between Tambara modules endowed with \(\alpha\) and \(\alpha'\), the
exterior part of this diagram commutes when \(\eta\) is a morphism of Tambara modules
and the interior part commutes when \(\eta\) is a coalgebra morphism.  By the universal
property of the end and the definition of \(\Theta\), they both imply each other.
\[\begin{tikzcd}
p(a,b) \rar{\alpha_m}\ar[bend left]{rr}{\alpha}\dar[swap]{\eta_{a,b}} & 
\int_m p(\underline{m}a, \underline{m}b) \rar{\pi_{m}}\dar{\Theta \eta_{a,b}} &
p(\underline{m}a, \underline{m}b) \dar{\eta_{\underline{m}a},\underline{m}b}\\
q(a,b) \rar{\alpha'}\ar[bend right]{rr}{\alpha'_m} &
\int_m q(\underline{m}a, \underline{m}b) \rar{\pi_{m}} &
q(\underline{m}a, \underline{m}b) \\
\end{tikzcd}\]
\end{proof}
\begin{proposition}[\cite{pastro08}]
\label{orgb64efb5}
The \(\Theta\) comonad has a left adjoint, which must therefore be a monad. On objects,
it is given by the following formula.
\[\Psi q(x,y) = \int^{m \in \M}\int^{a \in \C,b \in \D}
q(a,b) \times \C(\underline{m}a,x) \times \D(y,\underline{m}b).\]
That is, there exist a natural isomorphism \(\Nat(\Phi q,p) \cong \Nat(q,\Theta p)\).
\end{proposition}
\begin{proof}
This comonad can also be written as \(\Theta(p) = \int\nolimits_{m \in \M} p \circ (\underline{m} , \underline{m})\). From this definition we
can see that it has a left adjoint \(\Psi(p) = \int^{m \in \M} \mathsf{Lan}_{(\underline{m},\underline{m})} p\), which must be a monad.
More explicitly, the adjunction can be computed as follows for any
given \(p,q \colon \C^{op} \times \mathbf{D} \to \Sets\).
\begin{align*}
& q \Rightarrow \Theta p \\
\cong & \qquad\mbox{(Natural transformation as an end)} \\
& \int_{(a,b) \in \C^{op} \times \mathbf{D}} q(a,b) \to \Theta p (a,b) \\
\cong & \qquad\mbox{(Fubini rule)} \\
& \int_{a\in\C,b \in\D} q(a,b) \to \Theta p (a,b) \\
\cong & \qquad\mbox{(Definition of $\Theta$)} \\
& \int_{a\in\C,b \in \D} q(a,b) \to \int_{m \in \M} p (\underline{m}a,\underline{m}b) \\
\cong & \qquad\mbox{(Continuity of the end)} \\
& \int_{a \in \C, b \in \D} \int_{m \in \M} q(a,b) \to  p (\underline{m}a,\underline{m}b) \\
\cong & \qquad\mbox{(Fubini rule)} \\
& \int_{m \in \M} \int_{a \in \C, b \in \D} q(a,b) \to  p (\underline{m}a,\underline{m}b) \\
\cong & \qquad\mbox{(Ninja Yoneda lemma)} \\
& \int_{m \in \M} \int_{a \in \C, b \in \D}  q(a,b) \to \int_{x\in\C,y\in\D} \C(\underline{m}a,x) \times \D(y,\underline{m}b) \to p (x,y) \\
\cong & \qquad\mbox{(Continuity of the end)} \\
& \int_{m \in \M} \int_{a \in \C, b \in \D} \int_{x\in\C,y\in\D}  q(a,b) \to \left( \C(\underline{m}a,x) \times \D(y,\underline{m}b) \to p (x,y) \right) \\
\cong & \qquad\mbox{(Currying)} \\
& \int_{m \in \M} \int_{a \in \C, b \in \D} \int_{x\in\C,y\in\D}  q(a,b) \times \C(\underline{m}a,x) \times \D(y,\underline{m}b) \to p (x,y) \\
\cong & \qquad\mbox{(Fubini rule)} \\
& \int_{x\in\C,y\in\D}  \int_{a \in \C, b \in \D}  \int_{m \in \M}  q(a,b) \times \C(\underline{m}a,x) \times \D(y,\underline{m}b) \to p (x,y) \\
\cong & \qquad\mbox{(Cocontinuity of the coend)} \\
& \int_{x\in\C,y\in\D} \left(  \int^{m \in \M}\int^{a \in \C, b \in \D}  q(a,b) \times \C(\underline{m}a,x) \times \D(y,\underline{m}b) \right) \to p (x,y) \\
\cong & \qquad\mbox{(Definition of $\Psi$)} \\
& \int_{x\in\C,y\in\D} \Psi q(x,y) \to p (x,y) \\
\cong & \qquad\mbox{(Natural transformation as an end)} \\
& \Psi q \Rightarrow p. & \qedhere \\
\end{align*}
\end{proof}
\begin{remark}
Because of Lemma \ref{orgb64efb5}, Tambara modules are algebras for the monad
\(\Psi\). In particular, knowing that the category \(\hirata\) of Tambara modules is equivalently
the category of coalgebras for \(\Theta\) or the category of algebras for \(\Phi\), we can
construct free and cofree Tambara modules. Given a Tambara module \(q \in \hirata\) and
some arbitrary functor \(p \in \Prof(\C)\), we have the following adjuctions \(\Phi \dashv U \dashv \Theta\).
\[
\hirata(\Phi p , q) \cong \mathbf{Prof}(p,Uq), \qquad \hirata(q , \Theta p) \cong \Prof(Uq,p).
\]
\end{remark}
\section{The profunctor representation theorem}
\label{sec:org674ebca}
\begin{theorem}[\cite{riley18,boisseau18}]
\label{org9a4afc7}
\[\int_{p \in \hirata} \Sets(p(a,b) , p(s,t))
\cong
\Optic((a,b),(s,t)).
\]
\end{theorem}
\begin{proof}
The proof here is different from the ones in \cite{riley18,boisseau18}, although
it follows the same basic idea.
\begin{align*}
& \int_{p \in \hirata} \Sets(p(a,b), p(s,t)) \\
\cong & \qquad\mbox{(Yoneda lemma)} \\
& \int_{p \in \hirata} \Sets\left(\mathrm{Nat}(\hirayo_{(a,b)}, p) , p(s,t) \right) \\
\cong & \qquad\mbox{(Free-forgetful adjunction for Tambara modules)} \\
& \int_{p \in \hirata} \Sets\left( \hirata(\Psi_\M\hirayo_{(a,b)}, p) , p(s,t) \right) \\
\cong & \qquad\mbox{(Yoneda lemma)} \\
& \Psi\hirayo_{(a,b)}(s,t) \\
\cong & \qquad\mbox{(Definition of $\Psi$)} \\
& \int^{m \in \M} \int^{x\in\C,y \in\D} \C(s,\mact x) \times \D(\mact y, t) \times \hirayo_{(a,b)}(x,y) \\
\cong & \qquad\mbox{(Yoneda lemma)} \\
& \int^{m \in \M} \C(s,\mact a) \times \mathbf{D}(\mact b, t). & \qedhere\\
\end{align*}
\end{proof}

\begin{remark}
In fact, we can stop midway there and say that an optic is an element
of \(\Psi\hirayo_{(a,b)}(s,t)\), which is, again by Yoneda lemma, the same as a natural
transformation \(\hirayo_{(s,t)}
\Rightarrow \Psi\hirayo_{(a,b)}\).  We could have defined
our category of optics to be the opposite of the full subcategory of
representable functors of the Kleisli category for \(\Psi\).
\end{remark}

\section{Examples of profunctor optics}
\label{sec:orgad9cd3d}
We can apply the profunctor representation theorem
\ref{org9a4afc7} to each one of our optics and get their
profunctor versions.

\begin{itemize}
\item \textbf{Lenses} are described in Proposition \ref{prop:lenses} as optics for the
product.  Tambara modules for the cartesian product are
called \emph{strong profunctors} \cite{kmett15} or \emph{cartesian profunctors} \cite{boisseau18}.
The following is the consequence of the profunctor representation theorem.
\[
   \C(s , a) \times \C(s \times b , t) \cong \int_{p \in \hirata(\times)} \Sets(p(a,b), p(s,t)).
   \]

\item \textbf{Prisms} are described in Proposition \ref{prop:prisms} as optics for the
coproduct.  Tambara modules for the coproduct are called \emph{choice profunctors}
\cite{kmett15} or \emph{cocartesian profunctors} \cite{boisseau18}.
\[
   \C(s , t + a) \times \C(b , t) \cong \int_{p \in \hirata(+)} \Sets(p(a,b), p(s,t)).
   \]

\item \textbf{Traversals} are described in Proposition \ref{prop:traversal} as optics for
the evaluation of power series functors.  We can write them as functions
parametric over Tambara modules for power series functors
\[
   \C \left(  s , \sum_{n \in \mathbb{N}} a^n \times (b^n \to t) \right) \cong
   \int_{p \in \hirata( \mathrm{Series})} \Sets(p(a,b), p(s,t)).
   \]
The commonly used characterization follows from Proposition \ref{prop:traversal2}, which
describes traversals as the optic for traversable functors.
\[
   \Sets \left(  s , \sum\nolimits_n a^n \times (b^n \to t) \right)
   \cong
   \int_{p\in \hirata(  \mathbf{Trv})} \Sets(p(a,b), p(s,t)).
   \]
\end{itemize}

A relevant question here is what happens when we compose two optics of
two different kinds: for instance a lens with a prism.  If we follow what Haskell implementations do,
we compose pointwise on profunctors that have Tambara structures for
\emph{both} monoidal actions.  Lenses and prisms would compose into something of 
the following form.
\[
\int_{(p,p) \in \hirata(\times) \times \hirata(+)} \Sets(p(a,b), p(s,t)).
\]
In our example, we would be taking an end over
the full subcategory of \(\hirata(\times) \times \hirata(+)\) given by structures with the
same underlying functor.  Categorically, this would be the following pullback
of categories.
\[\begin{tikzcd}
\hirata(\times) \times_{\Prof} \hirata( + ) \dar[swap]{\pi}\rar{\pi}& \hirata(\times) \dar{U}\\
\hirata( + ) \rar{U} & \Prof
\end{tikzcd}\]
We discuss this composition in \S \ref{orga66504e}.

\section{The category of optics as a Kleisli object}
\label{sec:org31966dc}
We have tried to give a geodesic to the profunctor representation theorem,
but the proof outlined in \cite{pastro08} for the case of monoidal products
can be repeated in full generality.  The following lemma is a consequence
of Proposition \ref{orgb64b29f}.

\begin{lemma}
There exists an identity on objects functor \(\C^{op} \times \D \to \Optic\), which
means that the following \(\psi \colon (\C^{op} \times \D) \times (\C^{op} \times \D)\), given by the
set of morphisms of \(\Optic\), is a promonad.
\[
\psi((a,b),(s,t)) = \int^{m \in \M} \C(s , \mact a) \times \D(s , \mact b).
\]
Moreover, \(\Optic\) is the Kleisli object for this promonad.
\end{lemma}

With this characterization, we can show that copresheaves over
\(\Optic\) are precisely modules over the promonad; and use
that to show that these are Tambara modules.

\begin{proposition}
\label{org66b0ed4}
Copresheaves over \(\Optic\) are equivalent to Tambara modules.
\end{proposition}
\begin{proof}
We start by showing that \(\mathbf{Cat}(\Optic, \Sets) \cong \mathbf{Mod}(\psi)\)
By the universal property of the Kleisli object, we know that
profunctors from the terminal category \(\mathbf{1} \to \Optic\) correspond to right
promodules over the promonad, \(\Prof( \mathbf{1} , \Optic) \cong \mathbf{Mod}(\psi)\).  We note
that profunctors of that form are functors \(\mathbf{1} \times \Optic \to \Sets\).

Now, because of Lemma \ref{org8d8bb19}, modules over the promonad \(\psi\)
are equivalent to algebras over \(\Psi\), which are precisely Tambara modules.
\end{proof}

We can use that to reprove the profunctor representation theorem with
a technique that is called \emph{double Yoneda} in \cite{milewski17}.

\begin{lemma}[``Double Yoneda'' from \cite{milewski17}]
\label{org9c383aa}
For any category \(\mathbf{A}\), morphisms between \(x\) and \(y\) are naturally isomorphic
to natural transformations between the functors that evaluate
copresheaves in \(x\) and \(y\).
\[
\mathbf{A}(x,y) \cong 
[ [ \mathbf{A} , \Sets ] , \Sets ](-(x),-(y)).
\]
\end{lemma}
\begin{proof}
In a category of functors \([ \mathbf{A} , \Sets ]\), we can apply the Yoneda embedding
to two representable functors \(\mathbf{A}(y,-)\) and \(\mathbf{A}(x,-)\) to get the following.
\[
\mathrm{Nat}( \mathbf{A}(y,-) , \mathbf{A}(x,-) ) \cong
\int_{f \in [\mathbf{A}, \Sets]} 
\Sets\Big(\mathrm{Nat}( \mathbf{A}(x,-), f),
\mathrm{Nat}( \mathbf{A}(y,-) ,f)\Big).
\]
Here reducing by Yoneda lemma on both the left hand side and the two arguments
of the right hand side, we get the desired result.
\end{proof}

\begin{theorem}[Profunctor representation theorem]
\label{orga34cf7c}
\[\int_{p \in \hirata} \Sets(p(a,b) , p(s,t))
\cong
\Optic((a,b),(s,t)).
\]
\end{theorem}
\begin{proof}
We apply the double Yoneda lemma (Lemma \ref{org9c383aa}) to the
category \(\Optic\) and then use that copresheaves over it are precisely
Tambara modules (Lemma \ref{org66b0ed4}).
\end{proof}

\chapter{Composing optics}
\label{sec:org5794c45}
\section{Change of action}
\label{sec:orgba80d95}
The motivation of this section is to create functors between categories of
optics that are induced by morphisms of actions.

\begin{lemma}
\label{orgd31a70b}
Let \(\alpha \colon \M \to [\C,\C]\) and \(\beta \colon \N \to [\C,\C]\) be two actions. A morphism of
actions \(f \colon \M \to \N\) induces a comonad morphism \(f_{\ast} \colon \Theta_{\N} \to \Theta_{\M}\).  Moreover,
this assignment is contravariantly functorial, in the sense that for
any \(f \colon \M \to \N\) and \(g \colon \N \to \mathbf{L}\), we have \((g \circ f)_{\ast} = f_{\ast} \circ g_{\ast}\) and \(\id_{\ast} = \id\).
\end{lemma}
\begin{proof}
Recall that a comonad morphism is a natural transformation that commutes
with counits and comultiplications.  We start by defining the natural
transformation, using the universal property of the end, as the only
morphism making the following diagram commute for any profunctor \(p\) and
any morphism \(h \colon m \to m'\) in \(\M\).  Here the back of the diagram commutes
because of dinaturality of the first end and naturality of the structure
morphisms.  Here \(\phi_f\) is the natural isomorphism \(\underline{fm}a \cong \mact a\) giving the
morphism of actions.
\[\begin{tikzcd}
& \int_{n \in \N} p(\nact a, \nact b) \arrow[ld, swap, "\pi_{fm}"] \arrow[rd, "\pi_{fm'}"] \arrow[d, dashed, "f_\ast"] & \\
p(\underline{fm} a, \underline{fm} b)  \arrow[d, "\phi_{f}",swap] \arrow[rd] &
\int_{m \in \M} p(\underline{m} a, \underline{m} b)  \arrow[ld, crossing over] &
p(\underline{fm'} a, \underline{fm'} b) \arrow[d, "\phi_f"] \arrow[ld] \\
p(\underline{m} a, \underline{m} b) \ar[swap]{rd}{p(\id,h)} &
p(\underline{fm}a, \underline{fm'}b) \dar & 
p(\underline{m'} a, \underline{m'} b)\ar{ld}{p(h,\id)} \arrow[from=lu, crossing over] \\
& p(\underline{m} a, \underline{m'} b) & 
\end{tikzcd}\]
We can see here that \((g \circ f)_{\ast} = f_{\ast} \circ g_{\ast}\) follows from the fact that
composition of monoidal actions is defined to have \(\phi_{g \circ f} = \phi_f \circ \phi_g\).

We show now that this natural transformation \(f_{\ast}\) is in fact a comonad morphism.
Counitality follows from commutativity of the following diagram. Here the upper
square commutes by definition of \(f_{\ast}\), the triangle commutes because of the conditions
of the end, and the pentagon commutes by coherence.
\[\begin{tikzcd}
& {\int_n p(\nact a, \nact b)} \dar{\pi_{fi}} \rar{f_{\ast}} \arrow[swap]{ld}{\pi_i}
& {\int_m p(\mact a, \mact b)} \dar{\pi_i} \\
{p(\iact  a, \iact  b)} \arrow{rd}{\phi} \arrow{r}{\phi} 
& {p(\underline{fi} a, \underline{fi} b)} \arrow{r}{\phi}           
& {p(\iact a,\iact b)} \dar{\phi} \\
& {p(a,b)} \arrow{r}{\id} & {p(a,b)}                                    
\end{tikzcd}\]
Comultiplicativity follows from the commutativity of the following diagram.
The squares of this diagram commute because of the definitions of \(\delta\), \(\Theta\) and \(f_{\ast}\)
and because of coherence.
\[\begin{tikzcd}[column sep=-2ex]
{\int_n p(\underline{n}a,\underline{n}b)} \arrow[rr, "f_\ast"] \arrow[rd, "\pi_{f(m\otimes m')}"] \arrow[rdd, "\pi_{fm \otimes fm'}" near end, bend right] \arrow[ddd, "\delta"'] &                                                                             & {\int_m p(\underline{m}a,\underline{m}b)} \arrow[rd, "\pi_{m\otimes m'}"] \arrow[ddd, "\delta"] &                                                                                \\
                                                                                                                                                                                     & {p(f(m\otimes m')a , f(m\otimes m')b)} \arrow[rr] \arrow[d, "\phi"]         &                                                                                                 & {p(\underline{m\otimes m'}a, \underline{m \otimes m'} b)} \arrow[dddd, "\phi"] \\
                                                                                                                                                                                     & {p(\underline{fm \otimes fm'} a, \underline{fm' \otimes fm} b)} \arrow[ddd] &                                                                                                 &                                                                                \\
{\int_n\int_{n'} p(nn'a, nn'b)} \arrow[rr, "\Theta f_\ast"', near end] \arrow[d, "\pi_n"]                                                                                                      &                                                                             & {\int_m\int_{m'} p(mm'a,mm'b)} \arrow[d, "\pi_m"]                                               &                                                                                \\
{\int_{n'}p(fmn'a, fmn'b)} \arrow[rd, "\pi_{fm'}"] \arrow[rr, "f_\ast", near end]                                                                                                              &                                                                             & {\int_{m'}p(mm'a, mm'b)} \arrow[rd, "\pi_{m'}"]                                                 &                                                                                \\
                                                                                                                                                                                     & {p(fmfm'a,fmfm'b} \arrow[rr, "\phi"]                                        &                                                                                                 & {p(mm'a,mm'b)}                                                                
\end{tikzcd}\]
\end{proof}

A comonad morphism induces a functor between the Eilenberg-Moore
categories of the comonad, see Definition \ref{org4babfba}.
In this case, the coalgebras of the
Pastro-Street comonad are Tambara modules and we get a functor \(\hirata_{\N} \to \hirata_{\M}\).
These are also the Eilenberg-Moore categories of the adjoint
monads; and the functor between Eilenberg-Moore categories induces
in turn (see Proposition \ref{org88623d6}) a morphism of monads \(\Psi_{\M} \to \Psi_{\N}\).

\begin{proposition}
A morphism of actions \(f \colon \M \to \N\) induces a functor between the
categories of optics \(\Optic_{\M} \to \Optic_{\N}\).
\end{proposition}
\begin{proof}
We have seen that the morphism of actions induces a morphism of
monads.  The morphism of monads gives an identity-on-objects functor
between the Kleisli categories of the monads.  The category of optics
is the full subcategory on representable functors of that Kleisli
category; and because the functor is the identity on objects, it is
sent again to a subcategory of representable functors.
\end{proof}

\section{Distributive actions}
\label{sec:org526b44c}
\subsection{Distributive laws for Pastro-Street comonads}
\label{sec:org728df61}
If we have a pair of comonads \(\Theta_{\M}\) and \(\Theta_{\N}\) with left adjoint monads \(\Phi_{\M} \dashv \Theta_{\M}\)
and \(\Phi_{\N} \dashv \Theta_{\N}\), then a distributive law \(\Theta_{\M} \circ \Theta_{\N} \tonat \Theta_{\N} \circ \Theta_{\M}\) makes \(\Theta_{\M} \circ \Theta_{\N}\)
a comonad with a left adjoint given by the composition of the two adjunctions,
\(\Phi_{\N} \circ \Phi_{\M} \dashv \Theta_{\M} \circ \Theta_{\N}\). This makes \(\Phi_{\N} \circ \Phi_{\M}\) a monad.   This is particularly
useful because we can work in the Kleisli category of this monad and it will
contain in particular the Kleisli categories for each one of the monads.  That
is, we are creating a kind of optic that contains the two kinds we started with.

With this in mind, let us study what is a distributive law between two Pastro-Street comonads
as in Definition \ref{orgdb10a66}.  Given pairs of monoidal actions \(\M \to [\C,\C]\), \(\M \to [\D,\D]\),
and \(\N \to [\C,\C]\), \(\N \to [\D,\D]\), we are looking for a natural transformation \(\Theta_{\M} \circ \Theta_{\N} \tonat \Theta_{\N}\circ \Theta_{\M}\).
Let us apply the Yoneda lemma to reduce this; note that the category of presheaves is cocomplete
and coends exist in it.
\begin{align*}
& \Theta_{\M}\Theta_{\N} \tonat \Theta_{\N}\Theta_{\M} \\
\cong & \qquad\mbox{(Natural transformation as an end)} \\
& \int_{c \in \C}\int_{d \in \D}\int_{p \in \Prof(\C,\D)} \Theta_{\M}\Theta_{\N}p(c,d) \to \Theta_{\N}\Theta_{\M}p(c,d) \\
\cong & \qquad\mbox{(Fubini, Definition of $\Theta$)} \\
& \int_{c, d, p} 
\left(  \int_{m, n} p(\nact\mact c,\nact\mact d) \right) \to 
\left(  \int_{n', m'} p(\underline{m'n'} c, \underline{m'n'} d) \right)  \\
\cong & \qquad\mbox{(Yoneda lemma)} \\
& \int_{c, d, p} 
\left(  \int_{m,n} \Nat \left( \hirayo_{\nact\mact c,\nact\mact d} ,  p \right) \right) \to 
\left(  \int_{n', m'} p(\underline{m'n'} c, \underline{m'n'} d) \right)  \\
\cong & \qquad\mbox{(Continuity)} \\
& \int_{c, d, p} 
\Nat \left( \int^{m,n} \hirayo_{\nact\mact c,\nact\mact d} ,  p \right) \to 
\left(  \int_{n', m'} p(\underline{m'n'} c, \underline{m'n'} d) \right)  \\
\cong & \qquad\mbox{(Yoneda lemma, Fubini)} \\
& \int_{c, d, n', m'} \int^{m, n} \hirayo_{\underline{nm} c, \underline{nm} d}(\underline{m'n'} c, \underline{m'n'} d)  \\
\cong & \qquad\mbox{(Definition of the representable functor in $\C^{op} \times \D$)} \\
& \int_{c, d, n', m'} \int^{m, n} \C(\underline{m'n'} c, \nact\mact c) \times \D(\nact\mact d, \underline{m'n'} d).  \\
\end{align*}
The resulting reduction is slightly more general, but we will be interested in
the case where, for all \(m',n'\), we give a natural isomorphism \(\underline{nm} \cong \underline{m' n'}\) for
some \(m,n\).  This induces the required \(\C(\mact\nact c, \underline{m'n'} c) \times \D(\underline{m'n'} d,\mact\nact d)\)
that define a distributive law.  In that case, we will prove that the distributive
law allows us to create a monad that defines a new optic.

\begin{theorem}
\label{org61ce7de}
Let \(\M \to [\C,\C]\), \(\M \to [\D,\D]\) and \(\N \to [\C,\C]\), \(\N \to [\D,\D]\) be two pairs monoidal
actions.  If for all \(m' \in \M\) and \(n' \in \N\) we have a natural isomorphism \(\underline{nm} \cong \underline{m'n'}\)
for some \(m \in \M\) and \(n \in \N\), chosen dinaturally, then \(\Phi_{\N} \circ \Phi_{\M}\) is again a monad.
Moreover, \(\Phi_{\N} \circ \Phi_{\M} = \Phi_{\mathbf{H}}\) for a pair of actions \(\mathbf{H} \to [\C,\C]\) and \(\mathbf{H}\to [\D,\D]\) that
we will construct.
\end{theorem}
\begin{proof}
We have already shown that the first part of the theorem is a valid way of constructing
the distributive law that makes \(\Phi_{\N} \circ \Phi_{\M}\) a monad. It is still left to show that
this new monad is the Pastro-Street monad for some action.

We will construct a pair of actions \(\mathbf{H} \to [\C,\C]\) and \(\mathbf{H} \to [\D,\D]\) such that
\(\Theta_{\M}\circ \Theta_{\N}\) is precisely \(\Theta_{\mathbf{H}}\).  We take \(\mathbf{H} = \N \times \M\) to be a product category,
but we endow it with a monoidal product given by \((n_2,m_2) \otimes (n_1,m_1) = (n_2\otimes n_0, m_0 \otimes m_1)\),
where \(\underline{n_0m_0} \cong \underline{m_2n_1}\).  The action is given by \((n,m) \mapsto \nact\circ \mact\). Because of Fubini,
we have that \(\Theta_{\mathbf{H}} \cong \Theta_\M \circ \Theta_\N\) as functors.
\[ \Theta_{\mathbf{H}} p(c,d) = \int_{(n,m) \in \mathbf{H}} p(\nact\mact c, \nact\mact d) \cong
\int_{m \in \M}\int_{n \in \N} p(\nact\mact c, \nact\mact d).
\]
Note that the comultiplication of the composition is given by the
distributive law and thus we can check that it is the same one \(\mathbf{H}\) uses.
\end{proof}

This is particularly interesting because we can compose pairs of optics
of different kinds whose monads have a distributive law \(d \colon \Phi_{\N}\circ \Phi_{\M} \tonat \Phi_{\M} \circ \Phi_{\N}\). 
If we have a pair of optics of different kinds written as Kleisli
morphisms, \(f \colon \hirayo_{s,t} \tonat \Phi_{\M}\hirayo_{a,b}\) and \(g \colon \hirayo_{a,b} \tonat \Phi_{\N} \hirayo_{x,y}\), regard both as 
Kleisli morphisms for the composite comonad \(\Phi_{\N} \circ \Phi_{\M} = \Phi_{\mathbf{H}}\) as follows.
\[\begin{tikzcd}[row sep=tiny]
l^{\ast} \colon \hirayo_{s,t} \rar{f}& \Phi_{\M}\hirayo_{a,b} \rar{\eta} & \Phi_{\N} \Phi_{\M} \hirayo_{x,y},\\
g^{\ast} \colon \hirayo_{a,b} \rar{\eta}& \Phi_{\M}\hirayo_{a,b} \rar{g} & \Phi_{\N} \Phi_{\M} \hirayo_{x,y} \rar{d} & \Phi_{\M} \Phi_{\N} \hirayo_{x,y}.
\end{tikzcd}\]
We will now discuss concrete examples of this construction.  In \S
\ref{org667a652} we construct two optics from the
fact that products distribute over sums and exponentials distribute
over products.

\subsection{Affine traversals and glasses}
\label{sec:org6be8620}
\label{org667a652}
Let \(\C\) be a bicartesian closed category, which implies that products
distribute over coproducts.  For any \(a,b \in \C\), we have
that \(a \times (b + (-)) \cong (a \times b) + a \times (-)\). The optic given by this action is
called an \emph{affine traversal}, has a concrete representation and was
first conjectured to exist in \cite{pickering17} and considered in
\cite{boisseau18}.   We can use Theorem \ref{org61ce7de} on the actions
of lenses and prisms to see that \(c + d \times (-)\) is an action describing
this new optic that contains both lenses and prisms.

\begin{proposition}
\label{prop:affine}
\textbf{Affine traversals} are optics for the action \(\mathbf{C}^2 \to [\mathbf{C},\mathbf{C}]\) that sends \(c, d \in \mathbf{C}\) to the
functor \(F(a) = c + d \times a\).  They have a concrete description
\[
\mathbf{Affine}
\left( \begin{pmatrix} a \\ b \end{pmatrix}, \begin{pmatrix} s \\ t \end{pmatrix} \right) 
= \mathbf{C}(s , t + a \times (b \to t)).
\]
\end{proposition}
\begin{proof}
\begin{align*}
& \int^{c,d} \mathbf{C}(s , c + d \times a) \times \mathbf{C}((c + d \times b), t) \\
 \cong & \qquad\mbox{{(Coproduct)}} \\
& \int^{c,d} \mathbf{C}(s , c + d\times a) \times \mathbf{C}(c , t) \times \mathbf{C}(d \times b , t) \\
\cong & \qquad\mbox{{(Exponential)}} \\
& \int^{d} \mathbf{C}(s , t + d \times a) \times \mathbf{C}(d , (b \to t)) \\
\cong & \qquad\mbox{{(Yoneda)}} \\
& \mathbf{C}(s , t + a \times (b \to t)). & \qedhere &
\end{align*}
\end{proof}
The affine traversal motivates us to look for other distributive laws.  Using
the theory previously developed, we will use that exponentials distribute over
products to create a new concrete optic. Let \(\C\) be a cartesian closed category.
For any \(a,b \in \C\), we can distribute exponentials over products as
\(a \to (b \times (-)) \cong (a \to b) \times (a \to (-))\).  We will call \emph{glass} to the 
kind of optic
generated by this action; lenses and grates are, in particular, glasses.
To the best of our knowledge, neither a description nor a derivation
of this optic were present in the literature; it remains to be seen what are its
potential applications.

\begin{proposition}
\label{prop:glasses}
\textbf{Glasses} are optics for the action \(\mathbf{C}^2 \to [ \mathbf{C} , \mathbf{C} ]\) that sends \(c, d \in \mathbf{C}\) to the
functor \(F(a) = d \times (c \to a)\).  They have a concrete description
\[
\mathbf{Glass}
\left( \begin{pmatrix} a \\ b \end{pmatrix}, \begin{pmatrix} s \\ t \end{pmatrix} \right) = 
\mathbf{C}(s \times ((s \to a) \to b) , t).
\]
\end{proposition}
\begin{proof}
\begin{align*}
& \int^{c,d} \mathbf{C}(s, c \times (d \to a)) \times \mathbf{C}((d \to b) \times c , t) \\
\cong & \qquad \mbox{{(Product)}} \\
& \int^{c,d} \mathbf{C}(s, c) \times \mathbf{C}(s , d \to a) \times \mathbf{C}((d \to b) \times c , t) \\
\cong & \qquad \mbox{{(Yoneda lemma)}}\\
& \int^{d} \mathbf{C}(s , d \to a) \times \mathbf{C}((d \to b) \times s , t) \\
\cong & \qquad \mbox{{(Product-Exponential)}}\\
& \int^{d} \mathbf{C}(d , s \to a) \times \mathbf{C}((d \to b) \times s , t) \\
\cong & \qquad \mbox{{(Yoneda lemma)}}\\
& \mathbf{C}(((s \to a) \to b) \times s , t). & \qedhere &
\end{align*}
\end{proof}

\section{Joining Tambara modules}
\label{sec:orgd23b5b7}
\label{orga66504e}
In Haskell, the composition of two optics of two different kinds is done as
follows.  Assume we have monoidal actions \(\alpha \colon \M \to [\C,\C]\) and \(\beta \colon \N \to [\C,\C]\),
determining two different kinds of optics. By virtue of Theorem \ref{org9a4afc7}, the optics can be written as functions
that will be polymorphic over Tambara modules for \(\M\) and \(\N\); that is, over coalgebras 
for \(\Theta_{\M}\) and \(\Theta_{\N}\).
\[\int_{p \in \hirata(\alpha)} \Sets(p(a,b), p(s,t)),\qquad
\int_{q \in \hirata(\beta)} \Sets(q(x,y), q(a,b)).\]
When we compose them, Haskell outputs a function that is polymorphic over profunctors
with Tambara module structure for both actions. These are profunctors with structure
of \emph{bicoalgebra}: coalgebra structure for both \(\Theta_{\M}\) and \(\Theta_{\N}\).
\[
\int_{(p,p) \in \hirata(\alpha) \times_{\Prof} \hirata(\beta)} \Sets(p(a,b), p(s,t))
\]
If we want to see this new function as an optic, we need to describe these bicoalgebras 
as coalgebras for another Pastro-Street comonad.  The action that will generate this 
new comonad will be the coproduct of the two actions we had, \(\alpha + \beta \colon \M + \N \to [\C,\C]\).
\[\begin{tikzcd}
\M \drar[swap, bend right]{\alpha}\rar{i} & \M + \N \dar[dashed]{\exists ! \alpha + \beta} & \N \lar[swap]{i} \dlar[bend left]{\beta} \\
& \left[ \C,\C \right]  &
\end{tikzcd}\]
That is, we are claiming that, for a fixed profunctor \(p\), a pair of coalgebra structures 
over \(\Theta_\M\) and \(\Theta_{\N}\) is exactly the same as a coalgebra structure over \(\Theta_{\M + \N}\). If
we call \(\hirata_{\M}p\) to the set of Tambara structures over \(p\) for the monoidal action
of \(\M\), we are claiming that \(\hirata_{\M+\N} p \cong \hirata_{\M}p \times \hirata_{\N}p\). In order to show this,
we will prove that \(\hirata_{(-)}p \colon \mathbf{MonCat}/[\C,\C] \to \Sets\) is representable.

\begin{theorem}
Let \(p \in \Prof(\C,\D)\) be a profunctor.  The associated functor
\[\hirata_{(-)}p \colon \mathbf{MonCat}/[\C,\C] \to \Sets
\]
is representable.
\end{theorem}
\begin{proof}
We will show that \(\hirata_{\M}p \cong \mathbf{MonCat}/[\C,\C](\M, \mathbf{D}_{p})\) for a monoidal action
\(\gamma \colon \mathbf{D}_p \to [\C,\C]\) from a category \(\mathbf{D}_p\) we are going to construct.  We define
the objects of \(\mathbf{D}_p\) to be pairs \((f, \eta)\) with \(f \in [\C,\C]\) and \(\eta \in \int\nolimits_{a,b} p(a,b) \to p(fa,fb)\).
Morphisms of \(\mathbf{D}_p\) between \((f,\eta)\) and \((g,\eta')\) are natural transformations
\(\alpha \colon f \tonat g\) such that the following triangle commutes.
\[\begin{tikzcd}[row sep=tiny]
& p(fa,fb) \ar{dd}{p(\alpha,\alpha)} \\
p(a,b) \urar{\eta}\drar[swap]{\eta'} & \\
& p(ga,gb)
\end{tikzcd}\]
Finally, the monoidal product of \(\mathbf{D}_p\) is given by \((f,\eta) \otimes (g,\eta') = (f \circ g , \eta \circ \eta')\).
The action \(\mathbf{D}_p \to [\C,\C]\) simply projects the first component.

A Tambara module over \(p\) is precisely choosing some \(p(a,b) \to p(fa,fb)\) for
every \(m \in \M\) such that \(\mact = f\).  This needs to be done in a natural way,
which is precisely the dinaturality of the Tambara module, and preserving the
monoidal structure, which precisely gives the Tambara axioms.
\end{proof}

The fact that it is representable implies \(\hirata_{\M+\N} p \cong \hirata_{\M}p \times \hirata_{\N}p\).
The conclusion is that, when Haskell joins the constraints of two profunctor
optics, it is creating a profunctor optic for the coproduct action.  By analogy
with the coproduct monoid, the coproduct action from \(\M+\N\) can be thought as
a string of objects of \(\M\) and \(\N\) acting consecutively.  For instance, if we
consider the actions given by the cartesian product and the disjoint union
of sets, the coproduct of these actions is an action taking a word of sets
\(x_1,y_1,\dots,x_n,y_n\) and acting on a set \(a\) as in the following formula.
\[
x_1 + y_1 \times (x_2 + y_2 \times (\dots (x_n + y_n \times a) \dots) )
\]

\begin{remark}
Algebras for the coproduct monad are precisely pairs of algebra structures
for each one of the components \cite{kelly80} \cite{adamek12}.  Coalgebras for the product comonad
are precisely pairs of coalgebra structures for each one of the components.
This suggests that the previous result is just trying to say \(\Theta_{\M + \N} \cong \Theta_{\M} \times \Theta_{\N}\),
where the product must be understood as the product in the category of comonads.
\end{remark}

\section{Clear optics and the lattice of optics}
\label{sec:orgaf5b80b}
\label{org7bbc905}
The two previous sections pose a problem: Haskell seems to be joining lenses and prisms into
an optic for a complicated action of the form \(x_1 + y_1 \times (\dots (x_n + y_n \times a) \dots)\),
but we actually only care about the action that determined affine traversals, \(x + y \times a\),
which is much simpler.  The relation between these two actions is that any action
with the form of the first one is always isomorphic in the category \([\C,\C]\) to some
action of the second form.  It is folklore that lenses and prisms \emph{compose} into affine
traversals but, to make that precise in our description of optics, we are implicitly 
quotienting by some isomorphisms.
We would like to know which quotienting is being done and have a formal description of
this construction.

A problem that looks vaguely related is that we can artificially
include irrelevant information in our optics.  Consider for instance
an adapter, described by the action of the functor from the terminal
category \(\mathbf{1} \to [\C,\C]\) picking the identity functor.  Consider now the
action given by the monoid (the discrete monoidal category) of natural
numbers \(\mathbb{N} \to [\C,\C]\) sending \(1 \in \mathbb{N}\) to the functor \((1 \times (-)) \in [\C,\C]\).
Both are essentially describing adapters, but the second one will \emph{``remember''}
the number that was used to construct it.  An element of the first optic
is just an adapter, while an element of the second optic is an adapter
together with a natural number.  The problem is that, inside \(\mathbb{N}\), the
functors \(1 \times (1 \times (-))\) and \(1 \times (-)\) are regarded as non-isomorphic;
that is,
\[
\int^{1} \C(s,a) \times \C(b,t) \not\cong \int^{n \in \mathbb{N}} \C(s,a) \times \C(b,t).
\]
One would think that these problems would dissapear when, instead of
considering arbitrary monoidal actions as in \cite{riley18}, we limit
the possible actions to be given by subcategories of endofunctors as
in \cite{boisseau18} and moreover we require them to be full.  However, that
comes with its own set of problems: the category where we are applying
Yoneda to get the concrete representations is \emph{not} always a full
subcategory of endofunctors.  We propose an intermediate way between these two
definitions. We will consider just the optics defined by \emph{pseudomonic} actions (as
in Definition \ref{orged771c2}): 
actions that determine subcategories that are full on isomorphisms.  
The second problem is solved because the
artificial description of adapters is not pseudomonic; the first problem
can be solved quotienting to the smallest pseudomonic action, which is
precisely the one that describes affine traversals.  Let us detail this
approach to what we will call \textbf{clear optics}.

\subsection{Clear optics}
\label{sec:org25e6ba0}
\begin{definition}
\label{org7d8d7fb}
An kind of optic is \textbf{clear} if it is given by a monoidal action \(\M \to [\C,\C]\)
that is a pseudomonic functor as in Definition \ref{orged771c2}.
\end{definition}

In other words, if we consider only replete subcategories, clear optics are 
these given by the subobjects of \([\C,\C]\) in the category \(\mathbf{MonCat}\).  
Let us show that some common optics are clear at least in the case of \(\C = \Sets\).

\begin{lemma}
Lenses are a clear optic.
\end{lemma}
\begin{proof}
\label{orgc3863a4}
We want to show that the action of the product is pseudomonic.
For some pair of sets \(C,D \in \Sets\), assume an isomorphism \(u \colon (C \times A) \cong (D \times A)\)
natural in \(A\).  We need to show that it is the image of some isomorphism \(C \cong D\)
under the action.  We take \(v \colon C \cong C \times 1 \cong D \times 1 \cong D\).  The following
diagram commutes by naturality.
\[\begin{tikzcd}
C \times A \dar{u_{A}} & C \times 1 \lar[swap]{\id \times a} \dar{u_1}& C \lar[swap]{\cong} \dar{v} \\
D \times A & D \times 1 \lar{\id \times a} & D \lar{\cong}
\end{tikzcd}\]
For every \(a \colon 1 \to A\), we have that \(u_A(c,a) = (v(c),a)\). That shows that
\(u = (v \times -)\), and thus the action of the product is pseudomonic.
\end{proof}

\begin{lemma}
Prisms are a clear optic.
\end{lemma}
\begin{proof}
\label{org0ababde}
We want to show that the action of the coproduct is pseudomonic.
For some pair of sets \(C,D \in \Sets\), assume an isomorphism \(u \colon (C + A) \cong (D + A)\)
natural in \(A\). We need to show that it is the image of some isomorphism \(C \cong D\)
under the action.  We take \(v \colon C \cong C + 0 \cong D + 0 \cong D\) and because of naturality,
the isomorphism \(C + A \cong D + A\) must behave like \(v\) on elements of \(C\).
\[\begin{tikzcd}
C + A \dar{u_{A}} & C + 0 \lar[swap]{\id \times a} \dar{u_1}& C \lar[swap]{\cong} \dar{v} \\
D + A & D + 0 \lar{\id \times a} & D \lar{\cong}
\end{tikzcd}\]
An isomorphism between sets \(C + A \cong D + A\) that behaves like an isomorphism 
\(C \cong D\) in elements of \(C\) must split into this isomorphism and an isomorphism \(A \cong A\).
Being natural in \(A\), this needs to be the identity because of Yoneda lemma.
\end{proof}

\begin{lemma}
Affine traversals are a clear optic.
\end{lemma}
\begin{proof}
For some sets \(C,D,C',D' \in \Sets\), consider an isomorphism \(D + C \times A \cong D' + C' \times A\).
Taking \(A = 0\) and following the same reasoning as in Lemma \ref{org0ababde}, ew
get an isomorphism \(C \times A \cong C' \times A\) natural in \(A\).  Following \ref{orgc3863a4}, we get
an isomorphism \(C \cong C'\) that together with \(D \cong D'\) shows the original isomorphism
as image under the action of these.
\end{proof}

\begin{remark}
Not every optic we have considered so far is clear.
Glasses, for instance, are not a clear optic because \(0 \times (1 \to a) \cong 0 \times (2 \to a)\)
while \(1 \not\cong 2\).  We could, however, take the replete image of the action
determining glasses and define a clear version of the same optic.
Given an action \(\alpha \colon \M \to [\C,\C]\) determining a category of optics
\(\Optic_{\M}\), we can consider the replete image (as in \S \ref{orga37ce19}) of \(\M\) and take the
action \(\repl(\img(\alpha)) \colon \repl(\img(\M)) \to [\C,\C]\).  This is a clear
optic because of Proposition \ref{orge16820a}.  The inclusion
\(\M \to \repl(\img(\M))\) is determining a morphism of actions that,
by virtue of Lemma \ref{orgd31a70b}, induces a functor
\(\Optic_{\M} \to \Optic_{\repl(\img(\M))}\).  This functor is full, and this can
be checked realizing that the hom-sets of the second category are just
coends with extra conditions and that the map with the following signature
is a surjection.
\[
\left( \int^{m \in \M} \C(s,\mact a) \times \C(\mact b , t) \right) \to 
\left( \int^{m \in \repl(\img(\M))} \C(s,\mact a) \times \C(\mact b , t) \right)
\]
In other words, there exists a clear version of every optic.  Even
when we have derived concrete forms only for the original version, the
same concrete forms work for the clear version, up to some
identifications.
\end{remark}

\subsection{Composing clear optics}
\label{sec:org510d6d1}
We want to reconcile here how Haskell joins Tambara modules to compose
profunctor optics, as in \S \ref{orga66504e}, with the intuition that lenses and
prisms should compose into affine traversals, as in \ref{prop:affine}.

We have shown that, after joining the Tambara structures of a
profunctor lens with a profunctor prism, we should get a profunctor
optic for the coproduct action.  The coproduct action 
takes a word of sets \(x_1,y_1,\dots,x_n,y_n\) and acts on a set \(a\) as in the following formula.
\[
x_1 + y_1 \times (x_2 + y_2 \times (\dots (x_n + y_n \times a) \dots) )
\]
The category giving this action is \emph{not} equivalent to the one given by affine
traversals \(x + y \times (-)\). The morphisms on the first category are componentwise
morphisms of the words, whereas a morphism \(x_1,y_1,\dots,x_n,y_n \to x_1',y_1',\dots,x_n',y_n'\) is
given by a family of morphisms \(x_i \to x_i'\) and \(y_i \to y_i'\).  This means that, contrary
to intuition, \(1 + 1 \times (1 + \times (-)) \not\cong 2 + (-)\) in this category.  The problem is
solved if we consider the corresponding clear optic: the repletion of this category is the same
as the repletion of the category giving affine traversals. Moreover, because
affine traversals are a clear optic, that means that they are equivalent as
categories.  To sum up, when we join Tambara modules, lenses and prisms compose 
into some complicated optic; however, the clear version of this optic is equivalent (in the categorical sense)
to the one describing affine traversals.  If we are willing to accept that we are working with clear
optics, we can say that Haskell composes lenses and prisms into affine traversals
in a sound way that is compatible with the Tambara structure.

\subsection{Lattice of optics}
\label{sec:orgb118f3f}
If clear optics are given by replete subcategories of \([\C,\C]\), we
can consider the lattice they form and translate it via Lemma \ref{orgd31a70b}
to a lattice of optics.  We aim to recover and expand the lattice of optics
that was first described in \cite{pickering17}.

We start by constructing the lattice of replete subcategories. We will
be working in the full subcategory of \(\mathbf{MonCat}/[\C,\C]\) where objects
are pseudomonic. Given two objects \(\alpha \colon \M \to [\C,\C]\) and \(\beta \colon \N \to [\C,\C]\),
let \(\alpha \wedge \beta\) be their product on this category, which is the following
pullback.
\[\begin{tikzcd}[column sep=tiny]
& \M \wedge \N \dlar[swap]{\pi}\drar{\pi} & \\
\M\drar[swap]{\alpha} && \N \dlar{\beta} \\
& \left[ \C,\C \right] &
\end{tikzcd}\]
This is again pseudomonic because pseudomonic functors are stable
under pullback. In terms of replete subcategories, we are taking
an intersection of the two categories.  Let \(\alpha \vee \beta\) be the coproduct on this category;
it is given by the smallest replete subcategory containing both.
Because of the construction of the replete image, this is precisely
the replete image of the coproduct, \(\repl(\img(\alpha + \beta))\), as in the
following diagram.
\[\begin{tikzcd}
\M\rar{i} \ar[bend right]{ddr}[swap]{\alpha} & \M + \N \dar{\repl\img} & \N\lar[swap]{i} \ar[bend left]{ddl}{\beta} \\
   & \M \vee \N \dar & \\
   & \left[\C , \C\right] 
\end{tikzcd}\]
We will actually construct a bounded lattice.  The top element
is given by the identity action \(\top \colon [\C,\C] \to [\C,\C]\), and the
bottom element is given by the action of the trivial monoidal
category \(\bot \colon \mathbf{1} \to [\C,\C]\). 

\begin{lemma}
\label{org17b89ae}
The previously defined operations \((\wedge, \vee, \top, \bot)\) form a bounded lattice
up to isomorphisms.
\end{lemma}
\begin{proof}
First, we note that \((\wedge, \vee)\) are commutative and associative
because they are products and coproducts. The meet \((\wedge)\) is
idempotent, \(\alpha \wedge \alpha = \alpha\), because pseudomonic functors are precisely those whose
pullback with themselves is again themselves, see Definition
\ref{orged771c2}.  The join \((\vee)\) is idempotent, \(\alpha \vee \alpha\), because of the construction of
the repleted image. As subcategories,
the image of \(\alpha \wedge \beta\) is necessarily contained in \(\alpha\), and that gives
\(\alpha \vee (\alpha \wedge \beta) = \alpha\).  On the other hand, \(\alpha\) is necessarily contained
in \(\alpha \vee \beta\), and that gives \(\alpha \wedge (\alpha \vee \beta) = \alpha\).

Every monoidal action must send the unit of the monoidal category
to the identity, this ensures \(\alpha \vee \bot = \alpha\).  Because they are all
subcategories of endofunctors, \(\alpha \wedge \top = \alpha\).
\end{proof}

\begin{theorem}
The lattice induces functors between categories of clear optics. 
We have the following hierarchy of optics.
\end{theorem}
\begin{proof}
We can consider many evident inclusions between the different actions
that define the optics in the lattice defined in Lemma \ref{org17b89ae}.
For instance, affine maps can be written as power series functors, so
the action defining affine traversals is contained in the one defining
traversals.  Because of Lemma \ref{orgd31a70b}, that lattice is transported to
functors between categories of optics.  Note that the bottom and the top
correspond to \(\mathbf{Adapter}\) and \(\mathbf{Setter}\), respectively.
\[\begin{tikzcd} && \mathbf{Setter} & \\ && \mathbf{Traversal} \ar{u}&& \mathbf{Kaleidoscope} \ar{ull}
\\ & \mathbf{Glass}\ar{uur} & \mathbf{Affine} \ar{u}& &
\\ \mathbf{Grate} \ar{ur} & \mathbf{Prism}\ar{ur}&& \mathbf{Lens}
\ar{ull}\ar{ul} &  \\ &&  & \mathbf{ListLens} \uar \ar{uuur} &
\\ && \mathbf{Adapter} \ar{uull}\ar{uul}\ar{ur} &
 \end{tikzcd}\]
Note that these are not the only optics we have considered. It is interesting
to compare how this matches the original lattice of optics in \cite{pickering17} and the
one in \cite{boisseau17}, described for Haskell constraints.
\end{proof}

\chapter{Applications}
\label{sec:org99f7cff}

\section{Library of optics}
\label{sec:org373af17}
\label{orge864844}
As an application, we have implemented a Haskell library that makes
use of the general form of the profunctor representation theorem
(Theorem \ref{org9a4afc7}) to provide a translation
between the existential and profunctor representations, parametric in
the type of optic.  This development follows closely the definitions in
\cite{boisseau18}, in that it considers optics as given by submonoids of
endofunctors instead of monoidal actions.

In order to achieve this, we consider the submonoid of endofunctors
described by the monoidal action as given by some Haskell \emph{constraint}. Constraints
are then passed to the relevant data constructors (in GADT style) using
the \texttt{ConstraintKinds} extension of the Glasgow Haskell Compiler.  The complete
code for the full description and translation between the two representations
of optics is surprisingly concise and it is included here with references
to the relevant theorems.

\begin{itemize}
\item Existential optic as in Definition \ref{org8ab3812}.

\begin{minted}[]{haskell}
data ExOptic mon a b s t where
  ExOptic :: (mon f) => (s -> f a) -> (f b -> t) 
                     -> ExOptic mon a b s t
\end{minted}

\item A class that witnesses that the given constraint declares a
submonoid of endofunctors. Note that the following morphisms can be
constructed when the constraint is defining a class of endofunctors
closed under identity and composition, see the construction of the
category of optics in \S \ref{org2dbf2a9}.

\begin{minted}[]{haskell}
class MonoidalAction mon where
  idOptic   :: ExOptic mon a b a b
  multOptic :: (mon f) => ExOptic mon a b s t 
                       -> ExOptic mon a b (f s) (f t)
\end{minted}

\item Definition of a Tambara module (Defintion \ref{org0779c9e}) and the
profunctor representation.

\begin{minted}[]{haskell}
class (Profunctor p) => Tambara mon p where
  action :: forall f a b . (mon f) => p a b -> p (f a) (f b)

type ProfOptic mon a b s t = 
  forall p . Tambara mon p => p a b -> p s t
\end{minted}

\item The equivalence given by the profunctor representation theorem
(Theorem \ref{org9a4afc7}) is constructed
explicitly here.

\begin{minted}[]{haskell}
instance Profunctor (ExOptic mon a b) where
  dimap u v (ExOptic l r) = ExOptic (l . u) (v . r)

instance (MonoidalAction mon) => Tambara mon (ExOptic mon a b) where
  action = multOptic

crOptic :: (MonoidalAction mon) => 
  ProfOptic mon a b s t -> ExOptic mon a b s t
crOptic p = p idOptic

mkOptic :: forall mon a b s t . 
  ExOptic mon a b s t -> ProfOptic mon a b s t
mkOptic (ExOptic l r) = dimap l r . (action @mon)
\end{minted}
\end{itemize}

Note that the code included here is enough for implementing the
translation between representations of optics in general.  We believe
this can lead to more concise optic libraries in the future.
Completing the library, including an implementation of the full range
of the combinators that optics libraries such as Kmett's \cite{kmett15}
provide, is left for further work.

\subsection{Example usage}
\label{sec:org479f013}
The code, together with implementations for the most common optics and
examples of usage, can be found at the following HTML address.

\begin{itemize}
\item \url{https://github.com/mroman42/vitrea/}
\end{itemize}

We take seriously Examples \ref{org729dff4}, \ref{org47c3ee8} and
\ref{orgadaa07e} to give an idea of how this library works. The exact
code for these examples can be found in the library.

\begin{exampleth}[Lenses and prisms]
Examples \ref{org729dff4}, \ref{org47c3ee8} show how we can compose
a \emph{lens} (\texttt{street}) with a \emph{prism} (\texttt{postal}) and still be able
to access the contents.  In the first code line we declare a string called \texttt{address}.
In the second line, we match it into a postal address using the prism
\texttt{postal}; this step could have failed.  Then we compose a prism with a lens
to get the affine traversal \texttt{postal.street}, that accesses the street inside
a string.  We can update that internal field and have the changes
propagate upwards.  The same prism can be reused with a different lens
(\texttt{city}) to access a different part of the address. This solution is modular.
\end{exampleth}
\begin{minted}[]{haskell}
  let address = "45 Banbury Rd, OX1 3QD, Oxford"
  address^.postal
    -- Street:  45 Banbury Rd
    -- Code:    OX2 6LH
    -- City:    Oxford
  address^.postal.street
    -- "45 Banbury Rd"
  address^.postal.street %~ "7 Banbury Rd"
    -- "7 Banbury Rd, OX1 3QD, Oxford"
  address^.postal.city <>~ " (UK)"
    -- "45 Banbury Rd, OX1 3QD, Oxford (UK)"
\end{minted}

\begin{exampleth}[Traversals]
The traversal (\texttt{mail}) from Example \ref{orgadaa07e} can be
used how to compose traversals with lenses.  In the first line of
code, we have a mailing list. In the second line of code, we extract
the emails using the traversal. In the third line of code we compose
it with two lenses (\texttt{email}, \texttt{domain}) to access particular subfields.  Note
that lenses and traversals compose again to traversals (Lemma \ref{org17b89ae}). In
the last line of code, we apply an \texttt{uppercase} function the contents of
these subfields; the changes propagate to the initial mailing list.
\end{exampleth}

\begin{minted}[]{haskell}
  someMailingList
    -- | Name           | Email                   | Frequency |
    -- |----------------+-------------------------+-----------|
    -- | Turing, Alan   | turing@manchester.ac.uk | Daily     |
    -- | Noether, Emily | emmynoether@fau.eu      | Monthly   |
    -- | Gauss, Carl F. | gauss@goettingen.de     | Weekly    |
  someMailingList^..mails
    -- ["turing@manchester.ac.uk", "emmynoether@fau.eu", "gauss@goettingen.de"]
  someMailingList^..mails.email.domain
    -- ["manchester.ac.uk", "fau.eu", "goettingen.de"]
  someMailingList^..mails.email.domain <~~~ uppercase
    -- | Name           | Email                   | Frequency |
    -- |----------------+-------------------------+-----------|
    -- | Turing, Alan   | turing@MANCHESTER.AC.UK | Daily     |
    -- | Noether, Emily | emmynoether@FAU.EU      | Monthly   |
    -- | Gauss, Carl F. | gauss@GOETTINGEN.DE     | Weekly    |
\end{minted}

\subsection{A case study}
\label{sec:org9371e5a}
The code presented at the beginning of this section makes the profunctor
representation theorem works for any arbitrary action. We can go beyond
the usual optics and elaborate an example that starts from the observation
of Remark \ref{org2930c80}: a \emph{list lens} can be composed with a \emph{kaleidoscope}
to create a new \emph{kaleidoscope}.   

\begin{itemize}
\item We start by defining \emph{kaleidoscopes} as the optic for applicative
functors. The definition of a new optic is simple under this
framework: we only need to witness the fact that applicative
functors are closed under composition and that they contain the
identity functor.
\begin{minted}[]{haskell}
type ExKaleidoscope   s t a b = ExOptic   Applicative a b s t
type ProfKaleidoscope s t a b = ProfOptic Applicative a b s t

instance MonoidalAction Applicative where
  idOptic = ExOptic Identity runIdentity
  multOptic (ExOptic l r) = ExOptic 
    (Compose . fmap l) (fmap r . getCompose)
\end{minted}
List lenses are defined in a similar fashion as the optic for
the product of a monoid.
\end{itemize}

Now assume we have a large dataset where, instead of reading all the entries,
we want to learn about the aggregate data.  In our case we pick the standard
\emph{iris} dataset \cite{fisher36}.  The entries of the dataset represent a \(\mathbf{Flower}\),
characterized by its \(\mathbf{Species}\) (which can be \emph{Iris setosa}, \emph{Iris versicolor} and \emph{Iris virginica})
and four numeric \(\mathbf{Measurements}\) on the length and width of its sepal and petal.
One could take \(\mathbf{Species} = \mathbf{3}\), \(\mathbf{Measurements} = \mathbb{R}_{+}^4\) and \(\mathbf{Flower} = \mathbf{Species} \times \mathbf{Measurements}\).

We consider a type-invariant \emph{list lens} of the following type.
\[ \mathrm{measureNearest} \in \mathbf{ListLens}((\mathbf{Flower},\mathbf{Flower}), (\mathbf{Measurements} , \mathbf{Measurements})),\]
It is given by a \emph{view} function \(\mathbf{Flower} \to \mathbf{Measurements}\), that projects the measurements of
a flower; and a \emph{classify} function \(\mathbf{Flower}^{\ast} \times \mathbf{Measurements} \to \mathbf{Flower}\), that
takes a list of flowers and uses it to classify the input measurements into some
species, outputting a new flower.  Internally, our classify function will use the
\emph{1-nearest neighbour algorithm} \cite{cover67}.  We also consider a type-invariant
\emph{kaleidoscope} of the following type.
\[ \mathrm{aggregate} \in \mathbf{ListLens}((\mathbf{Measurements},\mathbf{Measurements}), (\mathbb{R}_+ , \mathbb{R}_+)),\]
It is given by a function that takes as an input some monoid structure in
the positive real numbers, given as a map \(\mathbb{R}_+^{\ast} \to \mathbb{R}\) from the free monoid,
and induces a componentwise monoid structure on \(\mathbf{Measurements}\).

\begin{itemize}
\item List lenses are, in particular, lenses; we can use them to \emph{view} the
measurements of the first element of our dataset.
\begin{minted}[]{haskell}
  (iris !! 1)^.measureNearest
    --      Sepal length: 4.9
    --      Sepal width:  3.0
    --      Petal length: 1.4
    --      Petal width:  0.2
\end{minted}
\item They are more abstract than a lens in the sense that they can be used
to classify some measurements into a new species taking into account
all the examples of the dataset.
\begin{minted}[]{haskell}
  (iris ?. measureNearest) (Measurements 4.8 3.1 1.5 0.1)
    --  Flower: 
    --      Sepal length: 4.8
    --      Sepal width:  3.1
    --      Petal length: 1.5
    --      Petal width:  0.1
    --      Species:      Iris setosa
\end{minted}

\item Now we compose the list lens with the kaleidoscope and we get back a new
kaleidoscope that, from a monoid structure on the positive real numbers,
induces first a monoid structure on the measurements, uses it to 
aggregate the measurements of the dataset, and finally classifies the aggregated
measurements into a species.
\begin{minted}[]{haskell}
  (iris >- measureNearest.aggregate) mean
    -- Flower: 
    --     Sepal length: 5.843
    --     Sepal width:  3.054
    --     Petal length: 3.758
    --     Petal width:  1.198
    --     Species:      Iris versicolor
\end{minted}

\item Any other monoidal structure will work the same. Here we use the
maximum instead of the average.
\begin{minted}[]{haskell}
  (iris >- measureNearest.aggregate) maximum
    -- Flower: 
    --     Sepal length: 7.9
    --     Sepal width:  4.4
    --     Petal length: 6.9
    --     Petal width:  2.5
    --     Species:      Iris virginica
\end{minted}
\end{itemize}

\section{Agda implementation}
\label{sec:org099461e}
\label{orgcfdf01e}
We have also developed an Agda implementation of the library. The Agda
language \cite{norell08} is dependently typed, and the enhanced expressivity of the type
system justifies some changes from the Haskell version.  As a major
difference, optics are defined as arising from a monoidal action (as
in \cite{riley18}) instead of a submonoid of endofunctors.  The code can 
be found in the following link.

\begin{itemize}
\item \url{https://github.com/mroman42/vitrea-agda}
\end{itemize}

Foundations become relevant when writing a formalization of our
reasoning, and there is an important aspect of this text that we have
not discussed yet: even when we have been agnostic regarding
foundations, the theory of optics is \emph{constructive} in nature (in the
sense of \cite{troelstra14}).  The reader can check that we have avoided
the use of the excluded middle or the axiom of choice (which implies
excluded middle). There is a clear exception in Proposition
\ref{orge16820a}, but we are not considering it part of the main theory
and we think the use of the axiom of choice could be avoided in this
case considering \emph{anafunctors}, which are motivated by foundational
concerns exactly like this one (see \cite{makkai96}).  This makes it
possible to formalize parts of this text in some variant of Martin-Löf
type theory using a proof assistant.

\begin{itemize}
\item The library is built in turn over a small library we have developed
to deal automatically with trivial isomorphisms in the category
\(\mathbf{Sets}\).  The automation works on top of Agda's instance
resolution algorithm, and given two types depending on some
variables, it tries to find a isomorphism between them, natural on
the variables.

\item There is a partial formalization and a construction of the
bijection of the profunctor representation theorem (Theorem
\ref{org9a4afc7}). It provides a definition of monoidal
action (Definition \ref{org7b5a63d}) and of Tambara module over an
arbitrary action (Definition \ref{org0779c9e}).  From this partial
formalization, the algorithm that translates between the existential
and the profunctor form of an optic can be extracted.

\item The formal coend derivations are explicitly written down in the
code of the library.  From these proof-relevant derivations, the
algorithm that translates between the concrete and the existential
form can be extracted.  An example of the code for such a
derivation can be seen in Figure \ref{fig:orgd409da5} and Figure
\ref{fig:org88a5bb3}. Compare this to the proofs of Proposition \ref{prop:lenses}
and Proposition \ref{prop:prisms}.  This makes the library extremely
close to the theory and makes it possible for the code to justify
its own correctness.
\end{itemize}

\begin{figure}[htbp]
\centering
\includegraphics[width=12cm]{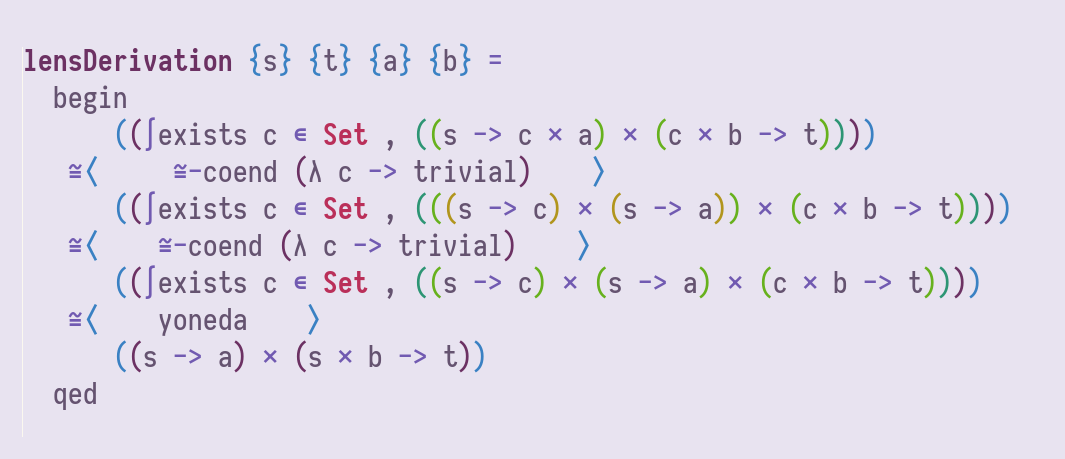}
\caption{\label{fig:orgd409da5}
(Co)end derivation of the concrete form of a lens in Agda. Compare it to the proof of Proposition \ref{prop:lenses}.}
\end{figure}

\begin{figure}[htbp]
\centering
\includegraphics[width=12cm]{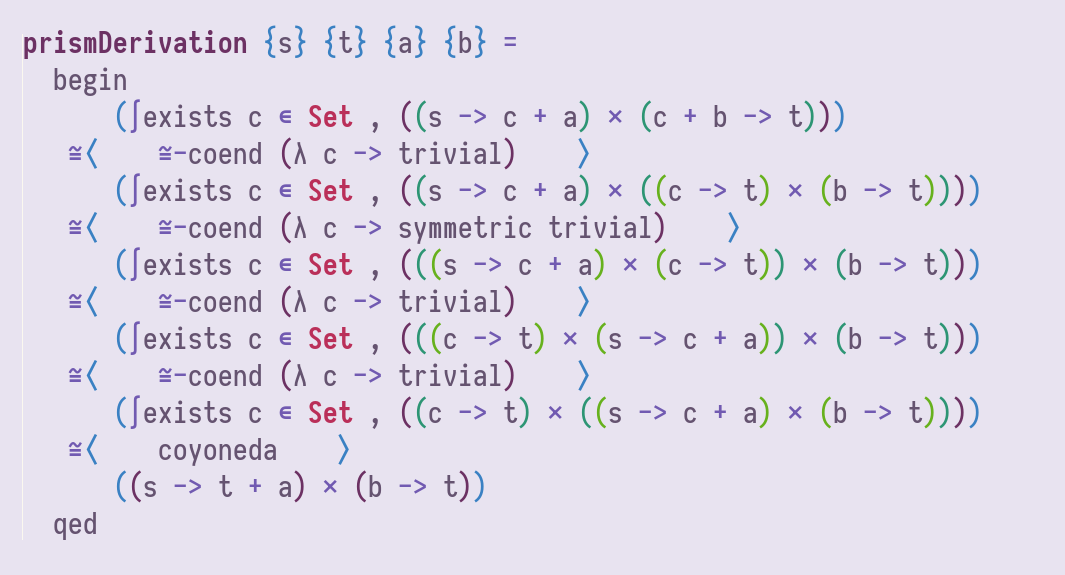}
\caption{\label{fig:org88a5bb3}
(Co)end derivation of the concrete form of a prism in Agda. Compare it to the proof of Proposition \ref{prop:prisms}.}
\end{figure}
\chapter{Conclusions}
\label{sec:org5fa989a}

\section{Conclusions and further work}
\label{sec:orga31939e}
The definition of optics in terms of monoidal actions \cite{riley18}
\cite{boisseau18} captures in an elementary way all the basic optics,
including traversals (Proposition \ref{prop:traversal}).  It is also a
valuable tool for constructing new optics (\S \ref{org087c8c4}), even
if getting an interesting concrete form is not done in general.
When studying the profunctor representation of optics, there seem to
be reasons to prefer the more restrictive notion of \emph{clear optics}
(Definition \ref{org7d8d7fb}).  We think it is interesting to
study how the different variations on the notion of \emph{submonoidal category} 
of \([\C,\C]\) give rise to different variants of the definition of optic.

Among optics, the case of \textbf{traversals} was particularly interesting to
us.  Milewski \cite{milewski17} posed the problem of finding an
\emph{as elementary as possible}
description of the traversal.  Our
derivation in Proposition \ref{prop:traversal} tries to achieve precisely
this goal.  The derivation described in \cite{riley18} for \emph{traversables}
was using a parameterised comonad from \cite{jaskelioff15}, so it made
sense to go back and try to simplify that approach with the
\emph{shape-contents} comonad we described in \S \ref{orgab52623}. Finally,
describing power series functors as linear species is what guided us
to the definition of \emph{unsorted} \emph{traversal} in Proposition
\ref{prop:unsortedtraversal}.   There is still work to be done for traversals
that we summarize in the following questions.

\begin{itemize}
\item Can we explain the description of traversables in terms of
monoidal, cartesian and cocartesian functors from \cite{pickering17}?
Note that \emph{monoidal} profunctors do not fit into the usual pattern of
Tambara modules we have been discussing so far.
\end{itemize}

\emph{Lawful optics} are a topic that we have decided not to consider in this
text. They are studied in detail in the brilliant work of Mitchell Riley \cite{riley18}. It would be
particularly interesting to check what are sensible laws for the
optics we have derived.

Combining different optics is one of the main motivations of the
theory.  With the description of the lattice in \S \ref{org7bbc905}, they
form a family of intercomposable bidirectional accessors.  Going up
in the lattice we forget the structure of our context, going down in
the lattice we make it more explicit.  Lenses, with their many
applications (\cite{ghani18}, \cite{fong19}) seem to be a sweet spot in
this hierarchy, but we believe the whole family of optics can have
potential applications.  Part of the work that remains to be done is
to determine what are the practical applications of many optics we are
describing, and to develop useful intuitions for them.

A final direction is to generalize the theory of optics.  There is a
particular way of doing so that seems to follow naturally from the
proofs in this text: repeat the same theory in the context of \emph{enriched}
category theory.  The reader can check that most of the text would work
the same if we substitute \(\mathbf{Sets}\) for an arbitrary cartesian Benabou cosmos \(\V\).
In some sense, this captures better its computational flavour, as it shows
that the same theory could have been developed, for instance, over the category of 
directed-complete partial orders \(\mathbf{Dcpo}\) (which is cartesian closed, complete
and cocomplete \cite[Theorem 3.3.3]{abramsky94}).  Recall that we discussed
in \S \ref{orgcfdf01e} how the theory was constructive in nature; we would expect to
encounter no major problems when repeating our reasoning internally to topoi
other than \(\Sets\). 

\section{A zoo of set-based optics}
\label{sec:orgbae8b4b}

\begin{center}
\begin{tabular}{ccc}
\hline
Name & Description (concrete form/monoidal action) & Proposition \\
\hline
Adapter            & \begin{tabular}{@{}c@{}}Identity functor \\ $(s \to a) \times (b \to t)$        \end{tabular} &  {\ref{prop:adapters}} \\[0.5cm]
Lens               & \begin{tabular}{@{}c@{}}Product \\ $(s \to a) \times (b \times s \to t)$        \end{tabular} &  {{\ref{prop:lenses}}} \\[0.5cm]
Prism              & \begin{tabular}{@{}c@{}}Coproduct \\ $(s \to t + a) \times (b \to t)$           \end{tabular} &  {{\ref{prop:prisms}}} \\[0.5cm]
Grate              & \begin{tabular}{@{}c@{}}Exponential \\ $((s \to a) \to b) \to t$                \end{tabular} &  {{\ref{prop:grates}}} \\[0.5cm]
Glass              & \begin{tabular}{@{}c@{}}Product and Exponential \\ $((s \to a) \to b) \to s \to t$ \end{tabular} &  {{\ref{prop:glasses}}} \\[0.5cm]
Affine Traversal   & \begin{tabular}{@{}c@{}}Product and Coproduct \\ $s \to t + a \times (b \to t)$ \end{tabular} &  {{\ref{prop:affine}}} \\[0.5cm]
Unsorted Traversal & \begin{tabular}{@{}c@{}}Combinatorial species \\ $s \to \sum\nolimits_n a^n/n! \times (b^n/n! \to t)$ \end{tabular} &  {{\ref{prop:unsortedtraversal}}} \\[0.5cm]
Traversal          & \begin{tabular}{@{}c@{}}Power series, Traversables \\ $s \to \Sigma n . a^n \times (b^n \to t)$ \end{tabular} &  {{\ref{prop:traversal}}} and {{\ref{prop:traversal2}}} \\[0.5cm]
Achromatic lens    & \begin{tabular}{@{}c@{}}Pointed product \\ $(s \to (b \to t) + 1) \times (s \to a) \times (b \to t)$ \end{tabular} &  {{\ref{prop:achromatic}}} \\[0.5cm]
Algebraic lens     & \begin{tabular}{@{}c@{}}Product by a monoid \\ $(s \to a) \times (\psi s \times b \to t)$ \end{tabular} &  {{\ref{prop:algebraiclens}}} \\[0.5cm]
Coalgebraic prism  & \begin{tabular}{@{}c@{}}Coproduct by a monoid \\ $(s \to \psi t + a) \times (b \to t)$ \end{tabular} &  {{\ref{prop:algebraiclens}}} \\[0.5cm]
Kaleidoscope       & \begin{tabular}{@{}c@{}}Applicative functors \\ $\prod\nolimits_n(a^n \to b) \to (s^n \to t)$ \end{tabular} &  {{\ref{prop:kaleidoscope}}} \\[0.5cm]
Setter             & \begin{tabular}{@{}c@{}}Any functor \\ $(a \to b) \times (s \to t)$ \end{tabular} &  {{\ref{prop:setters}}}\\
\end{tabular}
\end{center}

\newpage

\section*{Appendix: an alternative approach to traversables}
\label{sec:orgf00a6e9}
\label{appendixtraversables}

\begin{conjecture}
Let \(\mathbf{S} \to [\C,\C]\) be the monoidal action of applicative functors such that
the higher-order functor
\(\int_{F \in \mathbf{S}} \Ran_{F}F(-)\) is a comonad. Families of morphisms \(\int_{F \in \mathbf{S}} TF \tonat FT\)
satisfying linearity and unitarity correspond to comonad algebras
under the adjunction given by the right Kan extension.
\end{conjecture}

Assuming this result holds for the particular case where \(\mathbf{S}\) is
the category of traversable functors \(\mathbf{App}\), we can again derive the traversal as the optic
for traversables just by applying the construction of Proposition
\ref{orgb2fe82b} to that comonad and then the Lemma \ref{org8b21ea5}.

In any case, the only thing we need to prove Proposition \ref{prop:traversal2}
is the construction of cofree traversables.  We prove it here, to be able
to claim that result even after considering Remark \ref{remark:algebrafromlinearity}.
The main idea on that proof would be that the comonad defined there preserves ends.

\begin{proposition}
Let \(T \colon \Sets \to \Sets\) be a traversable functor and let \(H  \colon \Sets \to \Sets\)
be an arbitrary functor.  There exists an adjunction \(\mathbf{Trv}(T,KH) \cong [\Sets,\Sets]({\cal U}T,H)\),
where \({\cal U}\) is the forgetful functor.
\end{proposition}
\begin{proof}
We give \(KH\) traversable structure taking the morphisms \(n_F \colon KH \circ F \to F \circ KH\)
defined in Theorem \ref{orge517918}, which satisfies unitarity and linearity. 
Note that \(\trv \colon \int\nolimits_{F} TF \to FT\) determines a natural transformation \(\sigma \colon T \tonat KT\)
by Lemma \ref{org8b21ea5}.  This is a morphism of traversables \(\sigma \in \mathbf{Trv}(T, KT)\)
because the two sides of the
relevant commutative diagram have the same adjoint under the adjunction that
defines right Kan extensions, as the following diagram shows.
\begin{prooftree}
\AXC{\begin{tikzcd}[ampersand replacement=\&] TF \rar{\sigma} \& (\Ran_FFT) \circ F \rar{\varepsilon} \ar[bend left=15]{rr}{n_F} \& FT \rar{\sigma} \& F\circ (\Ran_FFT) \end{tikzcd}}
\UIC{\begin{tikzcd}[ampersand replacement=\&] T \rar{\sigma}  \& \Ran_FFT \rar{\id} \& \Ran_FFT \rar{\sigma} \& \Ran_FF(\Ran_FFT) \end{tikzcd}}
\UIC{\begin{tikzcd}[ampersand replacement=\&] TF \rar{\trv}\& FT \rar{\id} \& FT \rar{F\sigma} \& F (\Ran_FFT) \end{tikzcd}}
\end{prooftree}
Now, assume a morphism of traversables \(f \in \mathbf{Trv}(T, KH)\); that is, a natural
transformation between the underlying functors such that the following diagram
commutes for every \(F \in \mathbf{App}\).
\[\begin{tikzcd}
T \circ F \dar{f} \rar{\trv} & F \circ T \dar{Ff} \\
KH \circ F \rar{n_F} & F \circ KH
\end{tikzcd}\]
We want to show that there exists a unique natural transformation \(g \colon {\cal U}T \tonat H\) making
\(f\) factor uniquely as \(Kg \circ \sigma\).  Note that the previous diagram becomes,
under the adjunction, the following square, where \(\delta\) and \(\varepsilon\) were defined in
the comonad structure.
\[\begin{tikzcd}
T \dar{f} \rar{\sigma} & KT \dar{Kf} \\
KH \rar{\delta} & K^2H \lar[bend left=15]{K\varepsilon}
\end{tikzcd}\]
In the case where \(f = Kg \circ \sigma\), the fact that \(K\varepsilon \circ \delta = \id\)  implies \(K(\varepsilon \circ f) \circ \sigma = Kg \circ \sigma\).
Note that \(K(\varepsilon \circ f)\circ \sigma = f\). We will show that \(g = \varepsilon \circ f\) from \(K(\varepsilon \circ f) \circ \sigma = Kg \circ \sigma\),
proving uniqueness. In fact, again because of the adjunction determining right Kan extensions, we
have the following two adjunctions that must be equal for any \(F \in \mathbf{App}\).
\begin{prooftree}
\AXC{\begin{tikzcd}[ampersand replacement=\&] T \rar{\sigma} \& KT \rar[yshift=1ex]{Kg} \rar[yshift=-1ex,swap]{K(\varepsilon \circ f)} \& KH \end{tikzcd}}
\UIC{\begin{tikzcd}[ampersand replacement=\&] T\circ F \rar{\trv}  \& F\circ T \rar[yshift=1ex]{Fg} \rar[yshift=-1ex,swap]{F(\varepsilon \circ f)} \& F\circ KH \end{tikzcd}}
\end{prooftree}
In the particular case where \(F = \id\), we get \(\trv_{\id} = \id\) because unitarity and then
\(g = \varepsilon \circ f\).  This shows \(KH\) is the cofree traversal over \(H\).
\end{proof}

\bibliographystyle{alpha}
\bibliography{optics}

\newcommand{\etalchar}[1]{$^{#1}$}
\begin{thebibliography}{ASCG{\etalchar{+}}16}

\bibitem[AJ94]{abramsky94}
Samson Abramsky and Achim Jung.
\newblock Handbook of logic in computer science (vol. 3).
\newblock chapter Domain Theory, pages 1--168. Oxford University Press, Inc.,
  New York, NY, USA, 1994.

\bibitem[AMBL12]{adamek12}
Jir{\'\i} Ad{\'a}mek, Stefan Milius, Nathan Bowler, and Paul~B Levy.
\newblock Coproducts of monads on set.
\newblock In {\em Proceedings of the 2012 27th Annual IEEE/ACM Symposium on
  Logic in Computer Science}, pages 45--54. IEEE Computer Society, 2012.

\bibitem[ASCG{\etalchar{+}}16]{abou16}
Faris Abou-Saleh, James Cheney, Jeremy Gibbons, James McKinna, and Perdita
  Stevens.
\newblock Reflections on monadic lenses.
\newblock In {\em A List of Successes that can Change the World}, pages 1--31.
  Springer, 2016.

\bibitem[BG18]{boisseau18}
Guillaume Boisseau and Jeremy Gibbons.
\newblock What you needa know about {Y}oneda: Profunctor optics and the
  {Y}oneda {L}emma (functional pearl).
\newblock {\em {PACMPL}}, 2({ICFP}):84:1--84:27, 2018.

\bibitem[BLL97]{bergeron98}
Fran{\c{c}}ois Bergeron, Gilbert Labelle, and Pierre Leroux.
\newblock {\em Combinatorial species and tree-like structures}, volume~67 of
  {\em Encyclopedia of mathematics and its applications}.
\newblock Cambridge University Press, 1997.

\bibitem[BM98]{bird98}
Richard Bird and Lambert Meertens.
\newblock Nested datatypes.
\newblock In {\em International Conference on Mathematics of Program
  Construction}, pages 52--67. Springer, 1998.

\bibitem[Boi17]{boisseau17}
Guillaume Boisseau.
\newblock Understanding profunctor optics: a representation theorem.
\newblock Master's thesis, University of Oxford, 2017.

\bibitem[CH{\etalchar{+}}67]{cover67}
Thomas~M Cover, Peter Hart, et~al.
\newblock Nearest neighbor pattern classification.
\newblock {\em IEEE transactions on information theory}, 13(1):21--27, 1967.

\bibitem[CW01]{caccamo01}
Mario C{\'a}ccamo and Glynn Winskel.
\newblock A higher-order calculus for categories.
\newblock In {\em International Conference on Theorem Proving in Higher Order
  Logics}, pages 136--153. Springer, 2001.

\bibitem[Fis36]{fisher36}
Ronald~A Fisher.
\newblock The use of multiple measurements in taxonomic problems.
\newblock {\em Annals of eugenics}, 7(2):179--188, 1936.

\bibitem[FJ19]{fong19}
Brendan Fong and Michael Johnson.
\newblock Lenses and learners.
\newblock {\em CoRR}, abs/1903.03671, 2019.

\bibitem[GdSO09]{gibbons09}
Jeremy Gibbons and Bruno~C. d.~S.~Oliveira.
\newblock The essence of the {I}terator pattern.
\newblock {\em The {J}ournal of {F}unctional {P}rogramming}, 19(3-4):377--402,
  2009.

\bibitem[GHWZ18]{ghani18}
Neil Ghani, Jules Hedges, Viktor Winschel, and Philipp Zahn.
\newblock Compositional game theory.
\newblock In {\em Proceedings of the 33rd Annual {ACM/IEEE} Symposium on Logic
  in Computer Science, {LICS} 2018, Oxford, UK, July 09-12, 2018}, pages
  472--481, 2018.

\bibitem[Hed18]{hedges18l}
Jules Hedges.
\newblock Limits of bimorphic lenses.
\newblock {\em arXiv preprint arXiv:1808.05545}, 2018.

\bibitem[HFM{\etalchar{+}}19]{purescriptlens}
Thomas Honeyman, Phil Freeman, Brian Marick, Lukas Heidemann, Christoph
  Hegemann, and Liam Goodacre.
\newblock purescript-profunctor-lenses.
\newblock GitHub, Purescript-contrib,
  \url{https://github.com/purescript-contrib/purescript-profunctor-lenses},
  2019.

\bibitem[JC94]{jay94}
C~Barry Jay and J~Robin~B Cockett.
\newblock Shapely types and shape polymorphism.
\newblock In {\em European Symposium on Programming}, pages 302--316. Springer,
  1994.

\bibitem[JO15]{jaskelioff15}
Mauro Jaskelioff and Russell O'Connor.
\newblock A representation theorem for second-order functionals.
\newblock {\em The {J}ournal of {F}unctional {P}rogramming}, 25, 2015.

\bibitem[JR12]{rypacek12}
Mauro Jaskelioff and Ondrej Rypacek.
\newblock An investigation of the laws of traversals.
\newblock In {\em Proceedings Fourth Workshop on Mathematically Structured
  Functional Programming, MSFP@ETAPS 2012, Tallinn, Estonia, 25 March 2012.},
  pages 40--49, 2012.

\bibitem[Kel80]{kelly80}
G~Max Kelly.
\newblock A unified treatment of transfinite constructions for free algebras,
  free monoids, colimits, associated sheaves, and so on.
\newblock {\em Bulletin of the Australian Mathematical Society}, 22(1):1--83,
  1980.

\bibitem[Kme18]{kmett15}
Edward Kmett.
\newblock lens library, version 4.16.
\newblock Hackage \url{https://hackage.haskell.org/package/lens-4.16},
  2012--2018.

\bibitem[Koc09]{kock09}
Joachim Kock.
\newblock Notes on polynomial functors.
\newblock {\em Manuscript, version}, pages 08--05, 2009.

\bibitem[Lor15]{loregian15}
Fosco Loregian.
\newblock This is the (co)end, my only (co)friend.
\newblock {\em arXiv preprint arXiv:1501.02503}, 2015.

\bibitem[Mak96]{makkai96}
Michael Makkai.
\newblock Avoiding the axiom of choice in general category theory.
\newblock {\em Journal of Pure and Applied Algebra}, 108(2):109--173, 1996.

\bibitem[Mar14]{marsden14}
Dan Marsden.
\newblock Category theory using string diagrams.
\newblock {\em CoRR}, abs/1401.7220, 2014.

\bibitem[Mil17]{milewski17}
Bartosz Milewski.
\newblock Profunctor optics: the categorical view.
\newblock
  https://bartoszmilewski.com/2017/07/07/profunctor-optics-the-categorical-view/,
  2017.

\bibitem[ML71]{maclane71}
Saunders Mac~Lane.
\newblock {\em Categories for the working mathematician}, volume~5.
\newblock Springer Science \& Business Media, 1971.

\bibitem[MP08]{mcbride08}
Conor McBride and Ross Paterson.
\newblock Applicative programming with effects.
\newblock {\em The {J}ournal of {F}unctional {P}rogramming}, 18(1):1--13, 2008.

\bibitem[Mye]{david}
David~Jaz Myers.
\newblock Personal communication, 25th July 2019.

\bibitem[{nLa}18]{nlab}
{nLab authors}.
\newblock {{H}}ome{{P}}age.
\newblock \url{http://ncatlab.org/nlab/show/HomePage}, May 2018.
\newblock \href{http://ncatlab.org/nlab/revision/HomePage/262}{Revision 262}.

\bibitem[Nor]{norell08}
Ulf Norell.
\newblock {D}ependently {T}yped {P}rogramming in {A}gda.
\newblock In {\em {A}dvanced {F}unctional {P}rogramming, 6th {I}nternational
  {S}chool, {AFP} 2008, {H}eijen, {T}he {N}etherlands, {M}ay 2008, {R}evised
  {L}ectures}, pages 230--266.

\bibitem[Ole82]{oles82}
Frank~Joseph Oles.
\newblock {\em A Category-theoretic Approach to the Semantics of Programming
  Languages}.
\newblock PhD thesis, Syracuse, NY, USA, 1982.
\newblock AAI8301650.

\bibitem[PGW17]{pickering17}
Matthew Pickering, Jeremy Gibbons, and Nicolas Wu.
\newblock Profunctor optics: Modular data accessors.
\newblock {\em Programming Journal}, 1(2):7, 2017.

\bibitem[PS08]{pastro08}
Craig Pastro and Ross Street.
\newblock Doubles for monoidal categories.
\newblock {\em Theory and applications of categories}, 21(4):61--75, 2008.

\bibitem[Red95]{reddy95}
Uday~S Reddy.
\newblock Monads and algebras - {A}n introduction.
\newblock 1995.

\bibitem[Rie17]{riehl17}
Emily Riehl.
\newblock {\em Category theory in context}.
\newblock Courier Dover Publications, 2017.

\bibitem[Ril18]{riley18}
Mitchell Riley.
\newblock Categories of optics.
\newblock {\em arXiv preprint arXiv:1809.00738}, 2018.

\bibitem[RJ17]{rivas17}
Exequiel Rivas and Mauro Jaskelioff.
\newblock Notions of computation as monoids.
\newblock {\em Journal of functional programming}, 27, 2017.

\bibitem[Spi19]{spivak19}
David~I Spivak.
\newblock Generalized lens categories via functors $\mathcal{C}^{\mathrm{op}}
  \to \mathsf{Cat}$.
\newblock {\em arXiv preprint arXiv:1908.02202}, 2019.

\bibitem[Str72]{street72}
Ross Street.
\newblock The formal theory of monads.
\newblock {\em Journal of Pure and Applied Algebra}, 2(2):149--168, 1972.

\bibitem[Tam06]{tambara06}
Daisuke Tambara.
\newblock Distributors on a tensor category.
\newblock {\em Hokkaido mathematical journal}, 35(2):379--425, 2006.

\bibitem[TLu]{ninjayoneda}
{T}om~{L}einster (\url{https://mathoverflow.net/users/586/tom-leinster}).
\newblock Coend computation.
\newblock MathOverflow.
\newblock \url{https://mathoverflow.net/q/20451} (version: 2010-04-06).

\bibitem[TVD14]{troelstra14}
Anne~Sjerp Troelstra and Dirk Van~Dalen.
\newblock {\em Constructivism in mathematics}, volume~2.
\newblock Elsevier, 2014.

\bibitem[Yor14]{yorgey14}
Brent Yorgey.
\newblock {\em Combinatorial species and labelled structures}.
\newblock PhD thesis, University of Pennsylvania, 2014.

\end{thebibliography}
\end{document}